\documentclass[sigconf]{acmart}

\AtBeginDocument{%
  \providecommand\BibTeX{{%
    \normalfont B\kern-0.5em{\scshape i\kern-0.25em b}\kern-0.8em\TeX}}}

\copyrightyear{2023}
\acmYear{2023}
\setcopyright{acmlicensed}\acmConference[CCS '23]{Proceedings of the 2023 ACM SIGSAC Conference on Computer and Communications Security}{November 26--30, 2023}{Copenhagen, Denmark}
\acmBooktitle{Proceedings of the 2023 ACM SIGSAC Conference on Computer and Communications Security (CCS '23), November 26--30, 2023, Copenhagen, Denmark}
\acmPrice{15.00}
\acmDOI{10.1145/3576915.3623193}
\acmISBN{979-8-4007-0050-7/23/11}

\usepackage{nicefrac}
\usepackage{siunitx}
\usepackage{array,framed}
\usepackage{array}
\usepackage{booktabs} %
\usepackage{bigstrut}
\makeatletter
\newlength\mylena
\newlength\mylenb
\newcommand\mystrut[1][2]{%
    \setlength\mylena{#1\ht\@arstrutbox}%
    \setlength\mylenb{#1\dp\@arstrutbox}%
    \rule[\mylenb]{0pt}{\mylena}}
\makeatother

\usepackage{
  color,
  float,
  epsfig,
  wrapfig,
  graphics,
  graphicx,
  subcaption
}

\usepackage{amssymb}
\usepackage{textcomp}
\usepackage{setspace}
\usepackage{amsfonts}
\usepackage{latexsym,fancyhdr,url}
\usepackage{enumerate}
\usepackage[vlined,ruled, algo2e,linesnumbered]{algorithm2e}
\usepackage{algpseudocode}
\usepackage{hyperref}
\usepackage[capitalize,noabbrev]{cleveref}

\usepackage{graphics}
\usepackage{xparse} %
\usepackage{xspace}
\usepackage{multirow}
\usepackage{csvsimple}
\usepackage{balance}

\usepackage{
  tikz,
  pgfplots,
  pgfplotstable
}

\usetikzlibrary{
  shapes.geometric,
  arrows,
  external,
  pgfplots.groupplots,
  matrix
}

\pgfplotsset{compat=1.9}

\usepackage{mathtools,}

\DeclareMathAlphabet{\mathcal}{OMS}{cmsy}{m}{n}

\DeclareGraphicsExtensions{%
    .png,.PNG,%
    .pdf,.PDF,%
    .jpg,.mps,.jpeg,.jbig2,.jb2,.JPG,.JPEG,.JBIG2,.JB2}

\usepackage{natbib}

\usepackage{dblfloatfix}
\usepackage{url}
\usepackage{booktabs}       %
\usepackage{nicefrac}       %
\usepackage{microtype}      %
\usepackage{xcolor}         %
\usepackage[english]{babel}
\usepackage{dsfont}
\usepackage{graphicx,float}

\usepackage{adjustbox}
\usepackage{wrapfig}
\usepackage{multicol}

\usepackage{pgf}
\usepackage{collcell}
\usepackage{nccmath}
\usepackage{thmtools, thm-restate}
\usepackage{textcomp}
\usepackage{graphicx}
\usepackage{textcomp}
\usepackage{multirow}
\usepackage{siunitx}
\usepackage{chemformula}
\usepackage{collcell}
\usepackage{float}
\usepackage{supertabular}
\usepackage{makecell}
\usepackage{color}
\usepackage{adjustbox}
\usepackage{textcomp}
\usepackage{wrapfig,lipsum}
\usepackage{booktabs}
\usepackage{multirow}
\usepackage{makecell}
\newcommand{\Mcal}{\mathcal{M}}
\newcommand{\Aclass}{\mathbb{A}}
\newcommand{\Bclass}{\mathbb{B}}
\usepackage{pifont}%
\newcommand{\userdpfedavg}{\texttt{UserDP-FedAvg}}
\newcommand{\insdpfedavg}{\texttt{InsDP-FedAvg}}
\newcommand{\insdpfedsgd}{\texttt{InsDP-FedSGD}}

\newcommand{\inslevel}{instance-level}
\newcommand{\mnist}{MNIST}
\newcommand{\cifar}{CIFAR}
\newcommand{\sent}{Sent140}
\newcommand{\ceracc}{certified accuracy}

\newcommand{\cerk}{\mathsf{K}}

\definecolor{aliceblue}{rgb}{0.94, 0.97, 1.0}
\definecolor{lightgray}{gray}{0.9}
\usepackage[listings,skins,breakable]{tcolorbox}
\usepackage{amsmath,subcaption,booktabs,xcolor,amsthm,comment,graphics,comment,multirow,colortbl}

 \newcommand{\squishlist}{
	\begin{list}{$\bullet$}
		{
			\setlength{\itemsep}{0pt}
			\setlength{\parsep}{0.5pt}
			\setlength{\topsep}{0.5pt}
			\setlength{\partopsep}{0pt}
			\setlength{\leftmargin}{1em}
			\setlength{\labelwidth}{1em}
			\setlength{\labelsep}{0.5em} } }
	
\newcommand{\squishend}{
\end{list}  }
  \newcommand{\squishlistdash}{
	\begin{list}{$-$}
		{
			\setlength{\itemsep}{0pt}
			\setlength{\parsep}{0.5pt}
			\setlength{\topsep}{0.5pt}
			\setlength{\partopsep}{0pt}
			\setlength{\leftmargin}{1em}
			\setlength{\labelwidth}{1em}
			\setlength{\labelsep}{0.5em} } }
	
\newcommand{\squishenddash}{
\end{list}  }

\newcommand{\attackcost}[0]{attack inefficacy\xspace}

\usepackage{color}

\newcommand{\ApplyBlue}[1]{%
    \textcolor{black}{#1}
}
\newcolumntype{R}{>{\collectcell\ApplyBlue}{r}<{\endcollectcell}}

\definecolor{revcolor}{HTML}{0a46f4}  %
\newenvironment{rev}[1]
  {\color{black}}  %
  {\color{black}}

\newcommand{\revise}[1]{\textcolor{black}{#1}}

\usepackage{enumitem}
\theoremstyle{plain}
\newtheorem{theorem}{Theorem}
\newtheorem{lemma}{Lemma}

\newtheorem{example}{Example}

\newtheorem{definition}{Definition}

\begin{document}

\title[Unraveling the Connections between Privacy and Certified Robustness in FL Against Poisoning Attacks]{Unraveling the Connections between Privacy and Certified Robustness in Federated Learning Against Poisoning Attacks}

\author{Chulin Xie}
\email{chulinx2@illinois.edu}
\affiliation{%
  \institution{University of Illinois at
Urbana-Champaign}
\city{Urbana}
  \state{Illinois}
  \country{USA}
}

\author{Yunhui Long}
\email{ylong4@illinois.edu}
\affiliation{%
  \institution{University of Illinois at
Urbana-Champaign}
\city{Urbana}
  \state{Illinois}
  \country{USA}
}

\author{Pin-Yu Chen}
\email{pin-yu.chen@ibm.com}
\affiliation{%
  \institution{IBM Research}
  \state{New York}
  \country{USA}
}

\author{Qinbin Li}
\email{qinbin@berkeley.edu}
\affiliation{%
  \institution{UC Berkeley}
  \city{Berkeley}
  \state{California}
  \country{USA}
}

\author{Arash Nourian}
\email{nouriara@amazon.com}
\affiliation{%
  \institution{Amazon Web Services}
   \city{Santa Clara}
   \state{California}
   \country{USA}
}

\author{Sanmi Koyejo}
\email{sanmi@cs.stanford.edu}
\affiliation{%
  \institution{Stanford University}
  \city{Stanford}
  \state{California}
  \country{USA}
}

\author{Bo Li}
\email{lbo@illinois.edu}
\affiliation{%
  \institution{University of Illinois at
Urbana-Champaign}
\city{Urbana}
  \state{Illinois}
  \country{USA}
}

\begin{abstract}
Federated learning (FL)  provides an efficient paradigm to jointly train a global model leveraging  data  from distributed users.
As local training data comes from different users who may not be trustworthy, several studies have shown that FL is vulnerable to poisoning attacks. Meanwhile, to protect the privacy of local users, FL is usually trained in a differentially private way (DPFL). Thus, in this paper, we ask: \textit{What are the underlying connections between differential privacy and certified robustness in FL against poisoning attacks? Can we leverage the innate privacy property of  DPFL to provide certified robustness for FL? Can we further improve the privacy of FL to improve such robustness certification?}
We first investigate both user-level and instance-level privacy of FL and provide formal privacy analysis to achieve improved instance-level privacy.
We then provide two robustness certification criteria: \textit{certified prediction} and \textit{certified attack inefficacy} for DPFL on both user and instance levels. Theoretically,  we provide the certified robustness of DPFL based on both criteria given a bounded number of adversarial users or instances. 
Empirically, we conduct extensive experiments to verify our theories under a range of poisoning attacks on different datasets. We find that increasing the level of privacy protection in DPFL results in stronger \textit{certified attack inefficacy}; however, it does not necessarily lead to a stronger \textit{certified prediction}. Thus, achieving the optimal certified prediction requires a proper balance between privacy and utility loss.
\end{abstract}

\vspace{-5mm}


\vspace{-6mm}

\maketitle

\section{Introduction}
\label{sec:intro}
Federated Learning (FL), which aims to jointly train a global model with distributed local data,  has been widely deployed in different applications, such as finance~\cite{yang2019ffd} and medical analysis~\cite{brisimi2018federated}. However, the fact that the local data and the training process are entirely controlled by the \textit{local users}, who may be adversarial, raises great concerns from both security and privacy perspectives. 
In particular, recent studies show that FL is vulnerable to different types of training-time attacks,
such as model poisoning~\revise{\cite{shejwalkar2021manipulating,fang2020local,bhagoji2018analyzing}}, backdoor attacks~\cite{bagdasaryan2020backdoor,xie2019dba,wang2020attackthetails}, and label-flipping attacks~\cite{fung2020limitations}.

Several defenses have been proposed to defend against poisoning attacks in FL. 
For instance, various robust aggregation methods~\revise{\cite{blanchard2017machinekrum,el2018hidden,yin2018byzantine,pillutla2019robustrfa,nguyen2022flame}} identify and down-weight the malicious updates during aggregation, or estimate a true ``center'' of the received updates instead of taking a weighted average directly. Other defenses include robust FL protocols (e.g., clipping~\cite{sun2019can}, noisy perturbation~\cite{sun2019can}, and additional evaluation during training~\cite{xie2022zenops}) and post-training strategies (e.g., fine-tuning and pruning~\cite{wu2020mitigating}) that repair the poisoned global model. 
However, as these works mainly focus on providing empirical robustness on specific types of  attacks, they have been shown to be vulnerable to newly proposed strong adaptive attacks~\cite{wang2020attackthetails,xie2019dba,fang2020local}.
Recently, some \textit{certified} defenses have been proposed against poisoning attacks~\cite{weber2020rab,rosenfeld2020certifiedlableflip,jia2021intrinsic, jia2022certified, levine2020deep}, while they mainly focus on  centralized setting.

In the meantime, privacy concerns have motivated FL training, where the sensitive raw data is kept on local devices without sharing.
However, sharing other indirect information such as gradients or model updates during the FL training process can also leak sensitive user information~\cite{NEURIPS2019_deepleakage}. 
As a result, approaches based on differential privacy (DP)~\cite{dwork2014algorithmic}, homomorphic encryption~\cite{rouhani2018deepsecure}, and secure multiparty computation~\cite{ben1988completeness,bonawitz2017practical} have been proposed to protect the privacy of users in FL. \revise{In particular, differentially private federated learning (DPFL)~\cite{geyer2017differentially, mcmahan2018learning,naserilocal} provides strong privacy guarantees for user privacy, and has been deployed to real-world FL applications such as Google’s Gboard~\cite{googledpfl} and Apple’s Siri~\cite{appledpfl}.}

Recent studies observe that differential privacy (DP) is related to the robustness of ML models.
\revise{Intuitively, DP is designed to protect the privacy of individual data, such that the output of an algorithm should not change much when one individual record is modified.} 
Hence, the prediction of a DP model will be less impacted by a small amount of perturbation.
Consequently, several studies have been conducted to provide empirical and certified defenses against evasion attacks~\cite{Lecuyer2019,wang2021certified,liu2021certifiably} and data poisoning attacks~\cite{ma2019data,hong2020effectiveness} based on DP properties in the \textit{centralized} ML setting.
Empirical defense against backdoor attacks~\cite{gu2019badnets} based on DP has also been studied in \textit{federated learning} without theoretical guarantees~\cite{bagdasaryan2020backdoor,sun2019can,naserilocal}. 
To the best of our knowledge, despite the widespread use of DP in FL, there is no study exploring the underlying connections between DP and \textit{certified} robustness in FL against poisoning attacks, or providing {certified} robustness for DPFL leveraging its privacy properties.

Hence, in this paper, we aim to bridge this gap and answer the research questions: {
Can we {quantitatively} uncover the underlying connections between differential privacy and  the certified robustness of FL against poisoning attacks? 
Can we improve the privacy of FL to improve its certified robustness?}

To explore and exploit the inherent privacy properties of DPFL for robustness certifications of FL,
we mainly focus on two \underline{goals}: (1) conducting thorough privacy analysis of DPFL algorithms over multiple rounds of training; (2) providing certified robustness of DPFL as a function of its privacy parameters $(\epsilon,\delta)$ under different robustness criteria. 
In terms of privacy analysis, we revisit existing DPFL algorithms and provide improved privacy analysis.
We investigate user-level DP, which is commonly guaranteed in cross-device FL to protect the sensitive information of each user~\cite{agarwal2018cpsgd,geyer2017differentially,mcmahan2018learning,asoodeh2020differentially,liang2020dpfllaplacian}, as well as instance-level DP which is more suitable for cross-silo FL  to protect sensitive information in each data instance~\cite{malekzadeh2021dopamine,zhu2021votingbased,liu2022privacy}. 
Moreover, we carry out privacy analysis for instance-level DPFL algorithms, and provide an improved guarantee for FedSGD~\cite{mcmahan2016communication}-based algorithm with privacy amplification of user and batch subsampling. We also provide a formal privacy guarantee for FedAvg~\cite{mcmahan2016communication}-based algorithm with parallel composition~\cite{mcsherry2009privacy} considering local privacy budget accumulation and global privacy budget aggregation over training rounds.
In terms of certified robustness of FL, we introduce two robustness criteria: \textit{certified prediction} and \textit{certified \attackcost}, which can be adapted to different threat models in DPFL.
We prove that user-level (instance-level) DPFL is certifiably robust against a bounded number of adversarial users (instances).  
We also show that our analysis on certified robustness is \textit{agnostic} to the type of poisoning attack strategies as long as the number of adversarial users or instances is bounded.
Empirically, we quantitatively measure the relationship between privacy guarantee and the certified robustness of FL based on different robustness criteria. We present the first set of  certified robustness for DPFL on image datasets \mnist{}, \cifar{} and text dataset Tweets against various FL poisoning attacks, including backdoor attacks~\cite{bagdasaryan2020backdoor,sun2019can}, distributed backdoor attacks~\cite{xie2019dba}, label-flipping attacks~\cite{fung2020limitations}, model replacement attacks~\cite{bagdasaryan2020backdoor, bhagoji2018analyzing}, \revise{and optimization-based model poisoning attacks~\cite{shejwalkar2021manipulating}}. 
From our theoretical and empirical results, we provide the following insights:
\begin{enumerate}[label=(\emph{\arabic*}),leftmargin=*]
    \vspace{-0.3em}
\item Certified robustness in terms of \textit{certified prediction} is influenced by both the privacy guarantee and model utility. Moderately strong privacy protection enhances certified prediction, while overly strong privacy protection can harm. This is potentially caused by the significant loss of model utility. 
Thus,  optimal certified prediction is achieved by balancing privacy protection and utility.
\item \textit{Certified attack inefficacy} is always enhanced by stronger privacy protection. The certified lower bounds of \attackcost are generally tight when the number of poisoned users or instances is small, or the attack strategy is strong. 
\item Different DPFL algorithms yield varying certification robustness under the same privacy guarantee due to distinct training mechanisms (e.g., per-layer clipping or flat clipping).
\item Larger FL data heterogeneity leads to a smaller number of tolerable adversaries for certified prediction, due to degraded utility.
    \vspace{-1em}
\end{enumerate}
\textbf{\underline{Contributions.}} 
In this paper, we 
take the first step to characterize the underlying connections between privacy guarantees and certified robustness in FL. 
We hope our work can pave the way for more private and robust FL applications.
\begin{itemize}[leftmargin=*]
    \vspace{-0.3em}
\item We provide two criteria for certified robustness of FL against poisoning attacks {(Section~\ref{sec:cer_robust_user_level})}.
\item Given an FL model satisfying user-level DP, we prove that it is certifiably robust against arbitrary poisoning attacks with a bounded number of adversarial users   {(Section~\ref{sec:cer_robust_user_level})}.
\item We revisit two instance-level DPFL algorithms and  provide the improved privacy analysis  (Section~\ref{sec:ins_privacy}). We further prove that instance-level DPFL is certifiably robust against a bounded number of poisoning instances during training   {(Section~\ref{sec:ins_robustness})}.
\item We systematically evaluate the \textit{certified} robustness for user-level and instance-level DPFL based on two robustness criteria on both image and text datasets against \revise{five} types of poisoning attacks. We provide a series of ablation studies to further analyze the factors that affect the certified robustness, such as different DPFL algorithms and data heterogeneity. \revise{Our results also indicate that our certification approach offers strong \textit{empirical} robustness when compared to six empirical FL defenses (Section~\ref{sec_exp}).}
\end{itemize}

\section{Related work}

\subsection{Differentially Private Federated Learning}
To guarantee \textit{user-level privacy} for FL, \citet{mcmahan2018learning} introduce user-level DP-FedAvg and DP-FedSGD to train language models with millions of users, where the server clips the norm of each local update, then adds Gaussian noise on the summed update. User-level DP-FedAvg is also proposed independently by \citet{geyer2017differentially}. Both of these works calculate the privacy budget via the moment accountant~\cite{abadi2016deep}.
In CpSGD~\cite{agarwal2018cpsgd}, each user clips and quantizes the model update, and adds noise drawn from Binomial distribution,  achieving both communication efficiency and DP.
\citet{bhowmick2018protection} derive DP  for FL via R{\'e}nyi divergence~\cite{mironov2017renyi} and study its resilience against data reconstruction attacks. 
\citet{liang2020dpfllaplacian} utilize Laplacian smoothing for each local update to enhance model utility. 
\citet{asoodeh2020differentially} propose a different way to calculate the privacy budget by interpreting each round as a Markov kernel and quantifying its impact on privacy parameters.
Recent studies propose different regularization and sparsification techniques to improve utility~\cite{cheng2022differentially} and leverage sharpness-aware optimizer~\cite{foret2021sharpnessaware} to make the model less sensitive to weight perturbation~\cite{shi2023make}.

In terms of \textit{instance-level privacy} for FL,
Dopamine~\cite{malekzadeh2021dopamine} provides \inslevel{} privacy guarantee for FedSGD~\cite{mcmahan2016communication} where each user only performs one step of DP-SGD~\cite{abadi2016deep} at each FL round. \citet{girgis2021shuffled} introduce variants of instance-level DP-FedSGD with a trusted shuffler between the server and users to randomly permutes user gradients for privacy amplification through anonymization. Nonetheless, both works cannot be applied to the more general setting (e.g., FedAvg~\cite{mcmahan2016communication}) where each user performs multiple steps of SGD.
\citet{zhu2021votingbased} privately aggregate the label predictions from  users in a voting scheme and provide DP guarantees on both user and instance levels. However, it does not allow aggregating the gradients or updates and is thus not applicable to standard FL.
Recent works 
combine local DP-SGD training of clients with personalized FL algorithms~\cite{liu2021projected,noble2022differentially,liu2022privacy,yangprivatefl} to address the user heterogeneity issue in FL and improve privacy-utility tradeoff. 

In summary, the above works focus on privacy in FL while leaving its robustness unexplored. Our goal is to uncover the underlying connections between privacy guarantees with certified robustness.

\begin{table}[t]
\centering
\caption{Comparison between our work and existing studies on privacy and robustness in the context of poisoning attacks.}
 \vspace{-3mm}
\label{tab:compare_poisoning_relatedwork}
\scalebox{0.8}{
\begin{tabular}{lcccccc}
\toprule
 & \makecell{FL} & \makecell{DP}  & \makecell{Empirical \\ Robustness} & \makecell{Certifed \\ Prediction}  & \makecell{Certifed \\ Attack Inefficacy}  \\ 
\midrule
\cite{rosenfeld2020certifiedlableflip, weber2020rab,levine2020deep,wang2022improved} &   \textcolor{red}{$\times$}  & \textcolor{red}{$\times$} & \textcolor{green}{\checkmark}   &  \textcolor{green}{\checkmark} &  \textcolor{red}{$\times$}  \\
\cite{hong2020effectiveness} &  \textcolor{red}{$\times$} & \textcolor{green}{\checkmark} &  \textcolor{green}{\checkmark} & \textcolor{red}{$\times$} & \textcolor{red}{$\times$}  \\ 
\cite{ma2019data} &  \textcolor{red}{$\times$} & \textcolor{green}{\checkmark}  &  \textcolor{green}{\checkmark} &  \textcolor{red}{$\times$}  & \textcolor{green}{\checkmark}    \\ \midrule
\cite{xie2021crfl,cao2021provably} &   \textcolor{green}{\checkmark} & \textcolor{red}{$\times$} & \textcolor{green}{\checkmark} & \textcolor{green}{\checkmark} & \textcolor{red}{$\times$}      \\
\cite{bagdasaryan2020backdoor, wang2020attackthetails,sun2019can,naserilocal}  &   \textcolor{green}{\checkmark} &  \textcolor{green}{\checkmark} & \textcolor{green}{\checkmark} &  \textcolor{red}{$\times$} & \textcolor{red}{$\times$}   \\
\rowcolor{aliceblue} \textbf{Our work}    & \textcolor{green}{\checkmark} &  \textcolor{green}{\checkmark} &  \textcolor{green}{\checkmark}   & \textcolor{green}{\checkmark} & \textcolor{green}{\checkmark}  \\
\bottomrule
\end{tabular}
}
 \vspace{-6.5mm}
\end{table}

\looseness=-1
\vspace{-3mm}

\subsection{Certified Robustness against Evasion Attacks}

Machine learning models are susceptible to test-time evasion attacks~\cite{goodfellow2014explaining}, and different defenses have been proposed to enhance the robustness of models and provide certifications to guarantee consistent predictions under a specified perturbation radius~\cite{li2022sok}.
Pixel-DP~\cite{Lecuyer2019} first connects DP  to certified robustness against adversarial examples by adding noise on the test sample $O$ times and taking the expectation over the corresponding outputs.
Later on, \textit{randomized smoothing}~\cite{cohen2019certified} is proposed to provide a tight robustness certification.
\citet{wang2021certified} extends Pixel-DP~\cite{Lecuyer2019} to NLP tasks, and \citet{liu2021certifiably} improves the certification based on R{\'e}nyi  DP~\cite{mironov2017renyi}. 
However, such an approach of adding noise to test samples does not guarantee that the training algorithm itself satisfies DP. In contrast, our certification against poisoning attacks focuses on DPFL, which requires the \textit{training algorithm} to satisfy DP. Such analysis requires careful privacy budget analysis of DPFL models across multiple training rounds and aggregation.

\vspace{-3mm}
\subsection{Certified Robustness against Poisoning Attacks}\label{subsec:related_work_robust_poisoning}
Compared to test-time certifications against evasion attacks,  training-time certifications against poisoning attacks 
have been less explored due to the notably different threat models and the complexity of analyzing model training dynamics, even in a centralized setting.

{In \textbf{centralized setting}, current approaches primarily utilize \textit{randomized smoothing} to certify the model robustness under a bounded number of poisoned instances.} \citet{weber2020rab} and \citet{rosenfeld2020certifiedlableflip} propose to add noise directly to the training dataset, train multiple models on the randomized datasets, and take majority vote for the final prediction for certification. 
\citet{levine2020deep} and \citet{wang2022improved} propose to
partition a centralized dataset into disjoint subsets, train an independent model on each partition, and make  majority predictions among all models. 
However, these certifications do not apply to FL, where each local model can influence other users’ local models through periodic global model aggregation, so the malicious effect of one poisoned local model could spread to all local models, making the certified robustness in FL a far more challenging task. 
To achieve certified robustness in \textbf{FL},
CRFL~\cite{xie2021crfl} clips the aggregated FL model parameters and adds noise, but it does not consider the properties provided by DPFL.
Emsemble~\cite{cao2021provably} trains numerous FL global models (e.g., 500) on different subsets of users and takes majority prediction.  Similarly, it only leverages the randomness in user-subsampling and does not consider data privacy property during training. 
Our goal is to explore the underlying connections between DP properties of DPFL algorithms and their certified robustness, as well as provide recipes for achieving higher certified robustness.

Several studies have explored the robustness against poisoning attacks induced by \textbf{DP}, either in centralized learning or only empirically in FL.  \citet{ma2019data} first demonstrate that private learners are resistant to data poisoning for centralized regression models and analyze the lower bound of \attackcost. 
\revise{Here 
we extend such a lower bound of attack inefficacy from DP in centralized setting~\cite{ma2019data} to  user-level DP in FL, and further derive the upper bound of the attack inefficacy.}
\revise{We also provide certified prediction guarantees as another robustness certification criterion for general classification tasks in FL based on the privacy properties.}
Meanwhile, some \textit{empirical} studies \cite{bagdasaryan2020backdoor,sun2019can,hong2020effectiveness,naserilocal}  show that DP property can mitigate backdoor attacks. 
For instance, in the \textbf{centralized setting}, \citet{hong2020effectiveness} show that the off-the-shelf mechanism DP-SGD~\cite{abadi2016deep} can serve as a defense against poisoning attacks;
in \textbf{FL},~\cite{wang2020attackthetails,bagdasaryan2020backdoor,sun2019can} show that bounding the norm and adding Gaussian noise on model updates can mitigate backdoor attacks. Recently, \citet{naserilocal} revealed that both user-level DP and instance-level DP can defend against backdoor attacks empirically with varying levels of privacy protection. 
\revise{However, none of these studies provides certified robustness guarantees for DPFL or characterizes the quantitative relationships between privacy guarantees and certified robustness in FL.} \revise{In contrast, our work offers robustness certifications, which can be represented as a function of DP parameters $(\epsilon,\delta)$ based on different robustness criteria. }
We provide an overall comparison between our work and existing studies in \cref{tab:compare_poisoning_relatedwork}.

\section{Preliminaries}
We start by providing some background on Differential Privacy (DP) and Federated Learning (FL). %

\paragraph{Differential Privacy.}
DP provides a mathematically rigorous guarantee for privacy, which ensures that the output of a random algorithm is close no matter whether an individual data record is included in the input. 

\vspace{-1mm}
\begin{definition}[$(\epsilon,\delta)$-DP \cite{dwork2006our}] \label{def:dp}
A randomized mechanism $\mathcal{M}:\mathcal{D} \rightarrow \Theta$ with domain $\mathcal{D}$ and output set $\Theta$
satisfies $(\epsilon, \delta)$-DP if for any pair of two adjacent datasets $ d,d' \in \mathcal{D}$, and for any possible (measurable) output set $E 	\subseteq \Theta$, it holds that 
 \vspace{-1mm}
\begin{align}
	\operatorname{Pr}[\mathcal{M}(d) \in E] \leq e^{\epsilon} \operatorname{Pr}\left[\mathcal{M}\left(d^{\prime}\right) \in E\right]+\delta.
\end{align}
\end{definition}

Group DP follows immediately \cref{def:dp}, where the privacy guarantee decreases with the size of the group. 
\vspace{-1mm}
\begin{restatable}[Group DP]{definition}{lemmagroupdp}
\label{lemma_group_dp}
For mechanism $\mathcal{M}$ that satisfies $(\epsilon, \delta)$-DP, it satisfies $(k\epsilon, \frac{1-e^{k \epsilon }}{1-e^\epsilon}\delta)$-DP for groups of size $k$. That is, for any  $ d,d' \in \mathcal{D}$ that differ by $k$ individuals and any $E \subseteq \Theta$, it holds that
 \vspace{-1mm}
\begin{align}\label{eq:group_dp_definition}
	\operatorname{Pr}[\mathcal{M}(d) \in E] \leq e^{k\epsilon} \operatorname{Pr}\left[\mathcal{M}\left(d^{\prime}\right) \in E\right]+ \frac{1-e^{k \epsilon }}{1-e^\epsilon} \delta.
\end{align}
\end{restatable}

\paragraph{Federated Learning.}
The standard instantiation of FL is FedAvg~\cite{mcmahan2016communication}, which trains a shared global model in FL without directly accessing the local training data of users.
{We consider an FL system consisting of $N$ users, with $B$ representing the set of all users (i.e., $B:=[N]$) and $D:=\{D_1,\ldots,D_N\}$ denoting the union of local datasets across all users. }
At round $t$, the server sends the current global model $w_{t-1}$ to users in the selected user set $U_t$, where $|U_t|=m=qN$ and $q$ is the user sampling probability. 
Each selected user $i \in U_t$ then locally updates the model for $E$ local epochs with its dataset $D_i$ and learning rate $\eta$ to obtain a new local model. Then, the user sends the local model updates $\Delta w_{t}^i$ to the server. Finally, the server aggregates over the updates from all selected users into the new global model: $w_{t}= w_{t-1}+ \frac{1}{m} \sum_{i\in U_t} \Delta w_t^i$.

\section{User-level DP and Certified Robustness}
\label{sec_client_privacy}

\subsection{User-level DP and Background}

\cref{def:dp} leaves the definition of adjacent datasets flexible, which is application-dependent.
When DP is used for the privacy protection of individual users,
the adjacency relation is defined 
 as that differing by data from one user~\cite{mcmahan2018learning}. 

\vspace{-1mm}
\begin{definition}[User-level $(\epsilon,\delta)$-DP]
\label{def:userdp}
Let $B,B'$ be two user sets. Let $D$ and $D'$ be the datasets that are the union of local training examples from all users in $B$ and $B'$, respectively. Then, $D$ and $D'$ are adjacent if $B$ and $B'$ differ by one user.
The mechanism $\mathcal{M}$ satisfies user-level $(\epsilon, \delta)$-DP if it meets  \cref{def:dp} with  $D$ and $D'$ as adjacent datasets.
\end{definition}

Following the standard  user-level DPFL  ~\cite{geyer2017differentially,mcmahan2018learning}, we introduce \userdpfedavg{} (\cref{algo:Fedavg_client_privacy} in~\cref{sec:dpfl_algo}).
 Specifically, at each round, the server first clips the model update from each user with a threshold $S$ such that its $\ell_2$-sensitivity is upper bounded by $S$. Next, the server sums up the updates, adds Gaussian noise sampled from 
$\mathcal{N}(0, \sigma^{2} S^{2})$, and takes the average:  
 \vspace{-2mm}
\begin{align}
w_{t} \gets w_{t-1 } + \frac{1}{m}\left(\sum_{i\in U_t}  \mathrm{Clip}({\Delta w^i_{t} ,S})  +  \mathcal{N}\left(0, \sigma^{2} S^{2}\right) \right). 
\end{align}
During FL training, the users repeatedly query private datasets over training rounds; thus, the privacy guarantee composes.
We use the existing accountant~\cite{wang2019subsampled} based on R\'enyi Differential Privacy (RDP)~\cite{mironov2017renyi} for a tight privacy budget accumulation over $T$ rounds.

\subsection{Certified Robustness of User-level DPFL}
\label{sec:cer_robust_user_level}

\subsubsection{Threat Model.}\label{subsubsec:user_level_threat_model}
We consider there are $k$ adversarial users (attackers) out of $N$ users.
\squishlist
\item \textbf{Attack Goal}:
The {goal} of attackers is to fool the trained FL global model on the server side with specific attack objectives (e.g., misclassification).
\item  \textbf{Attack Capability}:
In line with prior works~\cite{sun2019can,naserilocal}, for attacker {capability}, we consider the attacker with full control of its local training data/model. The attacker can arbitrarily manipulate the features and labels of the local data and modify the weights of the local model before submitting it to the server.
However, the attacker has no control over the server operations nor over the local training process of other users. 
The trusted server conducts DP-related operations\revise{~\cite{geyer2017differentially}}, including model update clipping and noise perturbing, so that the trained FL global model satisfies user-level DP even in the presence of attackers.
\item   \textbf{Attack Strategy}:  The attacker {strategies} include backdoor attacks~\cite{gu2019badnets,chen2017targeted},  which alter local data to embed a backdoor trigger with a targeted adversarial label during local training, causing the FL global model to misclassify any test data with the backdoor trigger as the target label~\cite{bagdasaryan2020backdoor,xie2019dba,sun2019can,wang2020attackthetails}; label flipping attacks~\cite{biggio2012poisoning,huang2011adversarial} which switch the labels of local training data  from one source class to a target class while keeping the data features unchanged, causing the FL global model to misclassify any test data from source class to target class~\cite{fu2019attack}; and model poisoning attacks that directly manipulate local model weights to tamper global model convergence~\cite{fang2020local} or  \revise{amplify the malicious effects of the attacker's model updates derived from poisoning data by scaing the updates by a factor of $\gamma$~\cite{bagdasaryan2020backdoor,bhagoji2018analyzing}}. \revise{Note that by providing certified robustness for FL, which is agnostic to the actual attack strategies, our work is able to explore the worst-case robustness of FL and its relationship to privacy properties.}
\squishend
We denote $B'$ as the set of all users, among which $k$ users are adversarial, and $D':=\{{D'}_1,\ldots, {D'}_{k-1}, {D'}_{k}, {D}_{k+1},\ldots,D_N \}$ as the corresponding union of local datasets.

Next, we introduce two criteria for robustness certification in FL: \textit{certified prediction} and \textit{certified \attackcost}.

\subsubsection{Certified Prediction}\label{subsubsec:method_cer_pred_userdp}

Consider the classification task with $C$ classes.
We define the classification \textit{scoring function} $f:(
\Theta, \mathbb{R}^d)  \rightarrow \Upsilon^C$ which maps model parameters $\theta \in \Theta$ and an input data $x \in \mathbb{R}^d$ to a confidence vector $f(\theta,x)$, and $f_c(\theta,x)\in [0,1]$ represents the confidence of class $c$. We mainly focus on the confidence after normalization, i.e., $f(\theta,x) \in \Upsilon^C = \{p \in \mathbb{R}^C_{\geq0} :  \|p\|_1=1\} $  in the probability simplex. 
Since the DP mechanism $\mathcal{M}$ is randomized and produces a \textit{stochastic} FL global model $\theta = \mathcal{M}(D)$, 
it is natural to resort to a probabilistic expression as a bridge for quantitative robustness certifications. In particular, we will use the expectation of the model's predictions to provide a quantitative guarantee on the robustness of $\mathcal{M}$. Concretely,
we define the \textit{expected scoring function} $F:(\theta , \mathbb{R}^d)  \rightarrow \Upsilon^C$  where 
$F_{c}(\mathcal{M}(D), x)= \mathbb{E}[f_{c}(\mathcal{M}(D),x)]$ is the expected confidence for class $c$. The expectation is taken over DP training randomness, e.g., random Gaussian noise and random user subsampling.
The corresponding \textit{prediction} 
 $H: (\theta , \mathbb{R}^d)  \rightarrow [C] $ is defined by 
 \vspace{-2mm}
\begin{align}  \label{eq:prediction_certification}
H(\mathcal{M}(D),x) :=\arg \max_{c \in [C]} F_c(\mathcal{M}(D),x),
\end{align}
 which is the top-one class based on expected prediction confidence.
We prove that such prediction allows robustness certification.

\textbf{Certified Prediction under One Adversarial User.}
Following our threat model above
and  the DPFL training mechanism in \cref{algo:Fedavg_client_privacy}, we denote the trained global model exposed to a poisoned dataset $D'$ as $\mathcal{M}(D')$. 
When the number of adversarial users $k=1$, 
$D$ and $D'$ are user-level adjacent datasets according to \cref{def:userdp}. 
Given that mechanism $\mathcal{M}$ satisfies user-level $(\epsilon,\delta)$-DP, based on the DP property, the distribution of the stochastic model $\mathcal{M}(D')$ is ``close'' to the distribution of $\mathcal{M}(D)$. 
Intuitively, according to the \textit{post-processing property} of DP~\cite{dwork2006our},
during testing, given a test sample $x$, 
we would expect the values of the expected confidence for each class $c$, i.e., $F_{c}(\mathcal{M}(D'), x)$ and $F_{c}(\mathcal{M}(D), x)$, to be  close, and hence the returned most likely class to be the \textit{same}, i.e., $H(\mathcal{M}(D),x) = H(\mathcal{M}(D'),x)$, indicating \textit{robust} prediction.

\begin{restatable}[Certified Prediction under One Adversarial User]{theorem}{thmonecertpred}
\label{thm_pred_consist_one_client}
\sloppy
Suppose a randomized mechanism $\Mcal$ satisfies user-level $(\epsilon, \delta)$-DP. For two user sets $B$ and $B^\prime$ that differ by one user, let $D$ and $D'$ be the corresponding training datasets. 
For a test input $x$, suppose $\Aclass, \Bclass \in [C]$ satisfy $ \Aclass=\arg \max_{c \in [C]} F_c(\mathcal{M}(D),x)$ and $\Bclass =\arg \max_{c \in [C]:c\neq \Aclass}  F_c(\mathcal{M}(D),x)$.
Then, it is guaranteed that 
 $  H(\mathcal{M}(D'),x) = H(\mathcal{M}(D),x) = \Aclass$ if:

\begin{align} \label{pred_cons_condition}
	 {F_{\Aclass}(\mathcal{M}(D), x)  >  e^{2\epsilon} F_{\Bclass}(\mathcal{M}(D), x)+  (1+ e^{\epsilon})\delta,} 
\end{align}
\end{restatable}
\vspace{-2mm}
\begin{proof}[Proof sketch]
\sloppy
The proof generalizes the analysis of {pixel-level DP in test-time}~\cite{Lecuyer2019}.
Specifically, with DP property for two FL neighboring datasets, we can lower bound  $	F_{\Aclass}(\mathcal{M}(D'), x)$ based on  $	F_{\Aclass}(\mathcal{M}(D), x)$, and upper bound $F_{\Bclass}(\mathcal{M}(D'), x)$ based on $F_{\Bclass}(\mathcal{M}(D), x)$. 
When the lower-bound of $	F_{\Aclass}(\mathcal{M}(D'), x)$ is strictly higher than the upper-bound of $F_{\Bclass}(\mathcal{M}(D'), x)$, the predicted class will be provably $\Aclass$ even under poisoning attack.
\cref{pred_cons_condition} states the condition for achieving such robustness. 
Full proofs are in \cref{sec:app_proofs}.    
\end{proof}

\textit{Remark.}
In \cref{thm_pred_consist_one_client}, if $\epsilon$ is large, i.e., weak privacy guarantee, such that the RHS of \cref{pred_cons_condition} $>1$, the robustness condition cannot hold since the expected confidence $F_{\Aclass}(\mathcal{M}(D), x) \in [0,1]$. On the other hand, to achieve small $\epsilon$, i.e., strong privacy guarantee,  large noise is required during training, which would hurt model utility and thus result in a small confidence margin between the top two classes (e.g., $F_{\Aclass}(\mathcal{M}(D), x)$ and $F_{\Bclass}(\mathcal{M}(D), x)$), making it hard to meet the robustness condition. 
This indicates that achieving certified prediction requires a reasonable privacy level $\epsilon$.

\textbf{Certified Prediction under $k$ Adversarial Users.}
When the number of adversarial users  $k>1$, we resort to group DP. 
According to \cref{lemma_group_dp}, given mechanism $\mathcal{M}$ satisfying user-level $(\epsilon,\delta)$-DP, it also satisfies user-level $(k\epsilon, \frac{1-e^{k \epsilon }}{1-e^\epsilon}\delta)$-DP for groups of size $k$.
When $k$ is smaller than a certain threshold,  leveraging the group DP property, we would expect that the distribution of the stochastic model $\mathcal{M}(D')$ is not too far  away from the distribution of $\mathcal{M}(D)$ such that they would make the close prediction for a test sample with high probability.
Next, we present the corresponding robustness certificate by studying the sufficient condition of $k$, such that the prediction for a test sample is consistent between the stochastic FL models trained from $D$ and $D'$ separately.

\vspace{-1mm}
\begin{restatable}[Upper Bound of $k$ for Certified Prediction]{theorem}{thmkcertpred}
\label{thm_pred_consist_k_client}
Suppose a randomized mechanism $\Mcal$ satisfies user-level $(\epsilon, \delta)$-DP. For two user sets $B$ and $B^\prime$ that differ by $k$ users, let $D$ and $D'$ be the corresponding training datasets. 
For a test input $x$, suppose $\Aclass, \Bclass \in [C]$ satisfy $ \Aclass=\arg \max_{c \in [C]} F_c(\mathcal{M}(D),x)$ and $\Bclass=\arg \max_{c \in [C]:c\neq \Aclass} F_c(\mathcal{M}(D),x)$,
then $H(\mathcal{M}(D'),x) = H(\mathcal{M}(D),x) = \Aclass$, $\forall k<\cerk$ where $\cerk$ is the certified number of adversarial users:

\begin{align} \label{eq:cerk}
{\cerk=   \frac{1}{2\epsilon} \log \frac{ F_{\Aclass}(\mathcal{M}(D),x) (e^\epsilon -1) + \delta }{F_{\Bclass}(\mathcal{M}(D),x) (e^\epsilon -1) + \delta }}
\end{align}

\end{restatable}

\begin{proof}[Proof sketch]
\sloppy
By solving \cref{thm_pred_consist_one_client} combined with Group DP definition, 
 we derive the above robustness condition.
Full proofs are in~\cref{sec:app_proofs}. 
\end{proof}

\vspace{-1mm}
\textit{Remark.}
\sloppy 
\textbf{(1)} In \cref{thm_pred_consist_k_client}, if we fix $F_{\Aclass}(\mathcal{M}(D), x)$ and $F_{\Bclass}(\mathcal{M}(D), x)$, the smaller $\epsilon$ of FL can certify larger $\cerk$. However, smaller $\epsilon$ also induces lower confidence due to the model performance drop, thus reducing the tolerable $\cerk$ instead.
 As a result, properly choosing  $\epsilon$ would help to improve the certified robustness and tolerate more adversaries during training (e.g. certify against a large  $\cerk$).
\textbf{(2)} \cref{thm_pred_consist_k_client} provide a \textit{unified} certification against $k$ adversarial users built upon $\epsilon$, which remains valid regardless of how $\epsilon$ is achieved. It thus offers the flexibility of choosing various types of noise, clipping, subsampling strategies, and FL training algorithms to achieve user-level $\epsilon$. DPFL mechanisms that can retain a larger prediction confidence margin under the same $\epsilon$ can certify a larger $\cerk$.
\textbf{(3)} \cref{thm_pred_consist_k_client} is distinct from the maximum adversarial perturbation magnitude against test-time attacks provided by Pixel-DP~\cite{Lecuyer2019} in three important aspects. 
First, we employ group DP to provide certifications against a discrete $k$ number of adversarial users under the threat model of FL poisoning attacks, while Pixel-DP measures maximum perturbation magnitude using the $\ell_p$-norm due to the continuous nature of pixels.
Second, the certification from Pixel-DP is based on the one-time noise in the direct input perturbation during test time, leading to different closed-form solutions for different types of noise distributions such as Laplace and Gaussian. In contrast, \cref{thm_pred_consist_k_client} based  on \textit{$\epsilon$} is a unified certification applicable to \emph{any} user-level DP FL mechanisms.
Third, the analysis of $\epsilon$ in DPFL takes into account more factors than sorely the noise, such as user subsampling and the privacy accountant techniques for DP composition over training rounds.

 \revise{
\textbf{Certified Prediction via R\'enyi DP.}
In addition to the theoretical guarantees of DP-based certified prediction, we also derive the certified prediction based on RDP~\cite{mironov2017renyi} with the randomized smoothing technique via R\'enyi Divergence~\cite{Dvijotham2020A} in \cref{app:renyi_certifications}. 
Yet, compared to DP-based certifications, RDP-based certifications are more intricate, due to the additional parameter, RDP order $\alpha$, and its foundational R\'enyi Divergence-based definition, which makes it more challenging to derive a straightforward upper bound $\cerk$ as in Theorem~\ref{thm_pred_consist_k_client}. 
In our main paper, we focus on DP-based certifications for the convenience of illustration.
}

\vspace{-2mm}
\subsubsection{Certified Attack Inefficacy}
\label{subsubsec:attack_inefficacy}
In addition to the certified prediction, 
we define a bounded \textit{\attackcost} for attacker $C: \Theta \rightarrow \mathbb{R}$, which quantifies the difference between the attack performance of the poisoned model and the \textit{attack goal}, following \cite{ma2019data}. In general, the attacker aims to minimize the \textit{expected} \attackcost $J({D'}):=\mathbb{E}[C(\mathcal{M}({D'}))]$ where  {$\mathcal{M}({D'})$ is the global model trained from poisoned dataset $D'$}, and
the expectation is taken over the randomness of DP training. 
The inefficacy function can be instantiated according to the concrete attack goal in different types of poisoning attacks, and we provide some examples below. {For instance, in \cref{example_backdoor} of backdoor attack, the \attackcost is defined as the loss of the poisoned  FL model $\theta'=\mathcal{M}({D'})$ evaluated on a backdoor testset. During the FL training stage, the attacker optimizes the poisoned FL model $\theta’$ with poisoned training data, so as to minimize the \attackcost $C({\theta'})$ during the test phase. 
The lower the \attackcost, the stronger the attack is. 
}

Given a global FL model {$\mathcal{M}({D'})$} satisfying user-level $(\epsilon,\delta)$-DP, we prove the lower bound of the \attackcost $J(D')$  when there are at most $k$ users.
The existence of the lower bound implies that $J(D')$ can not be arbitrarily low under the constraint of $k$ adversarial users, i.e.,
 the attack can not be arbitrarily successful, which reflects the robustness of the trained global model.
A higher lower bound of the \attackcost (i.e., less effective attack) indicates a more \textit{certifiably robust}  global model.

\vspace{-1mm}
\begin{example}(Backdoor attack~\cite{gu2019badnets})
\label{example_backdoor}
$C({\theta'}) = \frac{1}{n} \sum_{i=1}^n l({\theta'}, z_i^*)$, where $z_i^*= (x_i + \delta_x, y^*)$, $\delta_x$ is the backdoor pattern, $y^*$ is the target adversarial label.
Minimizing $J(D')$ {over model parameters $\theta'$} drives the prediction on test data with backdoor pattern $\delta_x$ to  $y^*$.  
\end{example}

\vspace{-2mm}
\begin{example}(Label Flipping attack~\cite{biggio2012poisoning})
\label{example_label_flip}
$C({\theta'}) = \frac{1}{n} \sum_{i=1}^n l({\theta'}, z_i^*)$,  where $z_i^*= (x_i  ,y^*)$ and  $y^*$ is the target adversarial label. 
Minimizing $J(D')$ {over model $\theta'$} drives the prediction on test data to $y^*$.
\end{example}

\vspace{-1mm}
\textbf{Certified Attack Inefficacy under $k$ Adversarial Users.} We discuss our main results on certified \attackcost below.
\begin{restatable}[Attack Inefficacy with $k$ Attackers]{theorem}{thmcostfunck}
\label{thm_costfunc_k_client}
Suppose a randomized mechanism $\Mcal$ satisfies user-level $(\epsilon, \delta)$-DP. For two user sets $B$ and $B^\prime$ that differ $k$ users, $D$ and $D'$ are the corresponding training datasets. Let $J(D)$ be the expected \attackcost where  $ | C(\theta)| \leq \bar C$, $\forall \theta$. Then,

\begin{equation}
\begin{aligned} 
&{\min\{	 e^{k\epsilon} J(D) +  \frac{e^{k \epsilon }-1}{e^\epsilon-1}\delta \bar C , \bar C \} \geq  J(D') } \\  
&{\geq \max \{e ^{-k\epsilon} J(D) - \frac{1-e^{-k \epsilon }}{e^\epsilon-1}\delta \bar C ,0 \}, \text{\quad if\quad} C(\cdot) \geq 0}  \\  
& {\min\{ e^{-k \epsilon} J(D)+ \frac{1-e^{-k \epsilon }}{e^\epsilon-1}\delta \bar C  , 0 \}  \geq   J(D') }\\  
&{ \geq \max \{ e ^{k\epsilon} J(D) - \frac{ e^{k \epsilon } - 1}{e^\epsilon-1}\delta \bar C  , -\bar C  \}  , \text{\quad if\quad}  C(\cdot) \leq 0 }
\end{aligned} 
\end{equation}
\vspace{-3mm}
\end{restatable}

\vspace{-1mm}
\begin{proof}[Proof sketch]
\cref{thm_costfunc_k_client} contains the lower bound and upper bound for \attackcost. 
For the lower bound, we generalize the proof from \textit{DP in centralized learning}~\cite{ma2019data} to the \textit{user-level DP in FL}. 
Concretely, we derive the lower bound of $J(D')$ based on $J(D)$ according to the satisfied condition in the Group DP definition for the neighboring datasets differing $k$ users. 
In addition, we prove the upper bound by leveraging the symmetric property of DP neighboring datasets.
The full proofs are omitted to \cref{sec:app_proofs}.
\end{proof}

\vspace{-2mm}
\textit{Remark.} In \cref{thm_costfunc_k_client}, 
\textbf{(1)} the lower bounds show to what extent the attack can reduce $J(D')$ by manipulating up to $k$ users, i.e., how successful the attack can be. The lower bounds depend on $J(D)$, $k$, and $\epsilon$.
{Here $J(D)$ is the \attackcost evaluated on the global model trained from clean dataset $D$, which is unrelated to the adversarial users and is only influenced by DPFL mechanism $\mathcal{M}$.}
When  $J(D)$ is higher {(i.e., the clean model $\mathcal{M}({D})$ is more robust)}, the DPFL model under poisoning attacks {$\mathcal{M}({D'})$} is more robust because the lower bounds are accordingly higher; a tighter privacy guarantee, i.e., smaller $\epsilon$, can also lead to higher robustness certification as it increases the lower bounds. On the other hand, with larger $k$, the attacker ability grows and thus leads to lower $J(D')$. 
\textbf{(2)} 
\revise{The upper bounds indicate the minimal adversarial impact caused by $k$ attackers, demonstrating the vulnerability of DPFL models in the most optimistic case (e.g., the backdoor pattern is less distinguishable).}
\textbf{(3)}
Leveraging the above lower bounds, we can lower bound the minimum number of attackers required to reduce \attackcost to a certain level associated with hyperparameter $\tau$ in \cref{thm_k_clients_bound}.
\vspace{-1mm}
\begin{restatable}[Lower Bound of $k$ Given $\tau$, extended from \cite{ma2019data}]{corollary}{corcostfunctau}
	\label{thm_k_clients_bound}
Suppose a randomized mechanism $\Mcal$ satisfies user-level $(\epsilon, \delta)$-DP. Let \attackcost function be $C(\cdot)$, the expected \attackcost be $J(\cdot)$.  In order to achieve $J(D^\prime) \leq  \frac{1}{\tau}  J(D) $ for $\tau \geq 1$ when $ 0\leq C(\cdot) \leq \bar C$, or achieve $J(D^\prime) \leq  \tau  J(D) $ for $ 1 \leq  \tau \leq -\frac{\bar C}{J(D)}$ when $ -\bar C \leq C(\cdot) \leq 0 $, 
the number of adversarial users should satisfy the following:
\begin{equation}
\begin{aligned} 
&{k \geq\frac{1}{\epsilon} \log \frac{\left(e^{\epsilon}-1\right) J(D) \tau+\bar{C} \delta \tau}{\left(e^{\epsilon}-1\right) J(D)+\bar{C} \delta \tau} \text{ \xspace or \xspace  }} 
 k \geq\frac{1}{\epsilon} \log \frac{\left(e^{\epsilon}-1\right) J(D) \tau - \bar{C} \delta}{\left(e^{\epsilon}-1\right) J(D) - \bar{C} \delta },  \nonumber
\end{aligned} 
\end{equation}
\end{restatable}

\vspace{-1mm}
\begin{proof}[Proof sketch]
The proof  generalizes the proof of DP in centralized learning~\cite{ma2019data} to the user-level DP in FL. Consider the case $ 0\leq C(\cdot) \leq \bar C$, when the lower bound of $ J(D')$ in \cref{thm_costfunc_k_client} is smaller than the desired \attackcost $\frac{1}{\tau}  J(D)$, the current \attackcost $J(D')$ will be smaller than the desired \attackcost, i.e., $J(D^\prime) \leq  \frac{1}{\tau}  J(D)$, indicating  the desired attack effectiveness under $k$ adversarial users. \cref{thm_k_clients_bound} states the aforementioned condition. 
The full proofs are omitted to \cref{sec:app_proofs}.
\end{proof}
\vspace{-2mm}
\textit{Remark.} \cref{thm_k_clients_bound} shows that 
 stronger privacy guarantee (i.e., smaller $\epsilon$) requires more attackers to achieve the same effect of the attack, indicating higher robustness.

\section{Instance-level DP and Certified Robustness}
\label{sec:ins_privacy_robustness}

\subsection{Instance-level Privacy} \label{sec:ins_privacy}

We start by introducing instance-level DP definition that 
protects privacy of individual instances, and guarantees that the trained stochastic FL model should not differ much if one instance is modified.

\vspace{-2mm}
\begin{definition}[Instance-level $(\epsilon,\delta)$-DP]
\label{def:insdp}
Let $D$ be the dataset that is the union of local training examples from all  users. Then,  $D$ and $D'$ are adjacent if they differ by one instance.
The randomized mechanism $\mathcal{M}$ is instance-level $(\epsilon, \delta)$-DP if it meets \cref{def:dp} with  $D$ and $D'$ as adjacent datasets.
\end{definition}

Next, we revisit \insdpfedsgd{}~\cite{malekzadeh2021dopamine} and \insdpfedavg{}, where each user adds noise in each training step using DP-SGD~\cite{abadi2016deep}  when training its local model based on Fed-SGD and Fed-Avg, respectively. Then, we formally provide the corresponding privacy analysis.

\subsubsection{Instance-level DP for FedSGD}\label{subsubsec:ins_privacy_fedsgd}
Dopamine~\cite{malekzadeh2021dopamine} provides the first \inslevel{} DP guarantee for the DP-SGD~\cite{abadi2016deep} training of FedSGD~\cite{mcmahan2016communication}. 
Although FedSGD performs the user sampling on the server and the batch sampling in each user during training, Dopamine neglects the privacy gain provided by random user sampling, unlike the privacy analysis in user-level DP. Therefore, we improve the privacy guarantee via privacy amplification~\cite{bassily2014private, abadi2016deep} with user sampling. 
In addition, we use the \revise{R\'enyi DP (RDP)} accountant~\cite{wang2019subsampled}, instead of the moment accountant~\cite{abadi2016deep} used in Dopamine~\cite{malekzadeh2021dopamine}, for a tighter privacy budget analysis, given its tighter compositions rules based on R\'enyi divergence~\cite{mironov2017renyi}.

Specifically, in \insdpfedsgd{} (\cref{algo:Fedsgd_sample_privacy} in~\cref{sec:dpfl_algo}), each user updates its local model by one step of DP-SGD~\cite{abadi2016deep} to protect the privacy of each training instance, the randomized mechanism $\Mcal$ that outputs the global model preserves the instance-level DP. The one-step update for the global model can be described as follows:

{
\begin{small}
    \begin{equation}\label{eq:onestep_globalupdate_fedsdg}
	w_t \gets w_{t-1} - \frac{1}{m} \sum_{i\in U_t}   \frac{\eta}{L} \left(\sum_{x_j\in b_t^i} \mathrm{Clip}(\nabla l_i(w_{t-1};x_j), S) + \mathcal{N}\left(0, \sigma^{2} S^{2}\right) \right ), 
\end{equation}
\end{small}
}
where $b_t^i$ is a local batch randomly sampled from the local dataset of user $i$, $L$ is the batch size, $\nabla l_i(w_{t-1};x_j)$ is the gradient for local sample ${x_j\in b_t^i}$ calculated upon the current FL model $w_{t-1}$, and $\mathcal{N}\left(0, \sigma^{2} S^{2}\right)$ is the Gaussian noise added to the per-sample gradient. 

\vspace{-1mm}
\begin{restatable}[\insdpfedsgd{} Privacy Guarantee]{proposition}{propinsdpfedsgd}
\label{proposition:insdpfedsgd_privacy_guarantee}
Given batch sampling probability $p$ without replacement, and user sampling probability $q=\frac{m}{N}$ without replacement, FL rounds $T$, the clipping threshold $S$, the noise parameter $\sigma$,  the randomized mechanism $\Mcal$ in \cref{algo:Fedsgd_sample_privacy}  satisfies  $
(T \epsilon'(\alpha) + \log \frac{\alpha -1 }{\alpha} - \frac{\log\delta + \log \alpha}{ \alpha -1}  , \delta)$-DP with $\epsilon (\alpha)= \alpha/(2m\sigma^2)$ where  $\alpha$ is the RDP order and 
\begin{small}
\begin{equation}
\begin{aligned}
& \epsilon^{\prime}(\alpha) \leq \frac{1}{\alpha-1} \cdot  \log \left(1+(pq)^{2}\left(\begin{array}{c}
\alpha \\
2
\end{array}\right) \min \left\{4\left(e^{\epsilon(2)}-1\right), e^{\epsilon(2)} \cdot \right.\right.\\
&\left.\left. \min \left\{2,\left(e^{\epsilon(\infty)}-1\right)^{2}\right\}\right\} +\sum_{j=3}^{\alpha} (pq)^{j}\left(\begin{array}{c}
\alpha \\
j
\end{array}\right) e^{(j-1) \epsilon(j)} \min \left\{2,\left(e^{\epsilon(\infty)}-1\right)^{j}\right\}\right)  \nonumber
\end{aligned}
\end{equation}
\end{small}
\end{restatable}
\vspace{-4mm}
\begin{proof}[Proof sketch]
We use $pq$ to represent \emph{instance-level} sampling probability, $T$ to represent FL training rounds, $\sigma \sqrt{m}$ to represent the \textit{equivalent global noise} level. 
The rest of the proof follows \textbf{(1)} RDP subsampling amplification~\cite{wang2019subsampled}, \textbf{(2)} RDP composition for privacy budget accumulation over $T$ rounds based on the RDP composition~\cite{mironov2017renyi}
and \textbf{(3)} transferring RDP guarantee to DP guarantee based on the conversion theorem~\cite{balle2020hypothesis}.
\end{proof}

\vspace{-1mm}
\subsubsection{Instance-level DP for FedAvg}
Dopamine only allows users to perform \textit{one} step of DP-SGD~\cite{abadi2016deep} during each FL round,
while in practice, users are typically allowed to update their local models for many steps before submitting updates to reduce communication costs. To solve this problem, we introduce \insdpfedavg{} (\cref{algo:Fedavg_sample_privacy} in~\cref{sec:dpfl_algo}),
where each user $i$ performs local DP-SGD for multiple steps so that the local training mechanism $\Mcal^i$ satisfies instance-level DP. Then, the server aggregates the updates by FedAvg. We prove that the global mechanism $\Mcal$ preserves instance-level DP using DP parallel composition theorem~\cite{mcsherry2009privacy}.

In \insdpfedavg{}, before FL training, local privacy costs $\{\epsilon_0^i\}_{i \in [N]}$ are initialized as 0.
At round $t$, if user $i$ is not selected, its local privacy cost is kept unchanged $\epsilon_{t}^i \gets \epsilon_{t-1}^i$ since local dataset $D_i$ is not accessed.
Otherwise, user $i$ updates the local model by running DP-SGD for $V$ local steps with batch sampling probability $p$, noise level $\sigma$ and clipping threshold $S$, and $\epsilon_{t}^i $ is accumulated upon $\epsilon_{t-1}^i$ via its local RDP accountant.
Next, the server aggregates the updates from selected users and leverages the parallel composition in \cref{thm:insdp_t_round} to calculate the global privacy cost $\epsilon_t= \max_{i \in [N]}\epsilon_t^i$.
After $T$ rounds, the mechanism $\Mcal$ that outputs the final FL global model satisfies instance-level $(\epsilon_T,\delta)$-DP.

To derive the privacy guarantee for \insdpfedavg{}, we analyze the privacy cost \textit{accumulation} for each local model across FL training rounds, as well as the privacy cost \textit{aggregation} during model aggregation on the server side at each round.

\vspace{-1mm}
\begin{restatable}[\insdpfedavg{} Privacy Guarantee]{proposition}{thminsdpfedavgprivacytround}
\label{thm:insdp_t_round}
\sloppy
In \cref{algo:Fedavg_sample_privacy}, during round $t$,  the local mechanism $\Mcal^i$ satisfies $(\epsilon^i_t,\delta^i)$-DP, and the global mechanism $\Mcal$ satisfies $\left(\max_{i \in [N]}\epsilon_t^i, \delta^i \right)$-DP. 
\end{restatable}
\vspace{-4mm}
\begin{proof}[Proof sketch]
When $D'$ and $D$ differ in one instance, the modified instance only falls into one user's local dataset for any $t$ training round, and thus parallel composability of DP~\cite{mcsherry2009privacy}  applies. Moreover, server aggregation does not increase privacy costs due to DP post-processing property. The local cost $\epsilon_i$ is only accumulated via the local RDP accountant. Finally, the privacy guarantee corresponds to the worst case and is obtained by taking the maximum local privacy cost across all the users. 
Proof is in \cref{sec:dpfl_algo}.
\end{proof}
\vspace{-2mm}
\textit{Remark.}
\cref{thm:insdp_t_round} provides the privacy guarantee for trained FL global model when users perform local DP-SGD training. To achieve that,  we examine the outcomes from FL local and global randomized mechanisms and analyze the accumulation of local privacy costs and subsequent aggregation of global privacy costs over different training rounds. In the centralized setting, \citet{yu2019differentially} analyzes disjoint data batching and presents similar results. Recent studies~\cite{liu2021projected,liu2022privacy,yangprivatefl}  directly apply the results from~\cite{yu2019differentially} for instance-level DPFL. However, these studies lack a thorough privacy analysis in the context of FL, and our analysis fills this gap.

\vspace{-1mm}
\subsection{Certified Robustness of Instance-level DPFL} 
\label{sec:ins_robustness}

\subsubsection{Threat Model.}\label{subsubsec:ins_level_threat_model}

We consider there are in total $k$ poisoned instances that the same or multiple users could control. 
\squishlist 
\item \textbf{Attack Goal.} The goal of attackers is to mislead the trained global model to make mispredictions by injecting poisoning data during local training. 
\item \textbf{Attack Capability.} In accordance with prior work~\cite{naserilocal}, for attack {capability}, we consider that local users, including adversaries, follow the DP training protocol to protect data privacy. That means the adversaries need to follow the training protocol  to sample local data randomly during training. This scenario is realistic for instance-level DPFL because FL users often run pre-defined programs~\cite{kairouz2021advances, bonawitz2019towards} that implement DP mechanisms.  For example, according to \citet{bonawitz2019towards}, ``If the device has been selected, the FL runtime receives the FL plan, queries the app’s example store for data requested by the plan, and computes plan-determined model updates and metrics.''  
On the other hand, the users have full control over their training data, so they can arbitrarily manipulate the local training data.
Under this setting, the trained FL model is guaranteed to satisfy instance-level DP. 
\item \textbf{Attack Strategy.} It includes data poisoning attacks, e.g., backdoor~\cite{gu2019badnets,chen2017targeted} or label-flipping~\cite{biggio2012poisoning,huang2011adversarial}. Our analysis of certified robustness is agnostic to the specific attack strategy employed.
\squishend

\subsubsection{Certified Robustness.}
According to the group DP property and the post-processing property for the FL model with instance-level $(\epsilon,\delta)$-DP, 
we prove that our robust certification results for user-level DP are also applicable to instance-level DP. Below is the formal theorem (proof is given in Appendix~\ref{sec:app_proofs}).

\vspace{-1mm}
\begin{restatable}[]{theorem}{thmapplytoinsdp}
\label{thm:thmsapplytoinsdp}
Suppose $D$ and $D'$ differ by $k$ instances, and mechanism $\Mcal$ satisfies instance-level $(\epsilon,\delta)$-DP. The results on user-level DPFL in \cref{thm_pred_consist_one_client},~\cref{thm_pred_consist_k_client},~\cref{thm_costfunc_k_client}, and \cref{thm_k_clients_bound} still hold for the instance-level DPFL $\Mcal$, $D$, and $D'$.	
\end{restatable}
\vspace{-2mm}
\textit{Remark.}
We analyze the underlying relationship between privacy and certified robustness under both user-level DPFL and instance-level DPFL, as well as the relationship between these two levels of privacy in FL.
From the privacy perspective, the same $\epsilon$ for these two different privacy levels signifies different privacy scopes. 
One straightforward way to convert instance-level DP to user-level DP is to use Group DP~\cite{dwork2014algorithmic} to incorporate all instances of a user, which could lead to a loose privacy bound.
On the other hand, a randomized mechanism that satisfies $(\epsilon, \delta)$ user-level DP also satisfies $(\epsilon, \delta)$ instance-level DP based on their definitions.
From the certified robustness perspective,
 the same $\epsilon$ on two different privacy levels  implies different levels of robustness.
 When considering the ability to tolerate  adversarial poisoning instances, instance-level DPFL provides rigorous certified robustness as a function of the number of poisoning instances, while
 user-level DPFL may indicate stronger robustness if we consider injecting all poisoning instances with one user. The user-level DPFL, however, might compromise the model utility when controlling per-user sensitivity during DP training.
Thus, different types of DPFL mechanisms and algorithms may be chosen to protect both privacy and robustness considering several factors such as adversarial strategies, data types, and trained model sensitivity. Our evaluation on diverse datasets and different DPFL algorithms in Section~\ref{sec_exp} will validate our analysis and findings on both user-level and instance-level DP, as well as provide more observational insights.

\vspace{-1mm}
\section{Experiments}
\label{sec_exp}

In this section, we conduct the evaluation on three datasets (both image and text data) for the certified robustness of different DPFL algorithms against various poisoning attacks to verify the insights from our theorems.
We highlight our main results and present some interesting findings and ablation studies. 

\vspace{-1mm}
\subsection{Experimental Setup}
\label{subsec:exp_setup}
\subsubsection{Datasets and Models.}
We consider three datasets: 
image classification on {\mnist{}}, {\cifar{}} and text sentiment analysis on tweets from {Sentiment140}~\cite{go2009twitter} (Sent140), which involves classifying Twitter posts as positive or negative. 
For \mnist{}, we use a CNN model with two Conv-ReLu-MaxPooling layers and two linear layers; for \cifar{}, we use the CNN architecture from PyTorch differential privacy library~\cite{opacus2021} which consists of four Conv-ReLu-AveragePooling layers and one linear layer. 
In line with previous work on 
DP ML~\cite{jagielski2020auditing,ma2019data} and 
backdoor attacks~\cite{NEURIPS2018_madry_sepctral,weber2020rab},
we mainly discuss the binary classification for \mnist{} (digit 0 and 1) and \cifar{} (airplane and bird) in the main text, and defer their 10-class results to \cref{sec:app_exp_details}.
For \sent{}, we use a two-layer LSTM classifier containing 256 hidden units with pretrained 300D GloVe embedding~\cite{pennington2014glove} following \cite{li2018federated}. 

\vspace{-0.5em}
\subsubsection{Training Setups}
\label{subsubsec:training_setups}
Unless otherwise specified, we split the training datasets for $N$ FL users in an i.i.d manner for \mnist{} and \cifar{}. For \sent{}, the local datasets are naturally non-i.i.d,  where each Twitter account corresponds to an FL user. We also study the effect of data heterogeneity degrees on certified robustness by simulating FL non-i.i.d setting based on Dirichlet distribution~\cite{hsu2019measuring} in \cref{sec:noniid_results}. 
FL users run SGD with learning rate $\eta$, momentum $0.9$, and weight decay $0.0005$  to update the local models. The training parameter setups, including the  number of total users $N$, the number of selected users per round $m$, local epochs $E$, the number of local SGD steps $V$, local learning rate $\eta$, batch size, etc., are summarized in \cref{tb:dataset_parameters}. 

\revise{
To simulate cross-device settings for \userdpfedavg{}, we follow the FL settings in previous studies and use Sent140 data with $\sim$ 800 clients~\cite{li2018federated}, and CIFAR/MNIST with 200 clients~\cite{mcmahan2016communication}. To simulate cross-silo FL settings for \insdpfedavg{}, we train  DPFL models}
on \mnist{} and \cifar{} with 10 users.
Following \cite{mcmahan2018learning} that use $\delta \approx \frac{1}{N^{1.1}}$ as privacy parameter, for \userdpfedavg{} we set $\delta=0.0029$ for \mnist{} and \cifar{}, and $\delta=0.000001$ for \sent{}   according to the total number of  users; for \insdpfedavg{} we set  $\delta=0.00001$ according the total number of training samples.
When training on CIFAR10, we follow the standard practice for differential privacy~\cite{abadi2016deep,jagielski2020auditing} that fine-tunes a whole model pre-trained non-privately on CIFAR100~\cite{cifar}.
We refer the readers to \cref{sec:app_exp_details} for more details about detailed hyperparameters for differential privacy.

{\begin{table}
\caption{Dataset description and parameters.}
\vspace{-4mm}
\begin{center}
\scalebox{0.7}{
\begin{tabular}{l  c c c c c c c c c c c}
\toprule
Algorithm & Dataset  &  $N$ &  $m$  & $E$ & $V$ & batch size & $\eta$ &  $S$ & $\delta$ & $\bar C$ \\
\midrule 
\userdpfedavg & \mnist{} & 200 & 20 & 10 & /& 60 &  0.02 & 0.7 & 0.0029& 0.5\\
\userdpfedavg & \cifar{}   & 200 & 40 & 5 & /& 50 &  0.05 & 1 & 0.0029& 0.2\\
{\userdpfedavg} & {\sent{}} &  {805} &  {10} &  {1} & / &  {10} &   {0.3} &  {0.5} &  {0.000001} & 1.4\\
\insdpfedavg & \mnist{}  & 10 & 10 & /& 25 & 50 &  0.02 & 0.7 & 0.00001& 0.5  \\
\insdpfedavg & \cifar{} & 10 & 10 & /&  100 & 50 &  0.05 & 1 & 0.00001& 2\\
\bottomrule 
\end{tabular}
}
\end{center}
\label{tb:dataset_parameters}   
\vspace{-7mm}
\end{table}
}

\vspace{-0.5em}
\subsubsection{Poisoning Attacks.}
\sloppy
We evaluate four poisoning attacks against our DPFL mechanisms, which represent the common threats in FL research.
We consider \textit{backdoor attacks (BKD)} on image datasets~\cite{bagdasaryan2020backdoor} and \textit{label flipping attacks (LF)} on image and text datasets~\cite{fung2020limitations} against both levels of DPFL. 
For \insdpfedavg{}, we evaluate the worst-case where $k$ backdoored or label-flipped instances are injected into the dataset of one user.
\revise{For \userdpfedavg{}, we additionally evaluate the 
\textit{static optimization attacks (STAT-OPT)}~\cite{shejwalkar2021manipulating}, which solve the adversarial optimization problem to find poisoning model updates, as well as}
\textit{distributed backdoor attack (DBA)}~\cite{xie2019dba}, which decomposes the backdoor pattern into several smaller ones and embeds them into different local training sets for different adversarial users. Moreover, we also consider BKD, LF, and DBA via \textit{model replacement} attack~\cite{bagdasaryan2020backdoor,bhagoji2018analyzing} where $k$ attackers train the local models using local datasets with $\alpha$ fraction of poisoned instances, and scale the malicious updates directly with hyperparameter $\gamma$, 
i.e., $\Delta w_{t}^i \gets  \gamma \Delta w_{t}^i$, 
before sending them to the server. 
This way, the malicious updates would have a stronger impact on the FL model. 
Note that even when attackers perform scaling after server clipping, the sensitivity of each model update is still upper-bounded by the clipping threshold $S$, so the privacy guarantee of user-level DPFL still holds under poisoning attacks via model replacement.

Specifically, for the attacks against {\userdpfedavg{}}, by default, the local poison fraction is $\alpha=100\%$, and the scale factor is $\gamma=50$. We use the same parameters setups for all $k$ attackers. In terms of label flipping attacks, the attackers swap the label of images in the source class (digit 1 for \mnist{}; bird for \cifar{}; positive for \sent{}) into the target label (digit 0 for \mnist{}; airplane for \cifar{}; negative for \sent ). 
In terms of backdoor attacks in \mnist{} and \cifar{}, the attackers add a triangle pattern in the right lower corner of the image as the backdoor pattern and swap the label of any sample with such pattern into the target label (digit 0 for \mnist{}; airplane for \cifar{}). In terms of distributed backdoor attacks, the triangle pattern is evenly decomposed and injected by the $k$ attackers. 
For the attacks against  {\insdpfedavg{}}, the same target classes and backdoor patterns are used as \userdpfedavg{}.

\vspace{-2mm}
\subsubsection{Evaluation Metrics.}\label{subsubsec:evalution_metrics}
We consider two evaluation metrics based on our robustness certification criteria. 
\squishlist
    \item \textbf{Certified Accuracy.} To evaluate the \textit{certified prediction}, we calculate {\ceracc{}}, which is the fraction of the test set for which the poisoned DPFL model makes correct and the same prediction compared with that of the clean model. \revise{The test set can be either poisoned or clean based on \cref{thm_pred_consist_k_client}. Given that the certifications are agnostic to the actual attack strategy, and certain attacks, such as model poisoning and label flipping, do not produce poisoned test input samples $x$,  we use the clean test samples to calculate the certification following the standard setting of certified robustness in centralized systems~\cite{cohen2019certifiedrandsmoothing}.}
    Given a test set of size $n$, for the $i$-th test sample \revise{$x_i$}, the ground truth label is $y_i$, the output prediction is $c_i$ , and the number of adversarial users/instances that can be certifiably tolerated is $\cerk_i$ \revise{based on Equation \ref{eq:cerk}}.
We calculate the \ceracc{} given $k$ adversarial users/instances as 
$\frac{1}{n} \sum_{i=1}^n \mathds{1}\{c_i=y_i$ and $\cerk_i  \ge k\}$. 

    \item \textbf{Lower bound of \attackcost.} To evaluate the \textit{certified \attackcost}, we calculate the {lower bound of \attackcost} in \cref{thm_costfunc_k_client}: $ \underline{J(D')} = \max \{e ^{-k\epsilon} J(D) - \frac{1-e^{-k \epsilon }}{e^\epsilon-1}\delta \bar C ,0 \} $. This lower bound represents the cost of the attacker for performing poisoning attacks. The lower the certified attack inefficacy is, the less robust the model is. We evaluate the tightness of $\underline{J(D')}$ by comparing it with the empirical \attackcost $J(D')$ under different attacks.
\squishend

\vspace{-1mm}
\subsubsection{Robustness Certification with Monte Carlo Approximation.}\label{subsubsec:certification_monte_carlo_approx}
The robustness certifications presented in our theorems depend on the expected confidence $F_{c}(\mathcal{M}(D), x)$ for class $c$ or expected attack inefficacy $J(D)$. We take $F_{c}(\mathcal{M}(D), x)$ as an example here, and denote $F_{c}(\mathcal{M}(D), x)$ as $F(\mathcal{M})$ for simplicity.
In practice, $F(\mathcal{M})$ is not directly used for prediction because the true expectation cannot be analytically computed for deep neural networks. \textit{To empirically verify the insights provided by our theorems}, we follow the convention in prior work on certified robustness~\cite{Lecuyer2019,ma2019data,weber2020rab, rosenfeld2020certifiedlableflip, cao2021provably,cohen2019certified} to use $\widetilde{F}(\mathcal{M})$, which is a Monte Carlo approximation of $F(\mathcal{M})$ by taking the average over $O$ models outputs for utility evaluation in our experiments.  
Note that \textbf{(1)} from the \textit{DP perspective}, increasing $O$ increases the overall privacy budget as the sampling process re-accesses the sensitive data and consumes the privacy budget. Based on standard DP composition theory~\cite{dwork2006our}, calculating $\widetilde{F}(\mathcal{M})$ costs $O\epsilon$ privacy budget, where $\epsilon$ is the privacy budget consumed by training one model; 
\textbf{(2)} From the \textit{robustness certification perspective}, the estimation of $\widetilde{F}(\mathcal{M})$  will be more accurate with higher confidence when we use larger $O$;
\textbf{(3)} Using a single model  for prediction is equivalent to computing $\widetilde{F}(\mathcal{M})$ with $O=1$, leading to strong privacy protection but low confidence for the robustness certification.

Specifically, we estimate the expected class confidence by $F_c(\Mcal(D),x) \approx \frac{1}{O} \sum_{s=1}^O f^s_c $ to evalute \cref{thm_pred_consist_k_client}, where  each $f^s_c = f_c(\Mcal(D),x)$ is obtained from one DPFL model.
Similarly, we approximate the \attackcost to evaluate \cref{thm_costfunc_k_client} and \cref{thm_k_clients_bound}.
\revise{We use a relatively large $O=1000$ for certified accuracy and $O=100$ for certified \attackcost in experiments so as to obtain an accurate approximation of the expectation following~\cite{ma2019data}} \revise{and precisely reveal the connections between the privacy parameters ($\epsilon$,$\delta$) and certified robustness under different criteria.}  In \cref{appendix_confidence_level}, we use Hoeffding's inequality~\cite{hoeffding1994probability} to calibrate the empirical estimation with confidence level parameter $\psi$.

{
\newlength{\accuserheightc}
\settoheight{\accuserheightc}{\includegraphics[width=.30\linewidth]{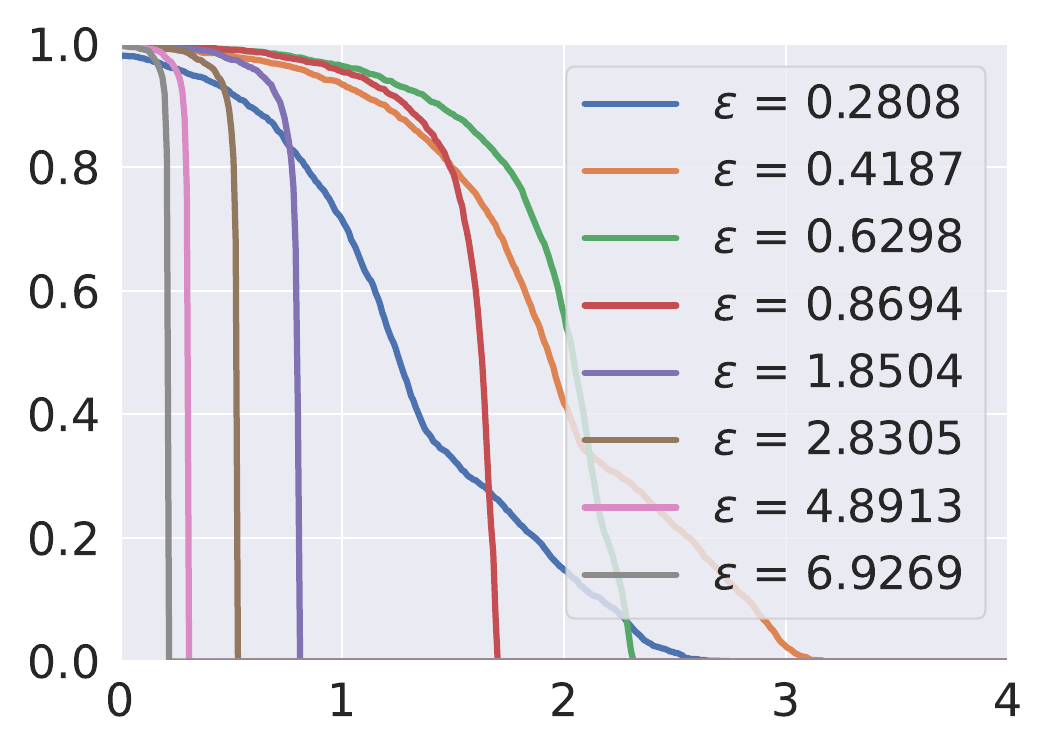}}

\newcommand{\rowname}[1]%
{\rotatebox{90}{\makebox[\accuserheightc][c]{\scriptsize #1}}}

\begin{figure}
\centering

{
\renewcommand{\tabcolsep}{10pt}
\begin{subtable}[]{\linewidth}
\centering
\begin{tabular}{@{}p{2mm}@{}c@{}c@{}c@{}c@{}c@{}c@{}c@{}}
        & \makecell{{\scriptsize (a) \mnist}}
        & \makecell{{\scriptsize (b) \cifar}}
        & \makecell{{\scriptsize (c) \sent }}\\
\rowname{\makecell{\ceracc}}&
\includegraphics[height=\accuserheightc]{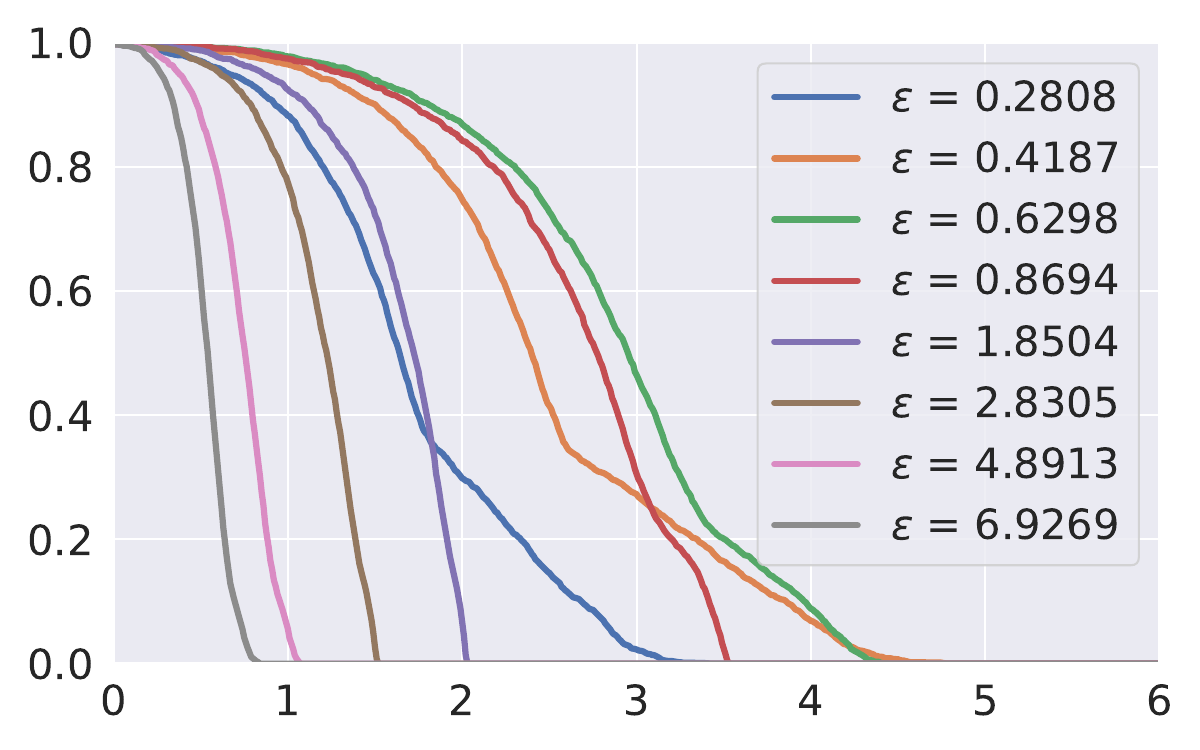}&
\includegraphics[height=\accuserheightc]{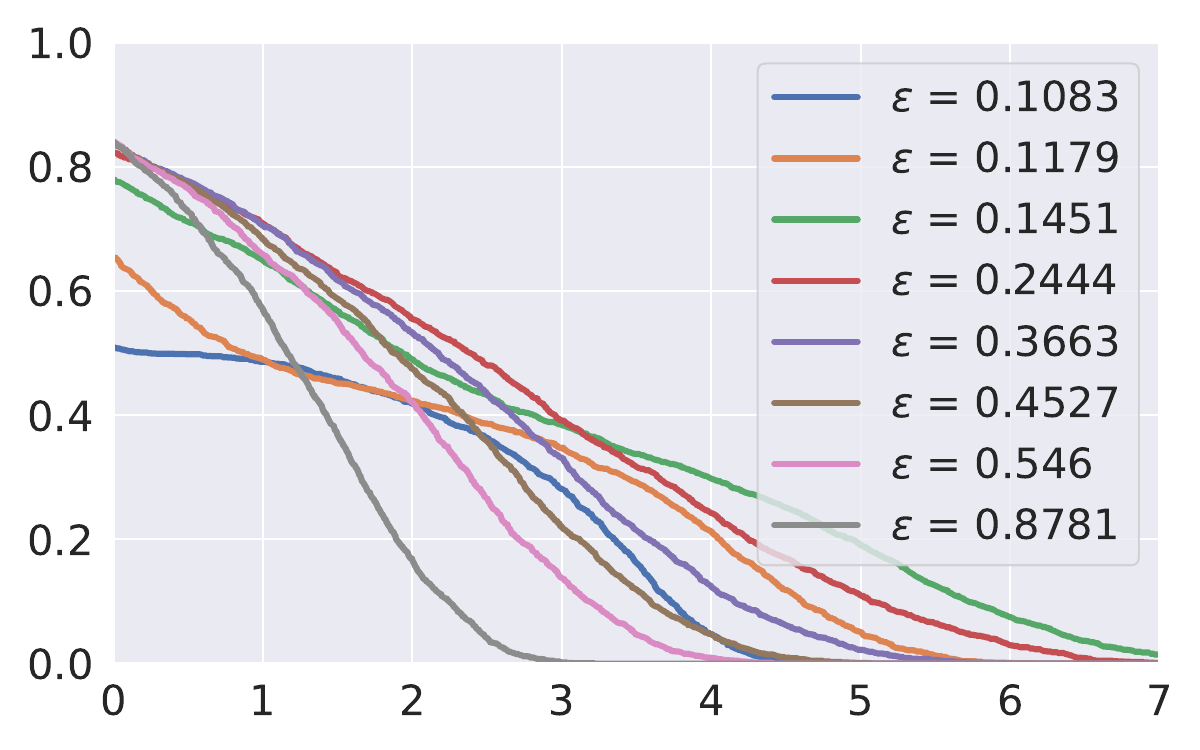}&
\includegraphics[height=\accuserheightc]{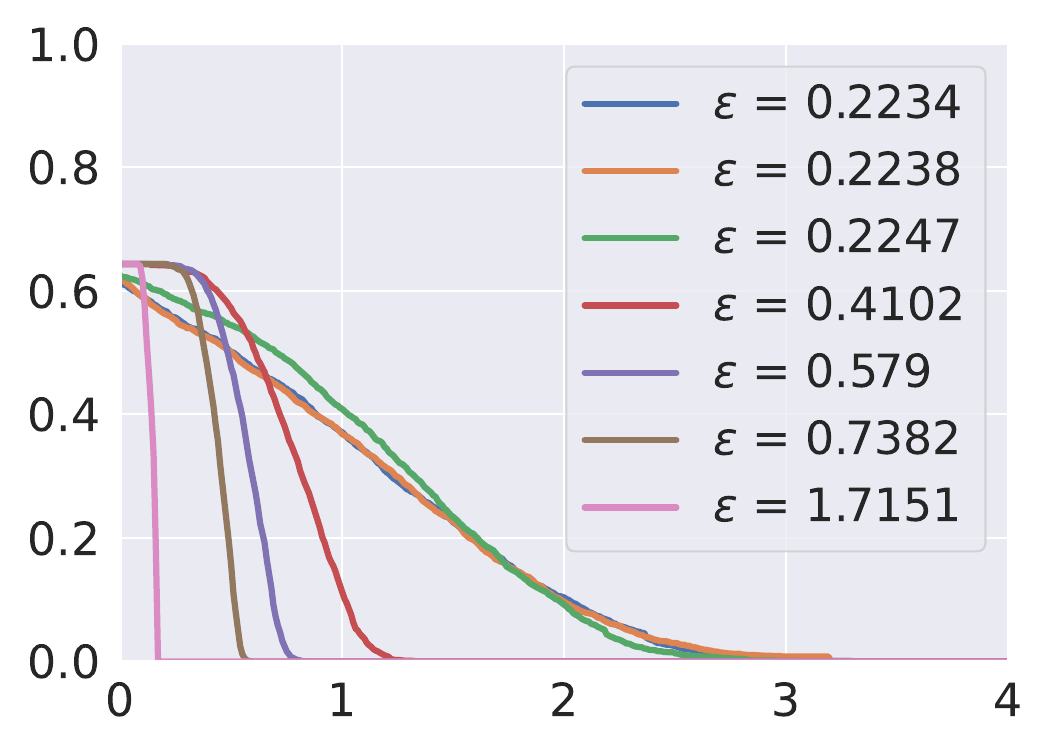}& \\[-2ex]
& \makecell{{\scriptsize $k$}}& \makecell{{\scriptsize $k$}}& \makecell{\scriptsize $k$} 
\end{tabular}
\end{subtable}
}
\vspace{-5mm}
\caption{\small Certified accuracy of \userdpfedavg{} under different privacy budgets $\epsilon$.}
\label{fig:userdp_ceracc}
\vspace{-2mm}
\end{figure}

}

\vspace{-2mm}
\subsection{Evaluation Results of User-level DPFL}

Here we present our main results on user-level DPFL based on the \textbf{certified accuracy} under different (1) privacy budget $\epsilon$, (2) DPFL algorithms, and (3) data heterogeneity degrees; \revise{\textbf{empirical accuracy} under (1) different poisoning attacks and (2) comparison to empirical FL defenses}; \textbf{certified \revise{and empircal} \attackcost} under (1) different $k$ and poisoning attacks, and (2) different $\epsilon$.
\vspace{-0.5em}
\subsubsection{Certified Accuracy under Different $\epsilon$.}
\label{sec:eval_user_level}
\cref{fig:userdp_ceracc} presents the user-level certified accuracy under different $\epsilon$ by training \userdpfedavg{} with different noise scale $\sigma$. \revise{(The uncertified accuracy of \userdpfedavg{} under non-DP training and DP training is deferred to \cref{subsubsec:training_details}.)}
\revise{Since each test sample {$x_i$} has its own certified {$\cerk_i$}, the largest $k$ that an FL model can reach is a threshold that none of the test samples have a larger \revise{$\cerk_i$} than it\revise{, i.e., $\cerk_i<k, \forall i$, which can be observed as the largest value on the x-axis of \cref{fig:userdp_ceracc}}.}  Note that here we calculate the certified \revise{$\cerk_i$} as the numerical upper bound in \cref{thm_pred_consist_k_client}, which could be fractional.

We observe that \revise{\textbf{(1)}} the largest number of adversaries $k$ can be certified when $\epsilon$ is around {$0.6298$ ($0.1451$, $0.2238$) on \mnist{} (\cifar{}, \sent{})}, which verifies the relationship between $\epsilon$ and certified accuracy as discussed in \cref{sec:cer_robust_user_level}.
In particular, when $\epsilon$ is too large, \revise{$\cerk_i$} decreases since $\epsilon$ is in the denominator of Equation~\ref{eq:cerk}; when $\epsilon$ is too small, large noise is added during training, which hurts the model utility, and the model is not confident in predicting the top-1 class, thus decreasing the margin between $F_{\Aclass}$ and $F_{\Bclass}$ and decreasing \revise{$\cerk_i$}. 
\revise{\textbf{(2)} Additionally, for each fixed $k$, there is an optimal $\epsilon$ that yields the maximum certified accuracy due to similar reasons. For example, to certify $k=2$ adversaries, the $\epsilon$ with highest certified accuracy is around $0.6298$ ($0.2444$, $0.2234$) on \mnist{} (\cifar{}, \sent{}). 
}
\revise{\textbf{(3)}} Given that there is a $\epsilon$ achieving maximal certified number of adversaries $k$ \revise{or yielding the maximum certified accuracy under a fixed $k$}, properly choosing $\epsilon$ would be important for certified accuracy. 
\revise{As the optimal $\epsilon$ is data/task-dependent,  one can find it automatically as hyperparameter tuning. Our evaluation can serve as a guide for similar data/tasks to narrow down the search space of $\epsilon$.}
\revise{\textbf{(4)}
We also notice that for certain datasets like \cifar{}, the ideal $\epsilon$ for certified accuracy can be small, primarily because the datasets are inherently difficult to learn. Nevertheless, on simpler datasets like \mnist{}, using $\epsilon$= 0.6298 to train DPFL models remains feasible (with 97\% clean accuracy) and yields the maximal certified $k \approx$ 4.  When DPFL algorithms offer improved utility and a larger confidence margin,  a larger  $\epsilon$ can be used to certify the same $k$, as indicated in \cref{thm_pred_consist_k_client}. Moreover, enhanced privacy accountants that produce a tighter DP bound naturally result in a smaller $\epsilon$ without impacting model utility. As our paper focuses on deciphering the privacy-robustness interplay, our findings — both theoretical and empirical — imply opportunities to further improve the utility of current DPFL algorithms or the tightness of privacy accountants in order to achieve higher certified robustness for FL.
}

{
\newlength{\accuserdpflheightc}
\settoheight{\accuserdpflheightc}{\includegraphics[width=.30\linewidth]{figures/ccsfinal_plots/cer_acc_conf/cer_acc_mnist.pdf}}

\newcommand{\rowname}[1]%
{\rotatebox{90}{\makebox[\accuserdpflheightc][c]{\scriptsize #1}}}

\begin{figure}
\centering
{
\renewcommand{\tabcolsep}{10pt}
\begin{subtable}[]{\linewidth}
\centering
\begin{tabular}{l}
\includegraphics[height=0.31\accuserdpflheightc]{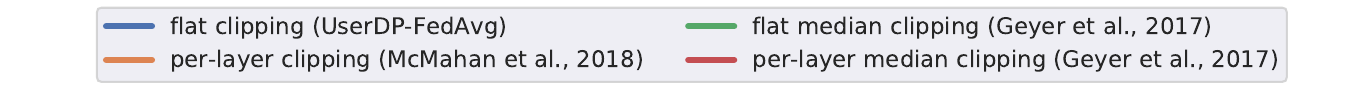}
\end{tabular}
\end{subtable}
\begin{subtable}[]{\linewidth}
\centering
\begin{tabular}{@{}p{2mm}@{}c@{}c@{}c@{}c@{}c@{}c@{}c@{}}
    & \makecell{\shortstack{\scriptsize (a) \mnist{}   ($\epsilon=0.63$)}}
        & \makecell{\shortstack{\scriptsize (b) \cifar{}   ($\epsilon=0.53$)}}
        & \makecell{\shortstack{\scriptsize (c) \sent{}  
 ($\epsilon=0.40$) }}\\
\rowname{\makecell{\ceracc}}&
\includegraphics[height=1.1\accuserdpflheightc]{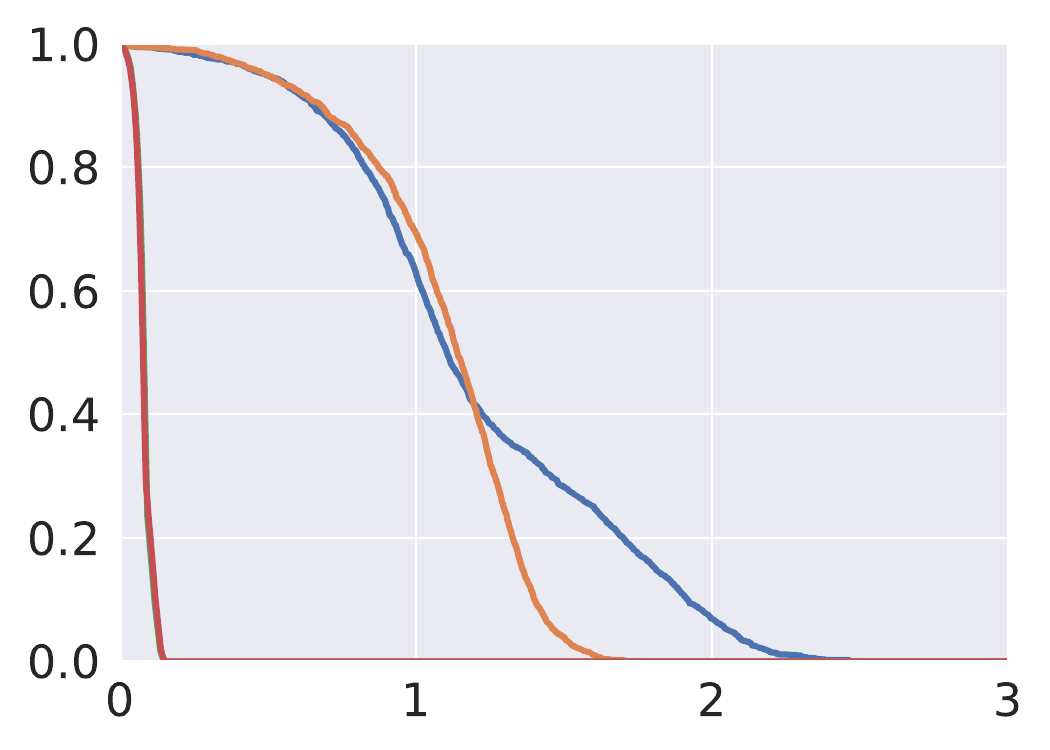}&
\includegraphics[height=1.1\accuserdpflheightc]{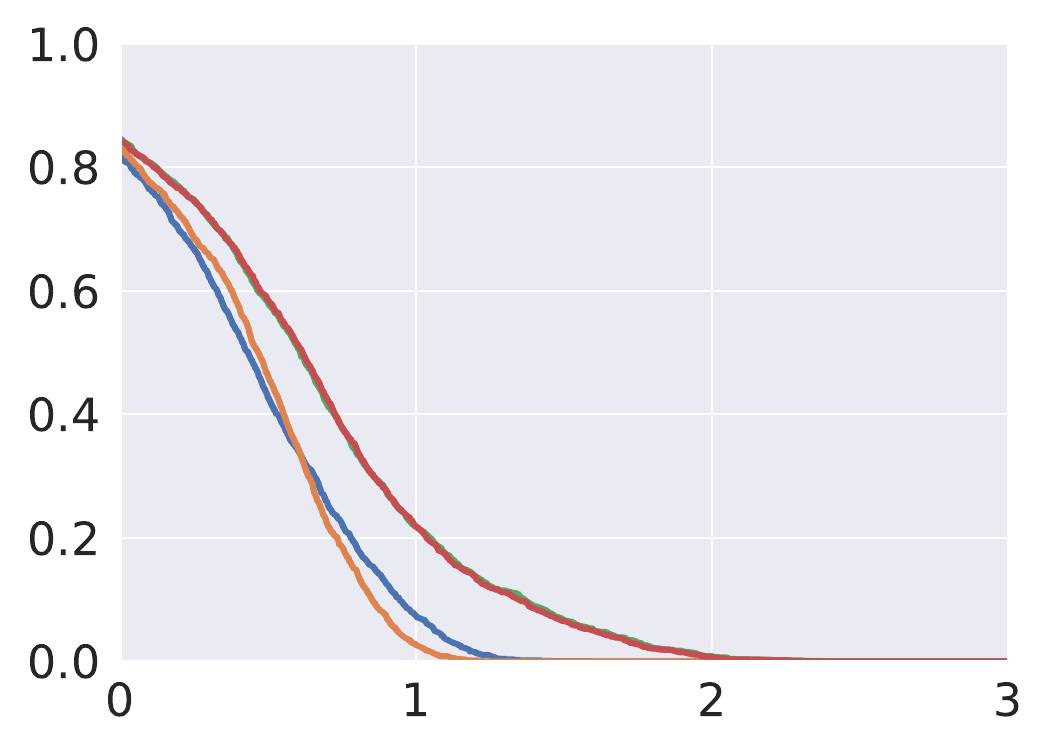}&
\includegraphics[height=1.1\accuserdpflheightc]{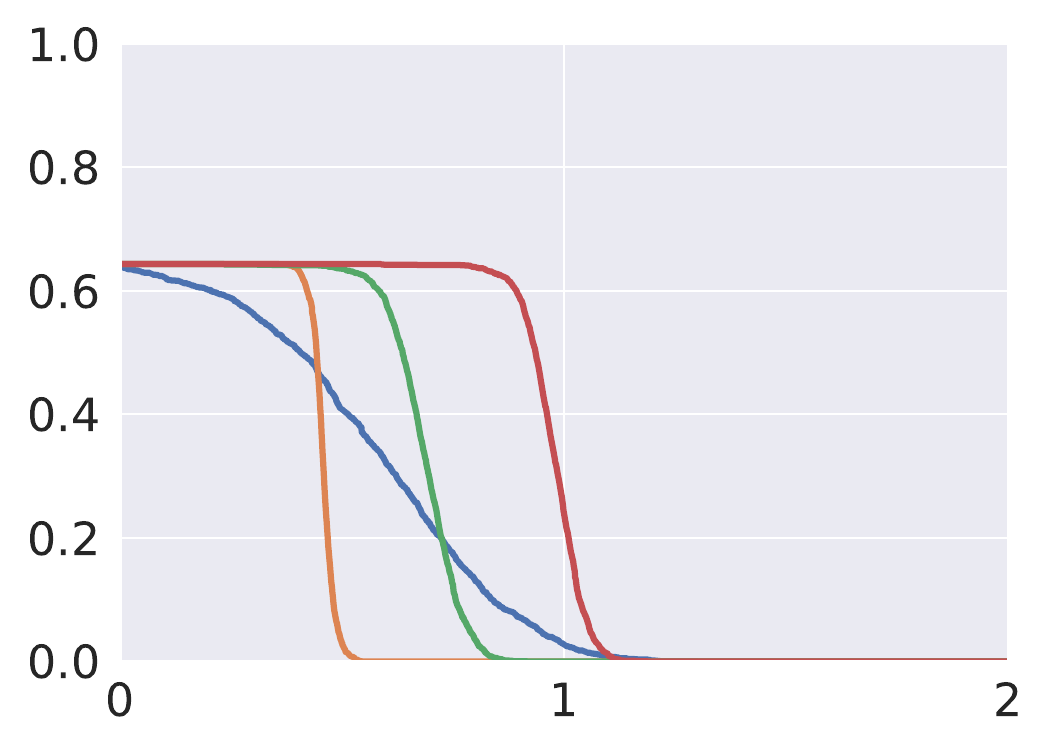}&\\[-1.2ex]
& \makecell{{\scriptsize $k$}}& \makecell{{\scriptsize $k$}}& \makecell{\scriptsize $k$} 
\end{tabular}
\end{subtable}
}
\vspace{-5mm}
\caption{\small Certified accuracy of \userdpfedavg{} under different user-level DPFL algorithms with the same $\epsilon$.}
\label{fig:userdp_ceracc_dpfl}
\vspace{-5mm}
\end{figure}

}

\vspace{-1mm}
\subsubsection{Certified Accuracy under Different DPFL Algorithms.}\label{subsubsec:cer_acc_under_diff_dpfl}
Given that our certifications are agnostic to DPFL algorithms (i.e., the certifications hold no matter how $(\epsilon, \delta)$ is achieved), we are able to compare the certified results of different DPFL algorithms given the same privacy budget $\epsilon$. 
Specifically, we consider the following four DPFL algorithms with different clipping mechanisms: 
\squishlist
    \item \textit{flat clipping} (\userdpfedavg{}) clips the concatenation of all the layers of model update with the L2 norm threshold $S$.
    \item \textit{per-layer clipping}~\cite{mcmahan2018learning} clips each layer of model update with the L2 norm threshold $S$.
    \item \textit{flat median clipping}~\cite{geyer2017differentially} uses the median\footnote{Strictly speaking, the median norm information can leak privacy \revise{and this slight looseness would extend to  robustness certifications which leverage the DP guarantee.} Nevertheless, the information leakage through the median is small, so median-clipping-based methods claimed to be DPFL in~\cite{geyer2017differentially}. 
} of the norms of clients’ model updates as the threshold $S$ for flat clipping.
    \item \textit{per-layer median clipping}~\cite{geyer2017differentially} uses the median of each layer’s norms of model updates as threshold $S$  for per-layer clipping. 
\squishend
\revise{We defer the detailed experimental parameters to \cref{app:exp_details_user_diff_dpfl_algos}.}

As shown in \cref{fig:userdp_ceracc_dpfl}, the models trained by different DPFL algorithms satisfying the same $\epsilon$ can have different certified robustness results. 
The flat clipping is able to certify the largest number of adversaries $k$ on \mnist{}; while on \cifar{} and \sent{}, the median clipping certifies the largest $k$ instead.
Moreover, flat clipping and per-layer clipping with the same $S$ lead to different certification results on all datasets, while the results of flat median clipping and per-layer median clipping are nearly identical on \mnist{} and \cifar{}.
\revise{We observe that no clipping mechanism is strictly better than others on all datasets. This is likely due to the significant difference in the norm of model updates when training on different datasets, which consequently affects the effectiveness of different clipping mechanisms, and thus the DP utility is dataset-dependent. Under the same DP guarantee $\epsilon$,  if one DPFL algorithm has higher utility and is more confident in predicting the ground-truth class, then it can increase the margin between the class confidences $F_{\Aclass}$ and $F_{\Bclass}$ in \cref{thm_pred_consist_k_client} and lead to a larger certified number of adversaries. }
Therefore, advanced DPFL protocols that have fewer clipping constraints or require less noise while achieving the same level of privacy  are favored  to improve \revise{both} utility and certified robustness. 
\revise{The practitioner can use our certifications to conduct offline comparisons of different DPFL algorithms under the same $\epsilon$, and better understand which DPFL algorithm provides better protection against poisoning attacks before real-world deployment. }

{

\begin{figure}
\newlength{\accnoniidheightc}
\settoheight{\accnoniidheightc}{\includegraphics[width=.48\linewidth]{figures/ccsfinal_plots/cer_acc_conf/cer_acc_mnist.pdf}}

\newcommand{\rowname}[1]%
{\rotatebox{90}{\makebox[\accnoniidheightc][c]{\footnotesize #1}}}

\centering

{
\renewcommand{\tabcolsep}{10pt}
\begin{subtable}[]{\linewidth}
\centering
\begin{tabular}{@{}p{5mm}@{}c@{}c@{}c@{}c@{}c@{}c@{}c@{}}
        & \makecell{{\scriptsize (a) \mnist{} i.i.d}}
        & \makecell{{\scriptsize (b) \cifar{} i.i.d}}
        \vspace{-1.7pt}\\
\rowname{\makecell{\ceracc}}&
\includegraphics[height=\accnoniidheightc]{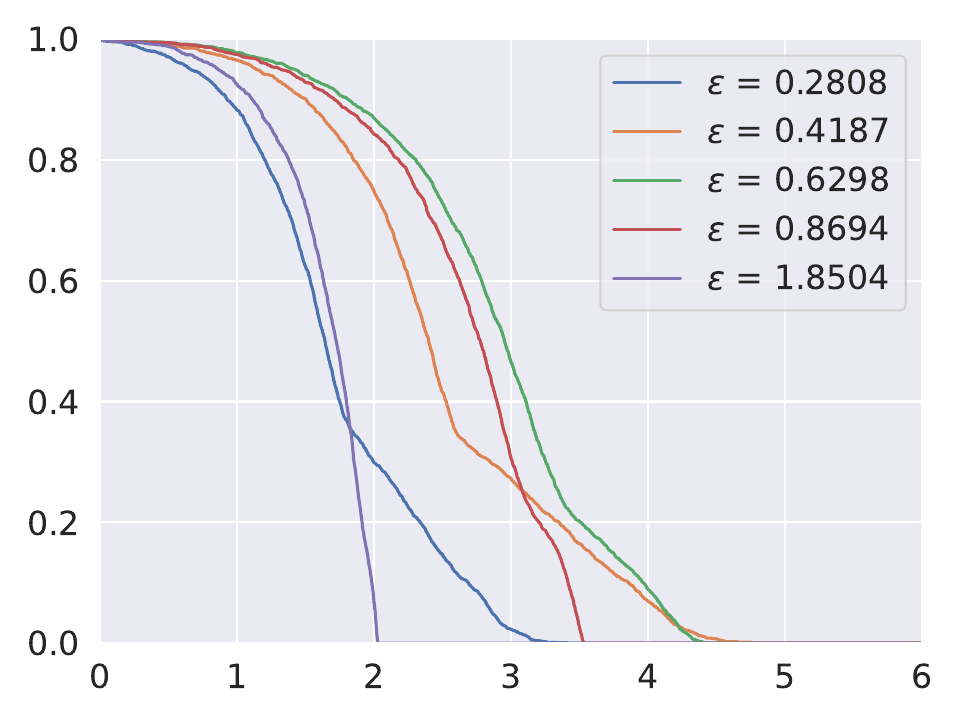}&
\includegraphics[height=\accnoniidheightc]{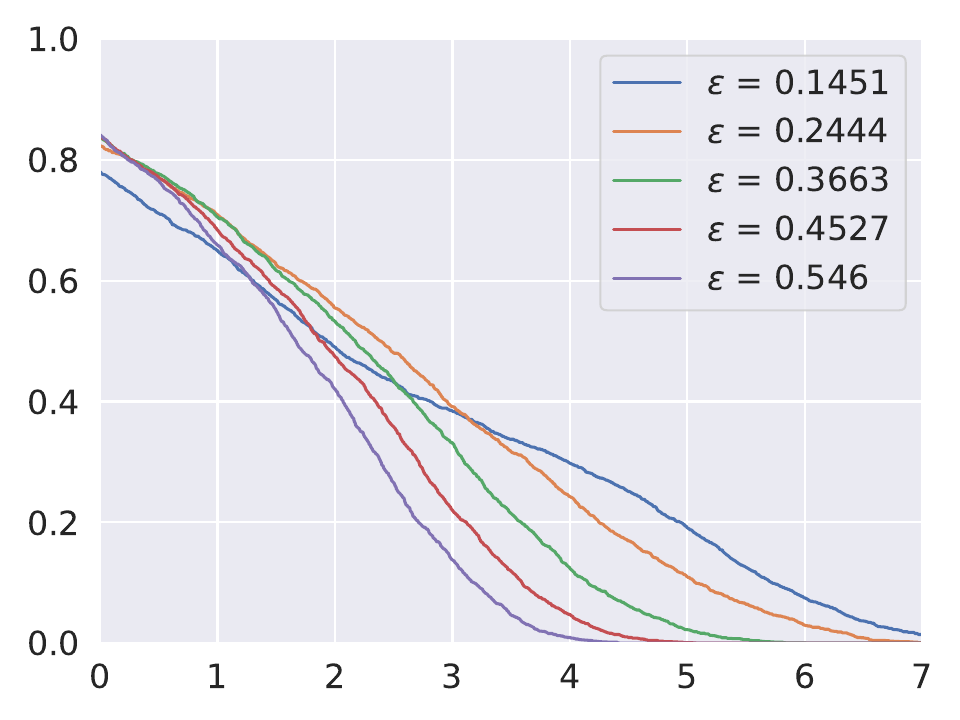}&
\\[-1.2ex]
& \makecell{\shortstack{\scriptsize (c) \mnist{}  $\operatorname{Dir}(1)$ }}
& \makecell{\shortstack{\scriptsize (d) \cifar{}    $\operatorname{Dir}(1)$ }}
 \vspace{-1.7pt}\\
\rowname{\makecell{\ceracc}}&
\includegraphics[height=\accnoniidheightc]{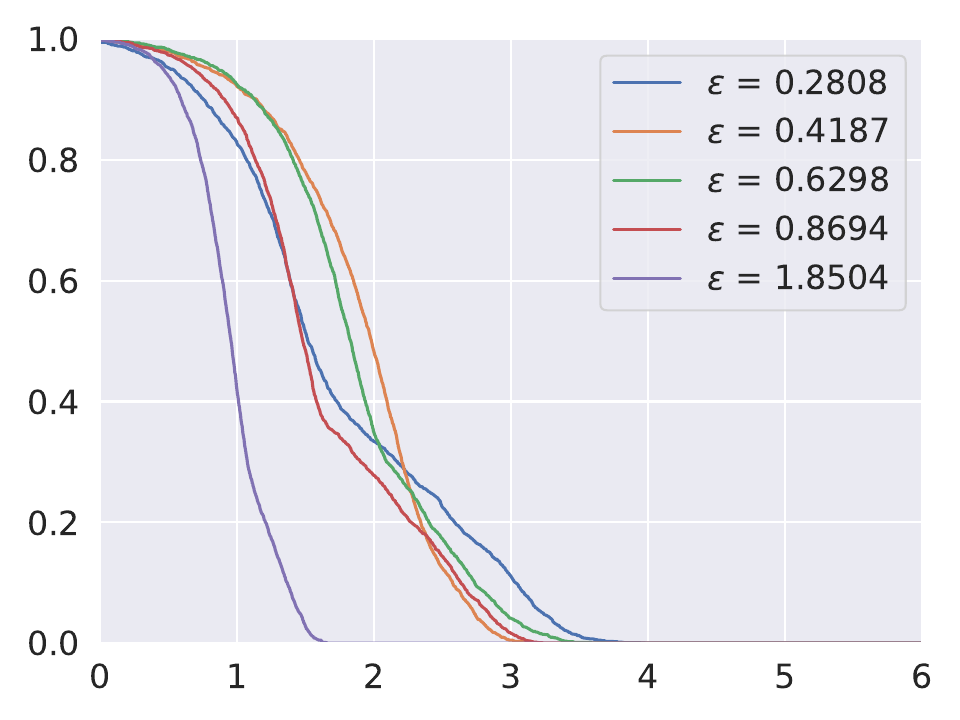}&
\includegraphics[height=\accnoniidheightc]{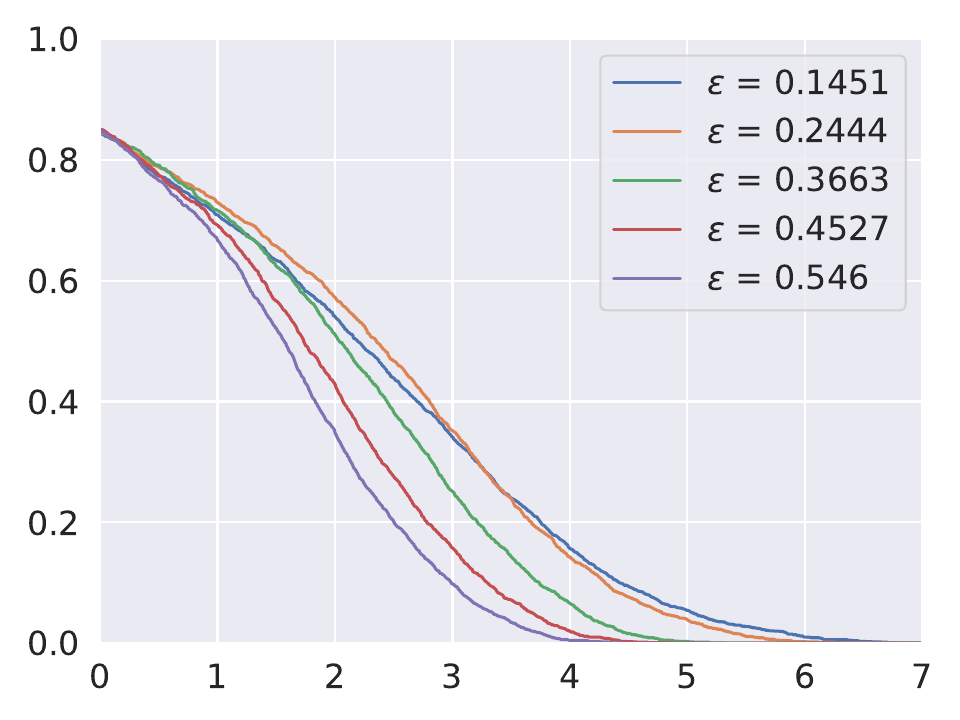}&
\\[-1.2ex]
& \makecell{\shortstack{\scriptsize (e) \mnist{}  $\operatorname{Dir}(0.5)$ }}
& \makecell{\shortstack{\scriptsize (f) \cifar{}    $\operatorname{Dir}(0.5)$ }}
 \vspace{-1.7pt}\\
\rowname{\makecell{\ceracc}}&
\includegraphics[height=\accnoniidheightc]{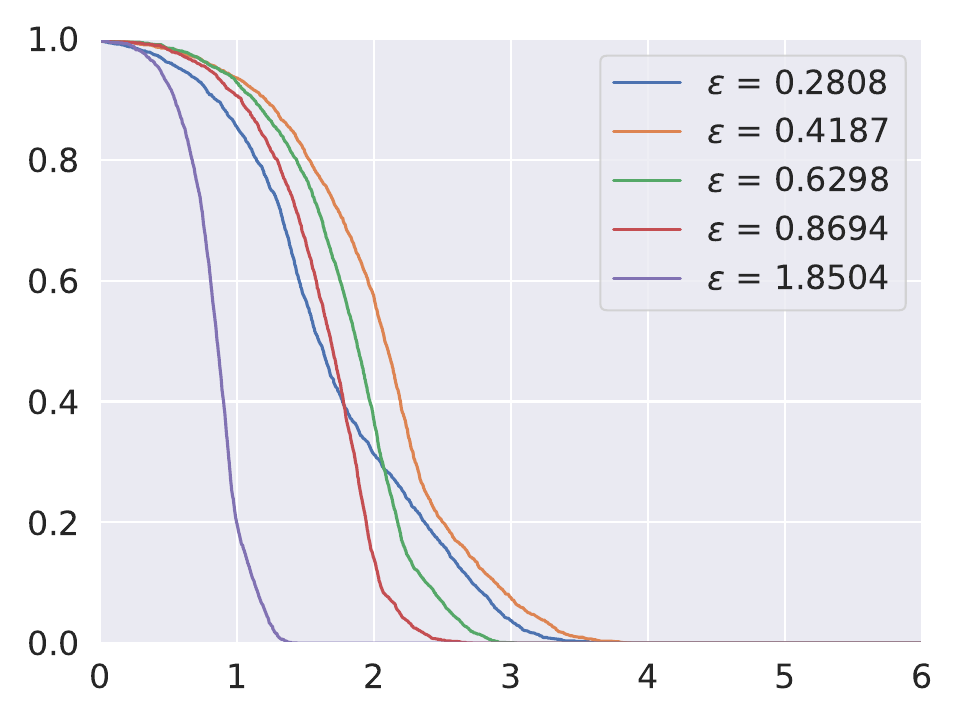}&
\includegraphics[height=\accnoniidheightc]{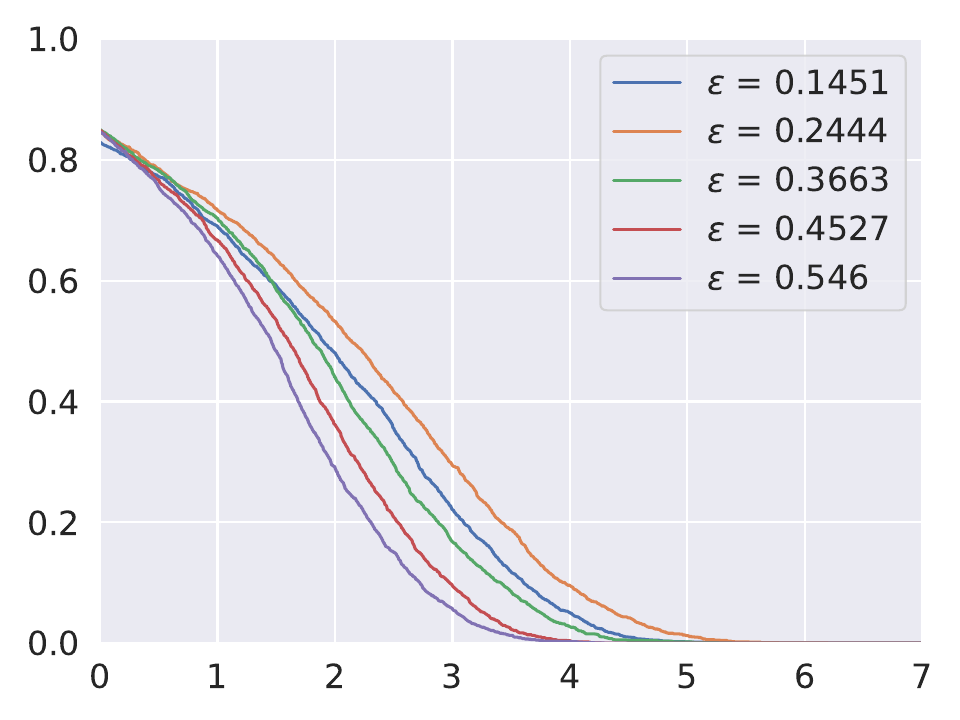}&
\\[-1.2ex]
& \makecell{{\footnotesize $k$}}& \makecell{{\footnotesize $k$}}
\end{tabular}
\end{subtable}
}
\vspace{-5mm}
\caption{Certified accuracy of \userdpfedavg{} under varying levels of data heterogeneity. We use Dirichlet distribution $\operatorname{Dir}(\alpha)$ to create FL heterogeneous data distributions, where smaller $\alpha$ indicates greater heterogeneity.}
\label{fig:cer_acc_noniid}
\vspace{-5mm}
\end{figure}

}

{
\begin{figure*}[!b]
\newlength{\costuserfigheightc}
\settoheight{\costuserfigheightc}{\includegraphics[width=0.17\linewidth]{figures/ccsfinal_plots/cer_acc_conf/cer_acc_mnist.pdf}}

\newcommand{\rowname}[1]%
{\rotatebox{90}{\makebox[\costuserfigheightc][c]{\scriptsize #1}}}

\centering

{
\renewcommand{\tabcolsep}{10pt}
\begin{subtable}[]{\linewidth}
\centering
\begin{tabular}{@{}p{5mm}@{}c@{}c@{}c@{}c@{}c@{}c@{}c@{}}
        & \makecell{\shortstack{\scriptsize (a) \mnist{} BKD  \scriptsize ($\epsilon=0.43$)}}
        & \makecell{\shortstack{\scriptsize (b) \cifar{} BKD  \scriptsize($\epsilon=0.53$)}}
        & \makecell{\shortstack{\scriptsize (c) \mnist{} LF \scriptsize ($\epsilon=0.40$)}}
        & \makecell{\shortstack{\scriptsize (d) \cifar{} LF \scriptsize ($\epsilon=0.59$)}}
         & \makecell{\shortstack{\scriptsize (e) \sent{} LF \scriptsize ($\epsilon=0.41$)}}
        \vspace{-1.7pt}\\
\rowname{\makecell{$J(D')$}}&
\includegraphics[height=\costuserfigheightc]{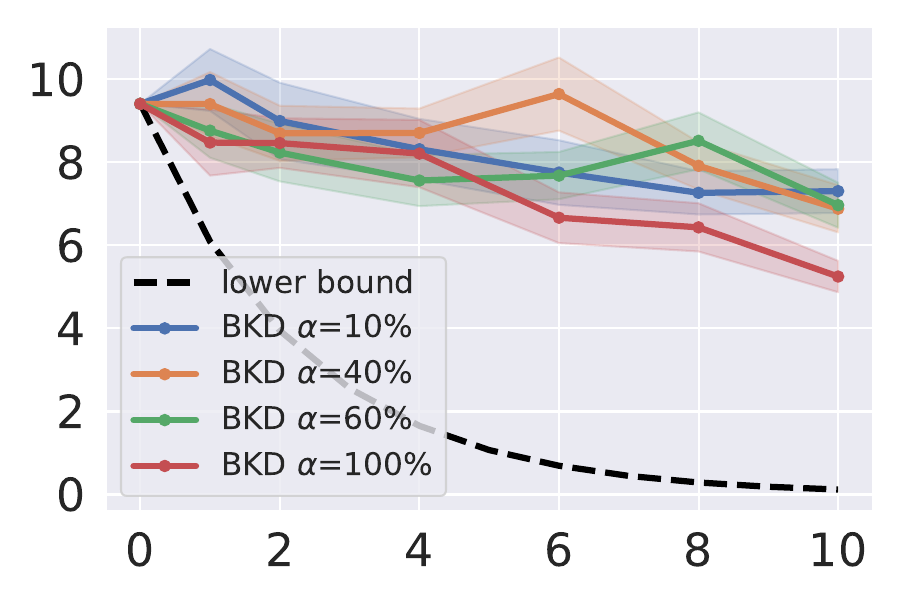}&
\includegraphics[height=\costuserfigheightc]{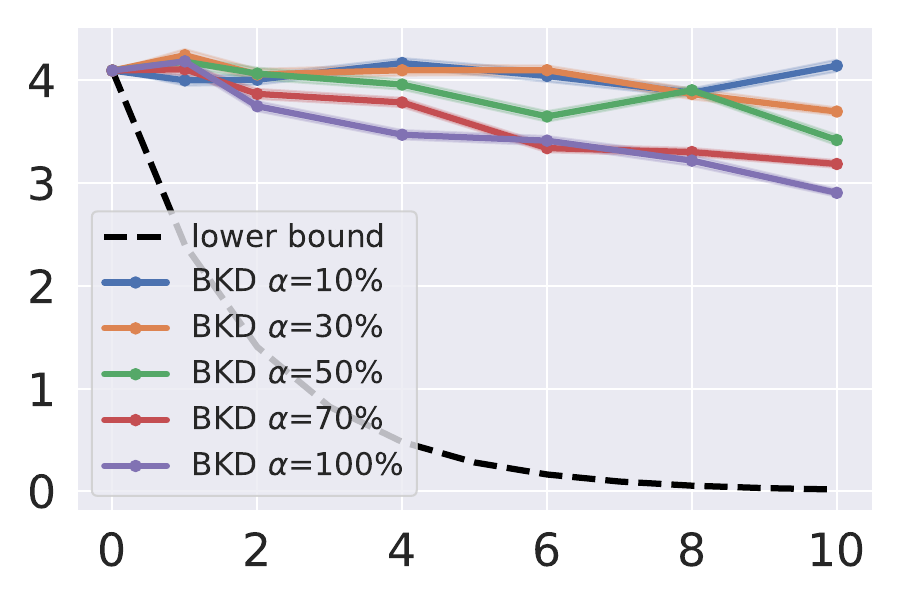}&
\includegraphics[height=\costuserfigheightc]{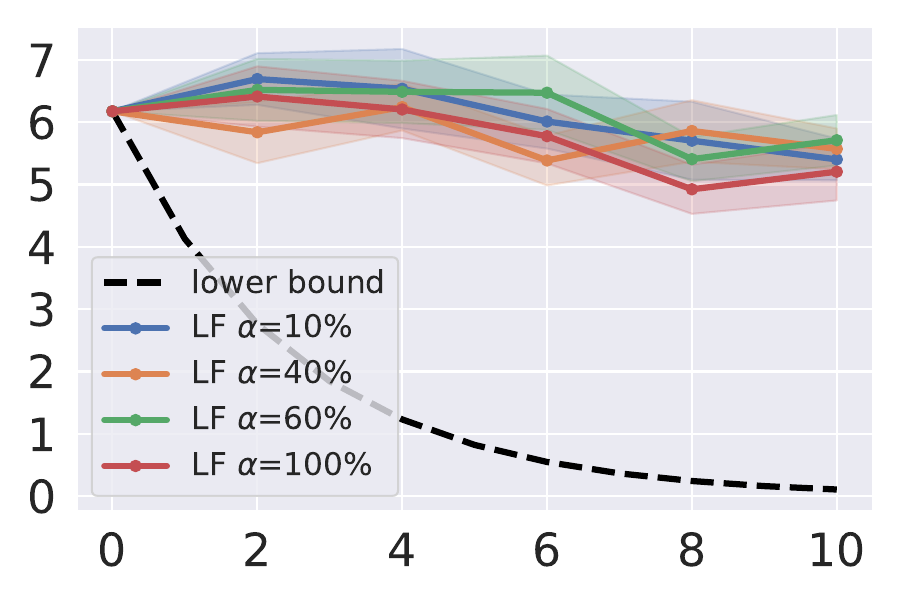}&
\includegraphics[height=\costuserfigheightc]{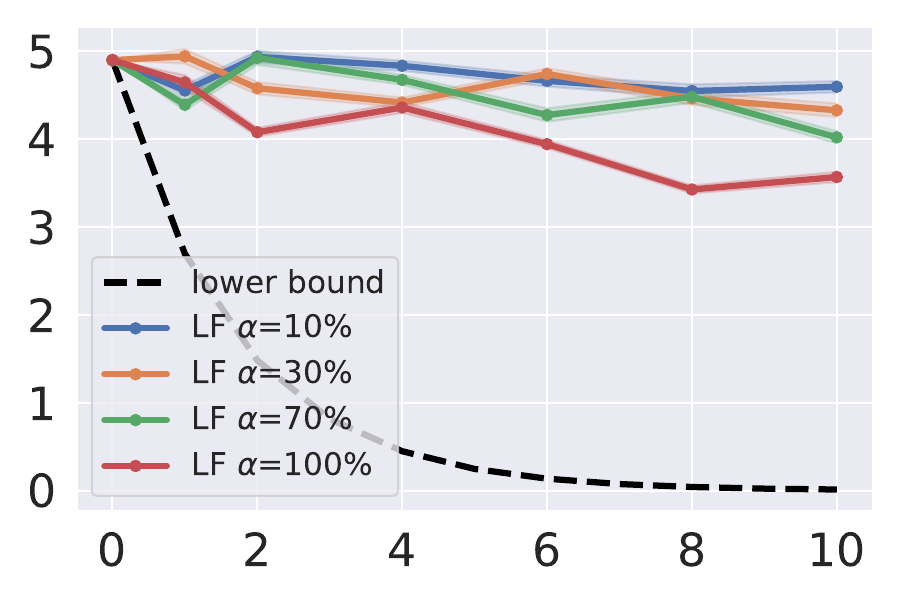}&
\includegraphics[height=\costuserfigheightc]{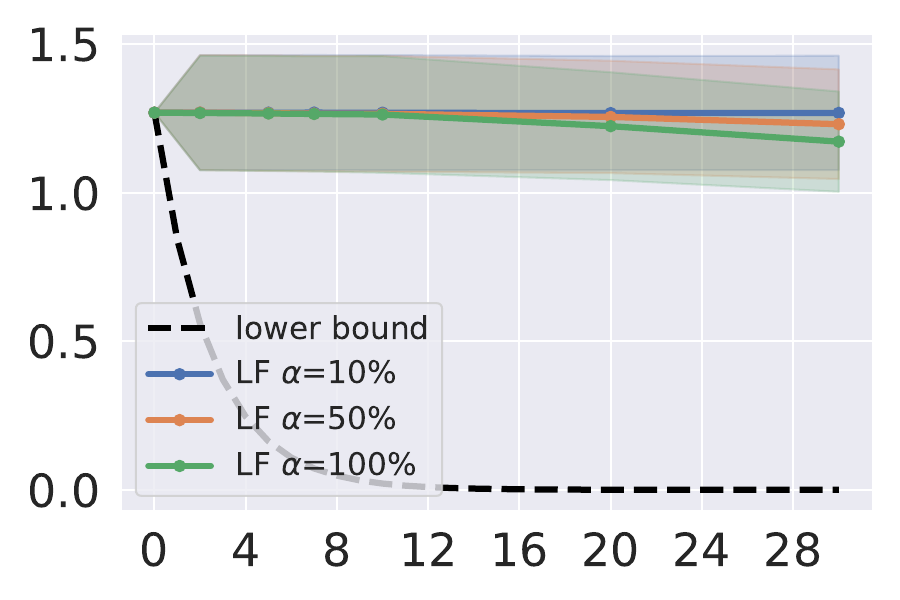}&
\\[-1.5ex]
        & \makecell{\shortstack{\scriptsize (f) \mnist{} BKD ($\epsilon=0.43$)}}
        & \makecell{\shortstack{\scriptsize (g) \cifar{} BKD ($\epsilon=0.53$)}}
        & \makecell{\shortstack{\scriptsize (h) \mnist{} LF ($\epsilon=0.40$)}}
        & \makecell{\shortstack{\scriptsize (i) \cifar{} LF ($\epsilon=0.59$)}}
         & \makecell{\shortstack{\scriptsize (j) \sent{} LF ($\epsilon=0.41$)}}
        \vspace{-1.7pt}\\
\rowname{\makecell{$J(D')$}}&
\includegraphics[height=\costuserfigheightc]{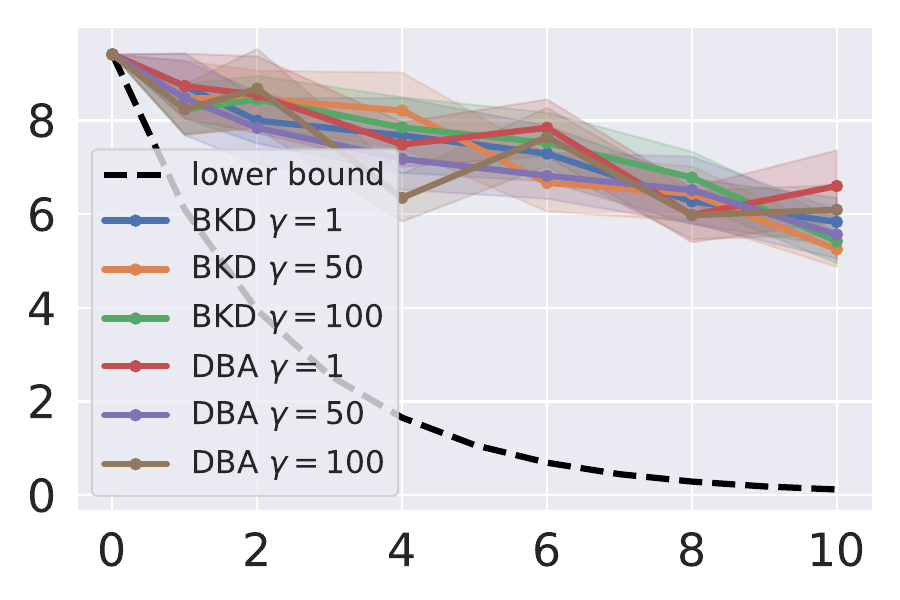}&
\includegraphics[height=\costuserfigheightc]{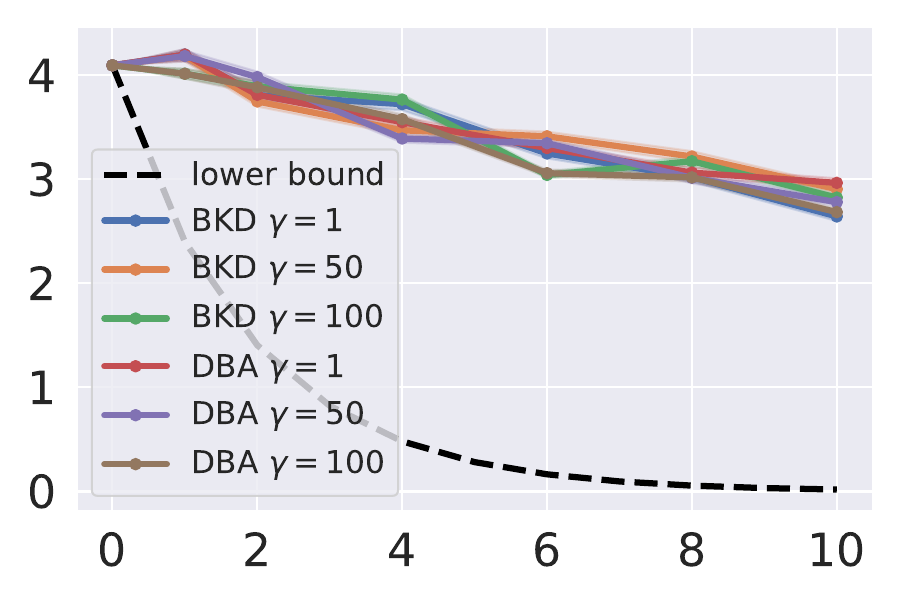}&
\includegraphics[height=\costuserfigheightc]{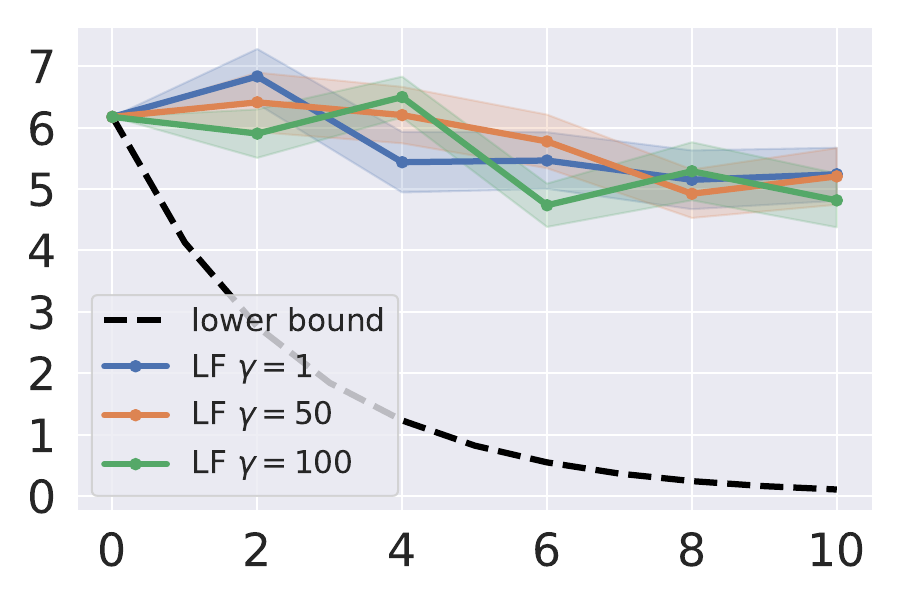}&
\includegraphics[height=\costuserfigheightc]{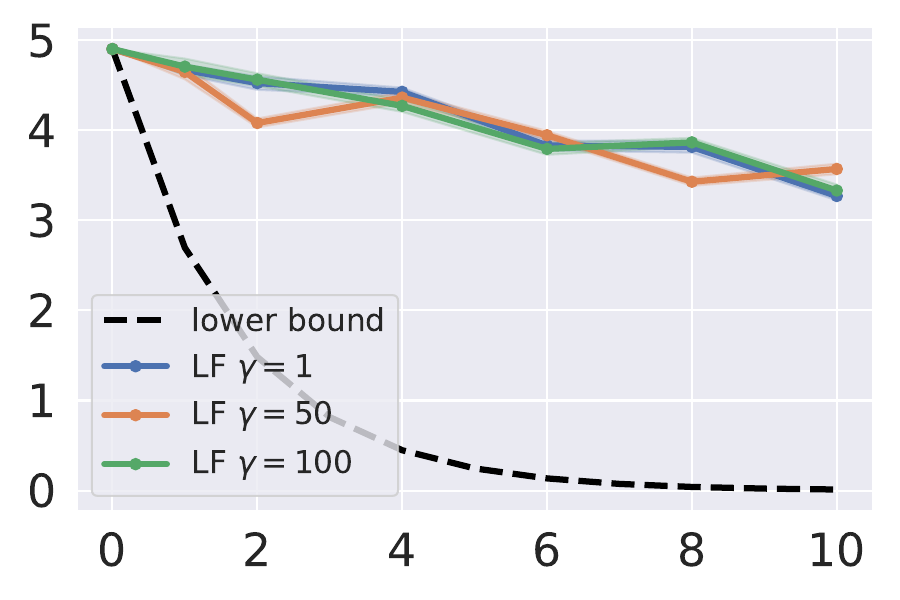}&
\includegraphics[height=\costuserfigheightc]{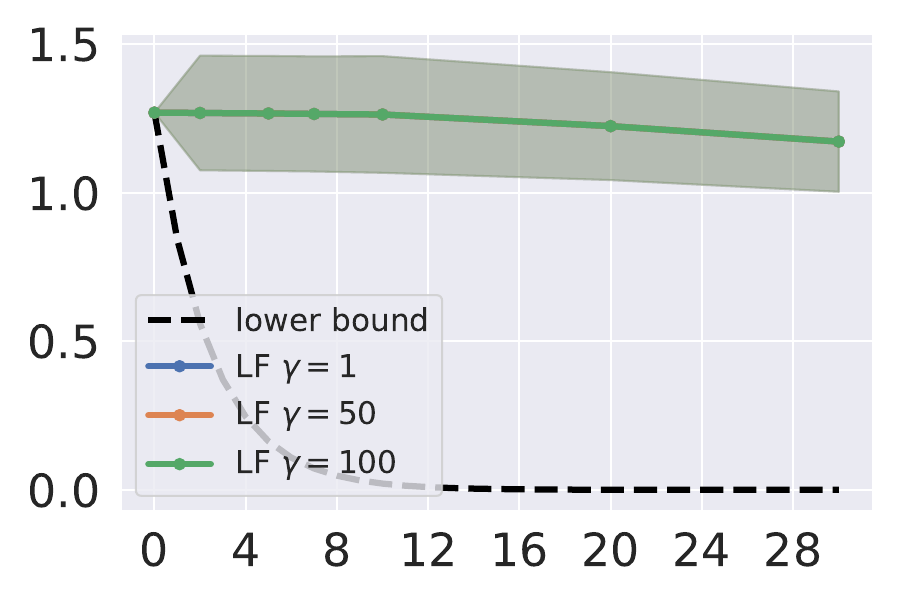}&
\\[-1.2ex]
& \makecell{{\scriptsize $k$}}& \makecell{{\scriptsize $k$}}& \makecell{\scriptsize $k$}& \makecell{\scriptsize $k$}& \makecell{\scriptsize $k$} 
\end{tabular}
\end{subtable}
}
\vspace{-5mm}
\caption{Certified \attackcost of \userdpfedavg{} given different $k$, under various attacks with different $\alpha$ or $\gamma$.}
\label{fig:user_k_gamma_alpha} 
\vspace{-5mm}
\end{figure*}

}

\vspace{-0.5em}
\subsubsection{Certified Accuracy under Different Data Heterogeneity Degrees.}
\label{sec:noniid_results}
Recent studies~\cite{yangprivatefl,noble2022differentially} show that DP makes the utility of the FL global model degraded more under heterogeneous data distributions among users, compared to the i.i.d data setting. Motivated by those findings, we study the impact of heterogeneity on the certified accuracy of DPFL models. 
We simulate varying levels of data heterogeneity on \mnist{} and \cifar{} using the Dirichlet distribution $\operatorname{Dir}(\alpha)$, which create FL heterogeneous data partitions with different local data sizes and label distributions for users, and smaller $\alpha$ indicates greater heterogeneity (more non-i.i.d). 

From the results in \cref{fig:cer_acc_noniid}, 
we find that \textbf{(1)} different non-i.i.d degrees have different optimal $\epsilon$ and the largest number of adversaries can be certified when $\epsilon$ is around 0.62, 0.28, 0.41 under the i.i.d, $\operatorname{Dir}(1)$, $\operatorname{Dir}(0.5)$ settings on MNIST, respectively. The optimal $\epsilon$ for CIFAR is around 0.14, 0.14, 0.24 under the i.i.d, $\operatorname{Dir}(1)$, $\operatorname{Dir}(0.5)$  settings, respectively. 
\textbf{(2)} Moreover, when FL data is more non-i.i.d, the largest number of adversaries that can be certified is smaller. This is mainly because the utility of the global model trained from the FedAvg-based DPFL degrades when FL data is more non-i.i.d, leading to a smaller confidence gap between  $F_{\Aclass}$ and  $ F_{\Bclass}$ in \cref{thm_pred_consist_k_client}.
This suggests that advanced FL algorithms designed for training more accurate FL models that tackle data heterogeneity issues can be applied to DPFL settings~\cite{noble2022differentially,liu2022privacy,yangprivatefl}. 
\revise{By doing so, it is possible to amplify the class confidences margin between  $F_{\Aclass}, F_{\Bclass}$ under non-i.i.d data and certify a larger $k$, subsequently improving both privacy-utility tradeoff and certified robustness. }

\begin{rev}

\vspace{-1mm}
\subsubsection{Empircal Robust Accuracy against State-of-the-Art Poisoning Attacks.}
\label{exp:poison}
In addition to the robustness certification, our DPFL certification process that produces prediction based on Equation~\ref{eq:prediction_certification}, 
 exhibits effective robustness \textit{empirically} against state-of-the-art poisoning attacks, even without theoretical guarantees.  
\cref{tb:empircal_cifar_sota_poisoning} in \cref{app:sota_poison_empirical} show that DPFL certification achieves high empirical robust accuracy on \cifar{} when $k=2,3,5,10$ against different attack strategies including STAT-OPT attacks~\cite{shejwalkar2021manipulating},  BKD and LF attacks boosted by the model replacement strategy~\cite{bagdasaryan2020backdoor,bhagoji2018analyzing}. Moreover, we see that the certified accuracy serves as the lower bound for the empirical robust accuracy. Details are deferred to \cref{app:sota_poison_empirical}.

{
\begin{figure*}
\newlength{\insmnistfigheightc}
\settoheight{\insmnistfigheightc}{\includegraphics[width=0.17\linewidth]{figures/ccsfinal_plots/cer_acc_conf/cer_acc_mnist.pdf}}

\newcommand{\rowname}[1]%
{\rotatebox{90}{\makebox[\insmnistfigheightc][c]{\footnotesize #1}}}

\centering

{
\renewcommand{\tabcolsep}{10pt}
\begin{subtable}[]{\linewidth}
\centering
\begin{tabular}{@{}p{2mm}@{}c@{}p{3mm}@{}c@{}c@{}c@{}c@{}c@{}c@{}}
        & \makecell{\shortstack{\footnotesize (a) \mnist{}}} &
        & \makecell{\shortstack{\footnotesize (b) \mnist{} BKD ($\epsilon=0.23$)}}
        & \makecell{\shortstack{\footnotesize (c) \mnist{} LF ($\epsilon=0.23$)}}
        & \makecell{\shortstack{\footnotesize (d) \mnist{} BKD ($k=10$)}}
         & \makecell{\shortstack{\footnotesize (e) \mnist{} LF  ($k=10$)}}
        \vspace{-1.7pt}\\
\rowname{\makecell{\ceracc}}&
\includegraphics[height=\insmnistfigheightc]{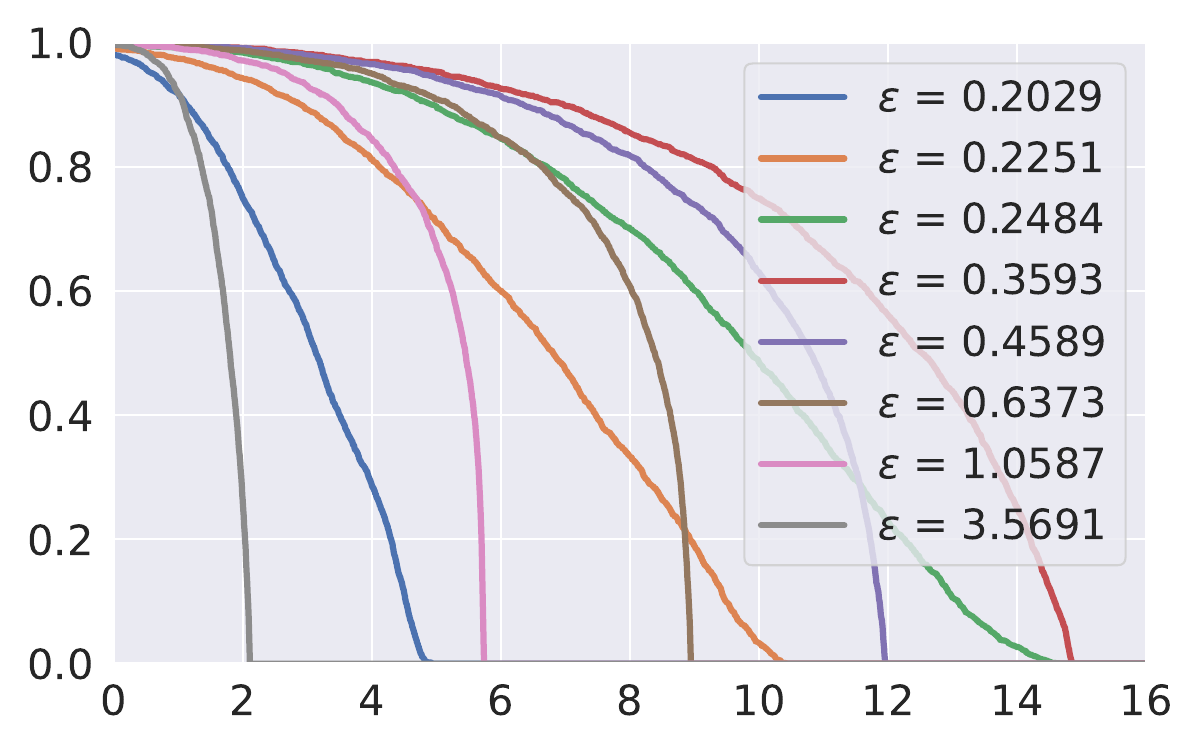}&
\rowname{\makecell{$J(D')$}}&
\includegraphics[height=\insmnistfigheightc]{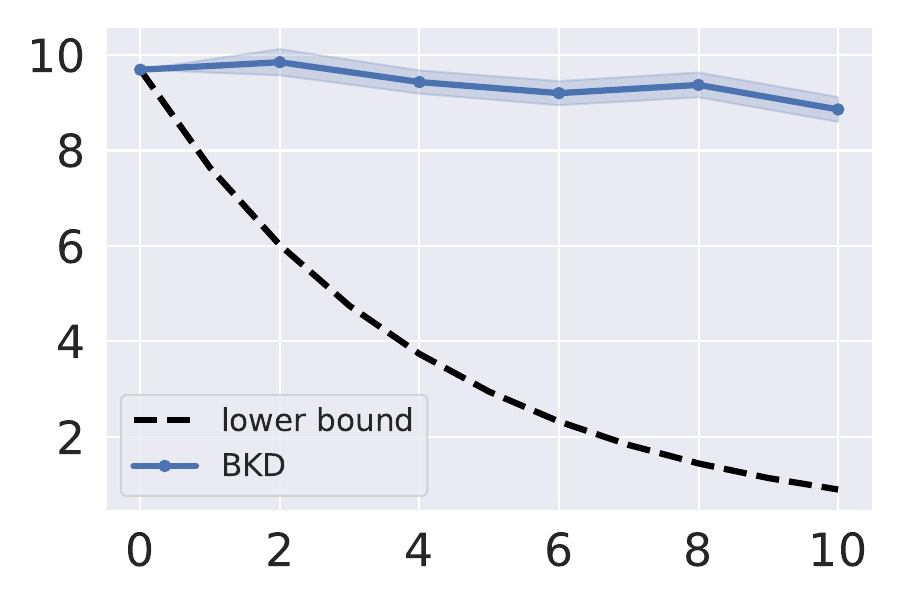}&
\includegraphics[height=\insmnistfigheightc]{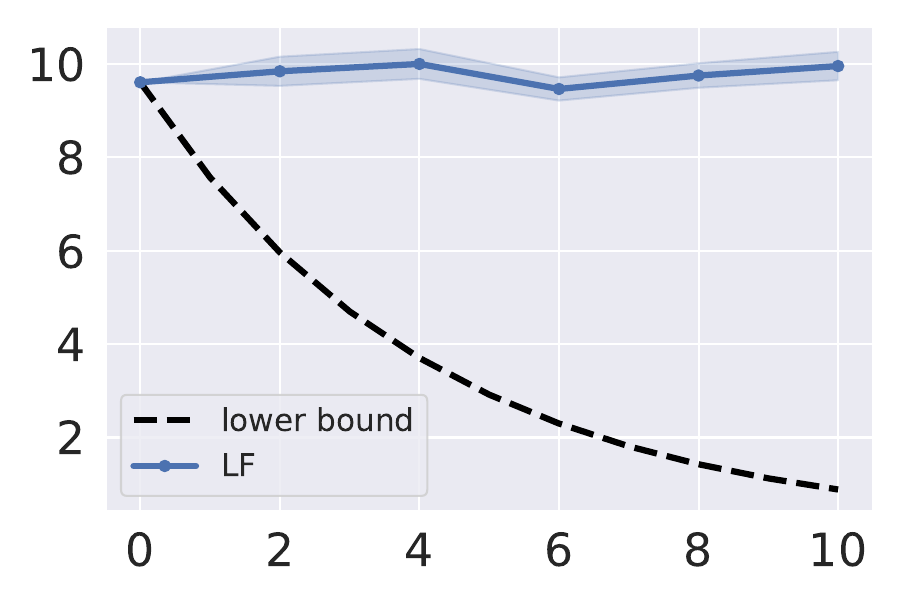}&
\includegraphics[height=\insmnistfigheightc]{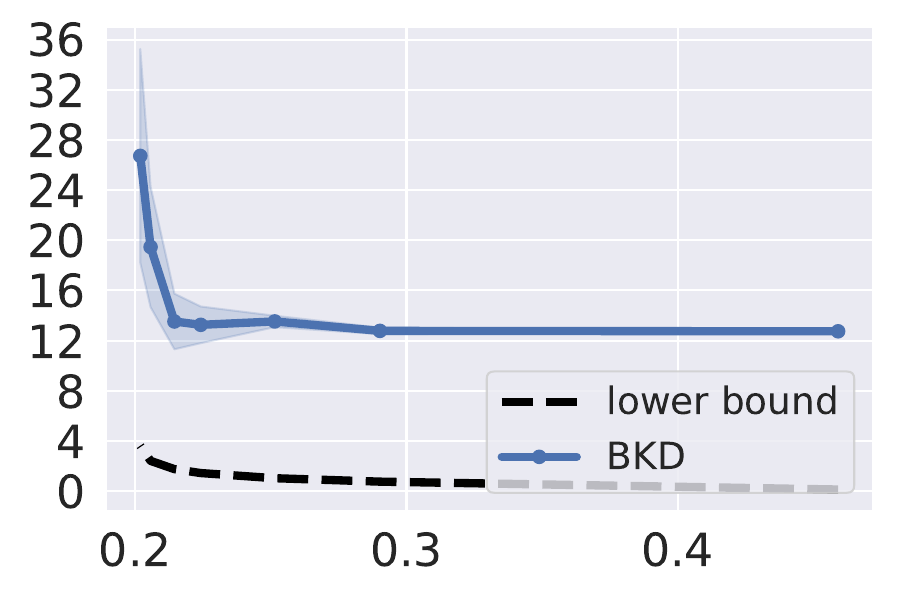}&
\includegraphics[height=\insmnistfigheightc]{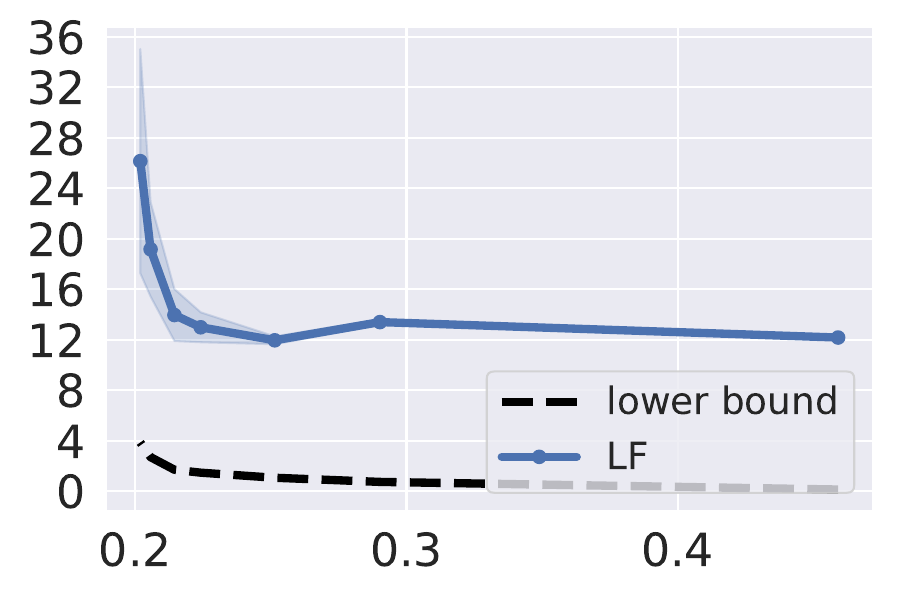}&
\\[-1ex]
& \makecell{{\footnotesize $k$}}&& \makecell{{\footnotesize  $k$}}& \makecell{\footnotesize{ $k$}}& \makecell{\footnotesize{ $\epsilon$}}& \makecell{\footnotesize{ $\epsilon$}} \\
\end{tabular}
\end{subtable}
}
\vspace{-2mm}
\caption{\small Certified accuracy (a) and certified \attackcost of \insdpfedavg{} on \mnist{} under different attacks given different $k$  (b-c)   and different $\epsilon$  (d-e).}
\label{fig:ins_ceracc_k_eps_mnist} 
\vspace{-2mm}
\end{figure*}

}
{
\begin{figure*}
\newlength{\inscifarfigheightc}
\settoheight{\inscifarfigheightc}{\includegraphics[width=0.17\linewidth]{figures/ccsfinal_plots/cer_acc_conf/cer_acc_mnist.pdf}}

\newcommand{\rowname}[1]%
{\rotatebox{90}{\makebox[\inscifarfigheightc][c]{\footnotesize #1}}}

\centering

{
\renewcommand{\tabcolsep}{10pt}
\begin{subtable}[]{\linewidth}
\centering
\begin{tabular}{@{}p{2mm}@{}c@{}p{3mm}@{}c@{}c@{}c@{}c@{}c@{}c@{}}
        & \makecell{\shortstack{\footnotesize (a) \cifar{}}} &
        & \makecell{\shortstack{\footnotesize (b) \cifar{} BKD ($\epsilon=0.31$)}}
        & \makecell{\shortstack{\footnotesize (c) \cifar{} LF ($\epsilon=0.31$)}}
        & \makecell{\shortstack{\footnotesize (d) \cifar{} BKD ($k=4$)}}
         & \makecell{\shortstack{\footnotesize (e) \cifar{} LF  ($k=4$)}}
        \vspace{-1.7pt}\\
\rowname{\makecell{\ceracc}}&
\includegraphics[height=\inscifarfigheightc]{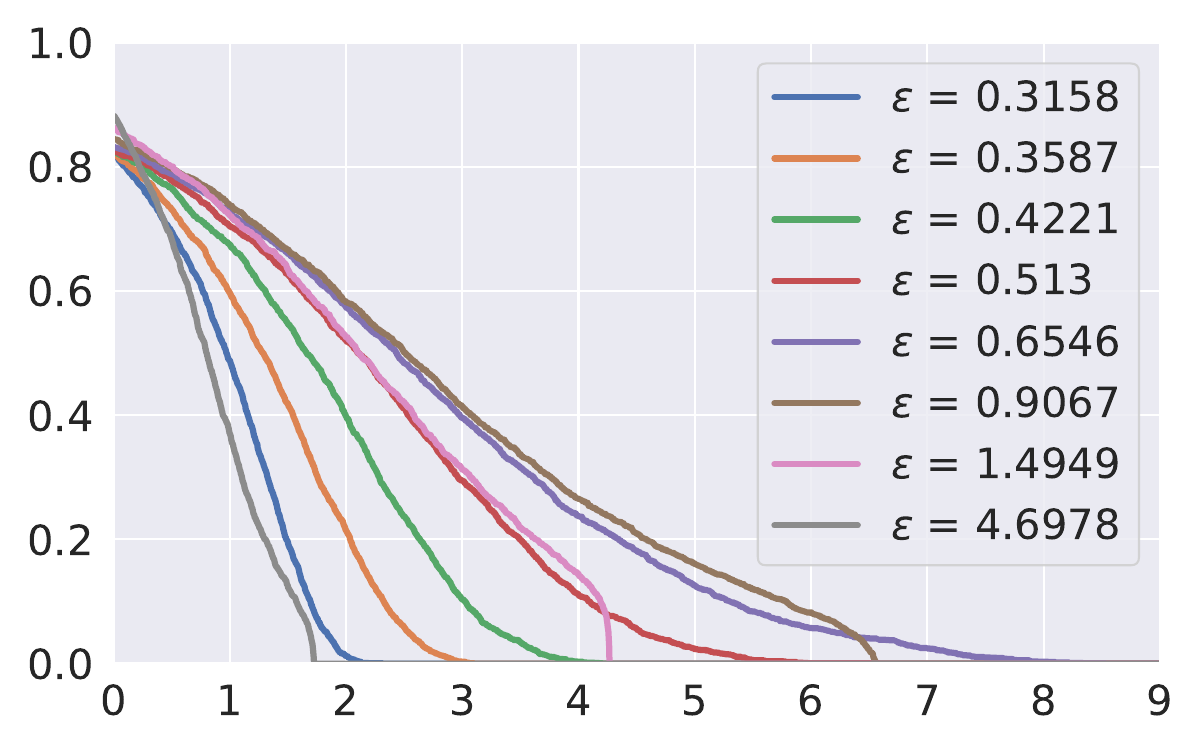}&
&
\includegraphics[height=\inscifarfigheightc]{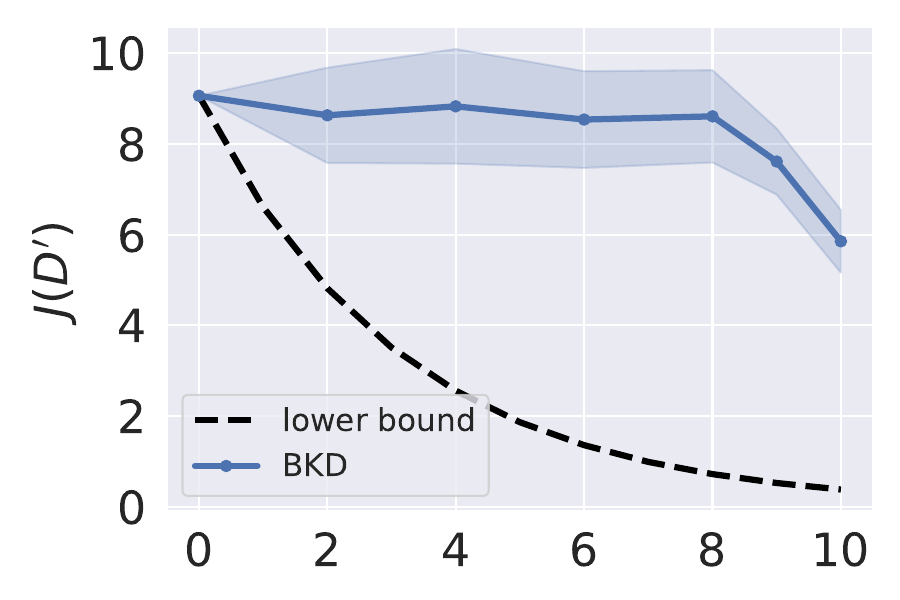}&
\includegraphics[height=\inscifarfigheightc]{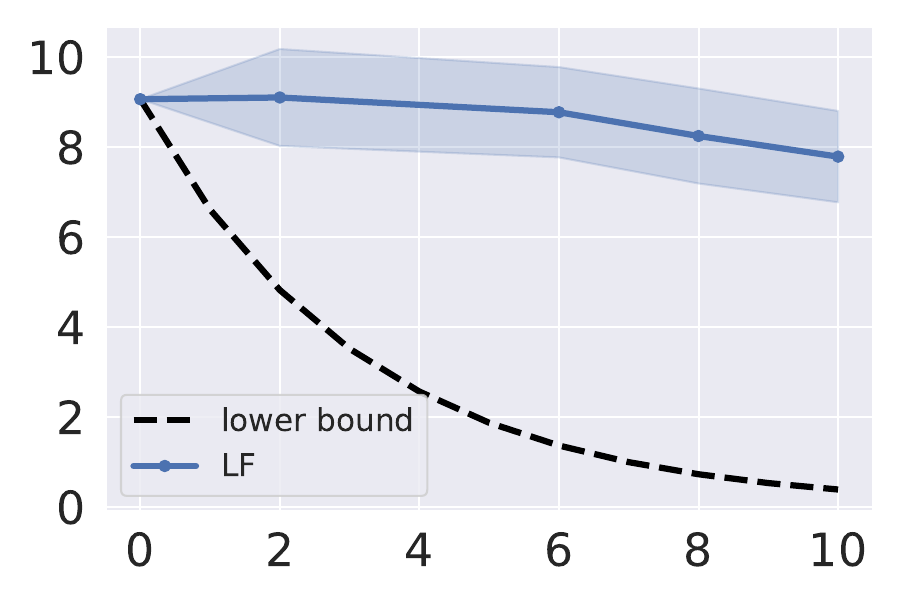}&
\includegraphics[height=\inscifarfigheightc]{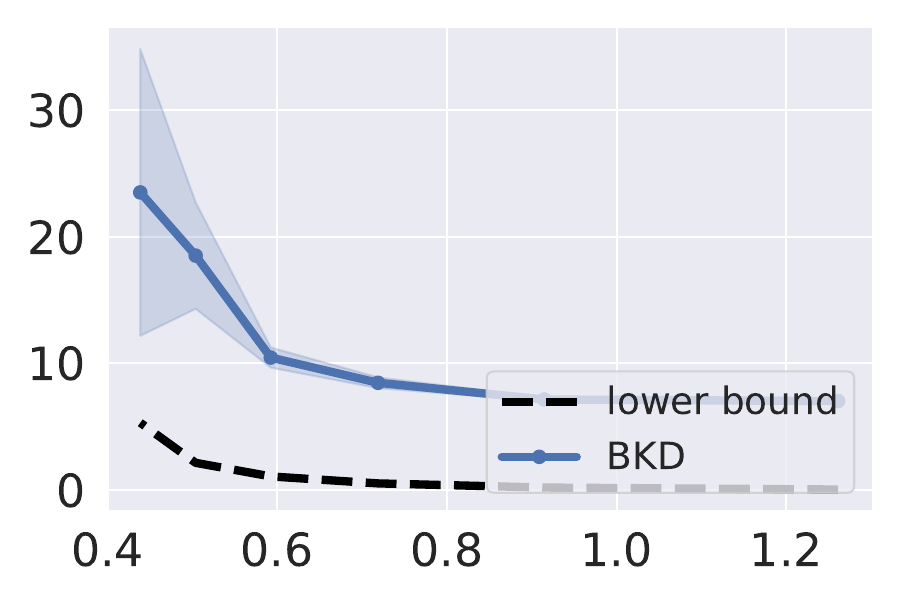}&
\includegraphics[height=\inscifarfigheightc]{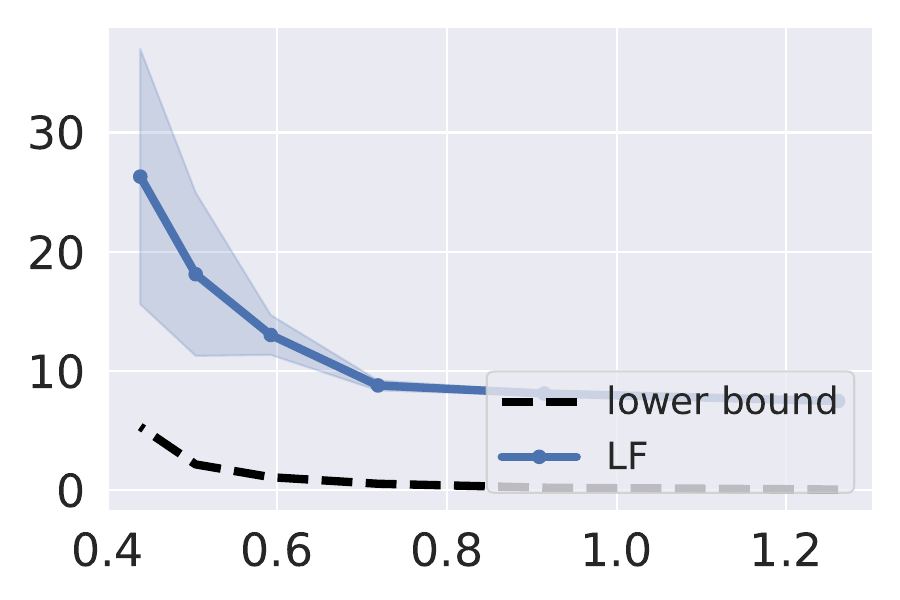}&
\\[-1ex]
& \makecell{{\footnotesize $k$}}&& \makecell{{\footnotesize  $k$}}& \makecell{\footnotesize{ $k$}}& \makecell{\footnotesize{ $\epsilon$}}& \makecell{\footnotesize{ $\epsilon$}} \\
\end{tabular}
\end{subtable}
}
\vspace{-2mm}
\caption{\small Certified accuracy (a) and certified \attackcost of \insdpfedavg{} on \cifar{} given different $k$ (b-c) and different $\epsilon$ (d-e).}
\label{fig:ins_ceracc_k_eps_cifar} 
\vspace{-2mm}
\end{figure*}

}
{
\begin{figure}
\newlength{\costuserepsfigheightc}
\settoheight{\costuserepsfigheightc}{\includegraphics[width=0.44\linewidth]{figures/ccsfinal_plots/cer_acc_conf/cer_acc_mnist.pdf}}

\newcommand{\rowname}[1]%
{\rotatebox{90}{\makebox[\costuserepsfigheightc][c]{\scriptsize #1}}}

\centering

{
\renewcommand{\tabcolsep}{10pt}
\begin{subtable}[]{\linewidth}
\centering
\begin{tabular}{@{}p{3mm}@{}c@{}@{}p{3mm}@{}c@{}c@{}c@{}c@{}c@{}c@{}}
        & \makecell{\shortstack{\scriptsize (a) \mnist{} ($k=4$) }}&
        & \makecell{\shortstack{\scriptsize (d) \mnist{} }}
        \vspace{-1.7pt}\\
\rowname{\makecell{\scriptsize $J(D')$}}&
\includegraphics[height=\costuserepsfigheightc]{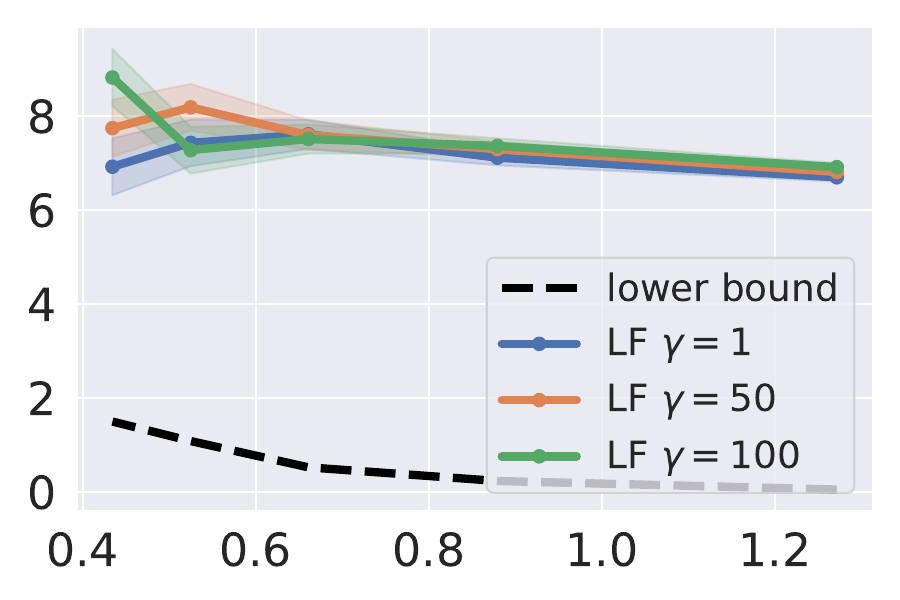}&
\rowname{\makecell{\scriptsize $k$}}&
\includegraphics[height=\costuserepsfigheightc]{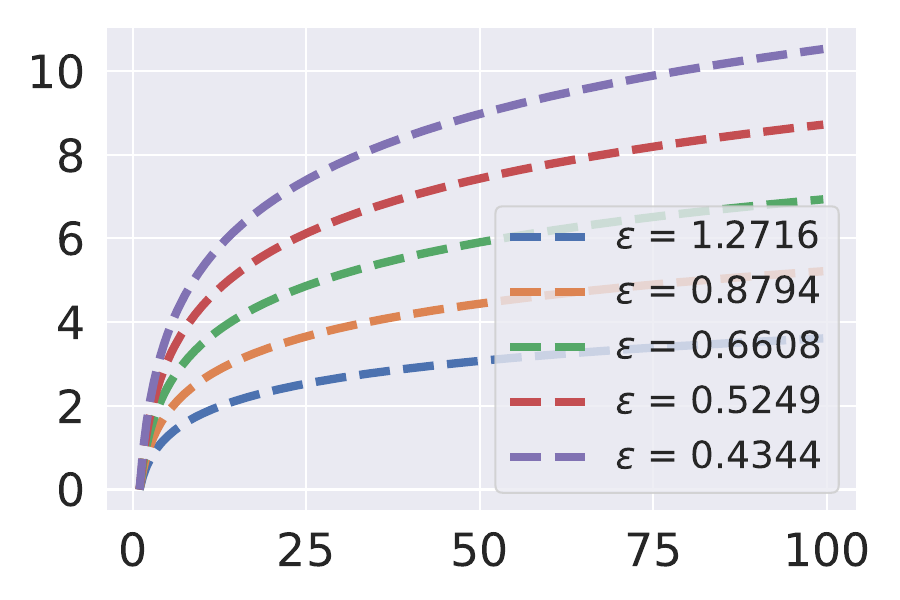}
\\[-2.5ex]
& \makecell{{\scriptsize $\epsilon$}}& &\makecell{{\scriptsize $\tau$}} \\[-1ex] 
 & \makecell{\shortstack{\scriptsize (b) \cifar{} ($k=4$) }} &
        & \makecell{\shortstack{\scriptsize (e) \cifar{} }}
        \vspace{-2.5pt}\\
\rowname{\makecell{\scriptsize $J(D')$}}&
\includegraphics[height=\costuserepsfigheightc]{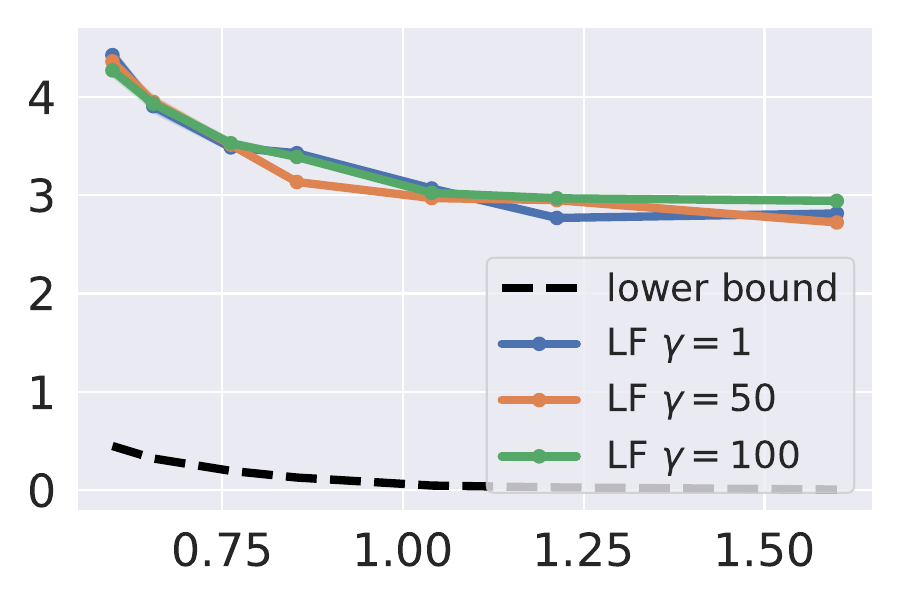}&
\rowname{\makecell{\scriptsize $k$}}&
\includegraphics[height=\costuserepsfigheightc]{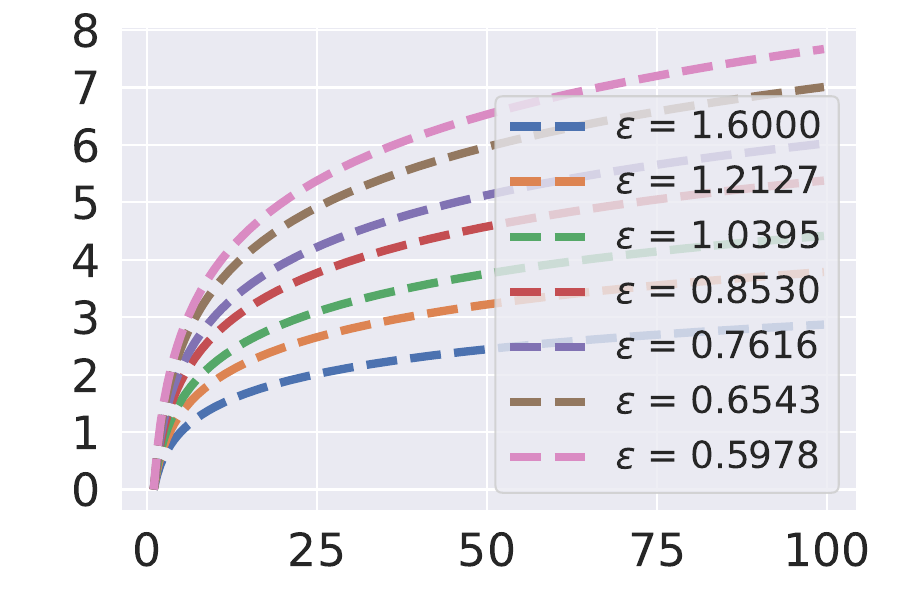}
\\[-2.5ex]
& \makecell{{\scriptsize $\epsilon$}}& &\makecell{{\scriptsize $\tau$}} \\[-1ex]  
 & \makecell{\shortstack{\scriptsize (c) \sent{} ($k=4$) }} &
        & \makecell{\shortstack{\scriptsize (f) \sent{} }}
        \vspace{-1.7pt}\\
\rowname{\makecell{\scriptsize $J(D')$}}&
\includegraphics[height=\costuserepsfigheightc]{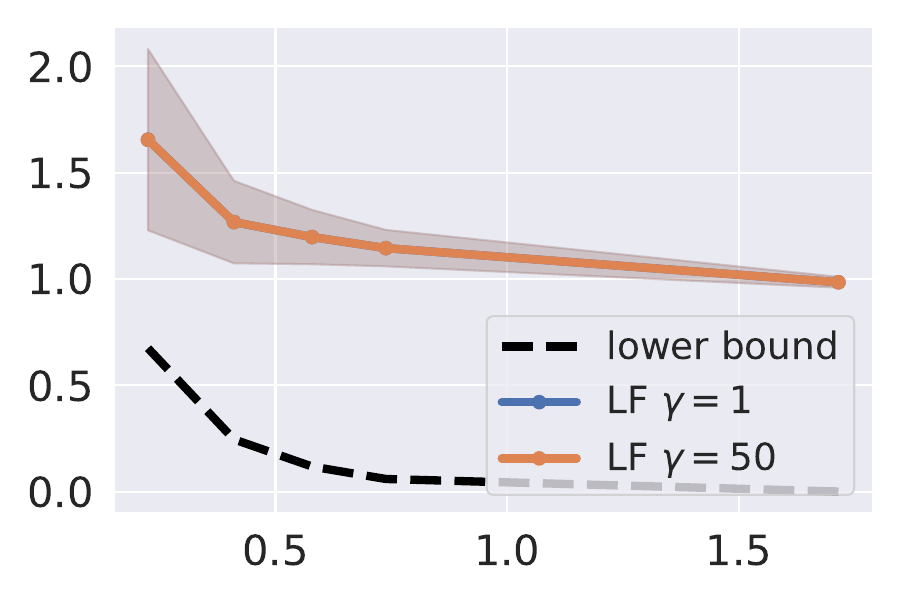}&
\rowname{\makecell{\scriptsize $k$}}&
\includegraphics[height=\costuserepsfigheightc]{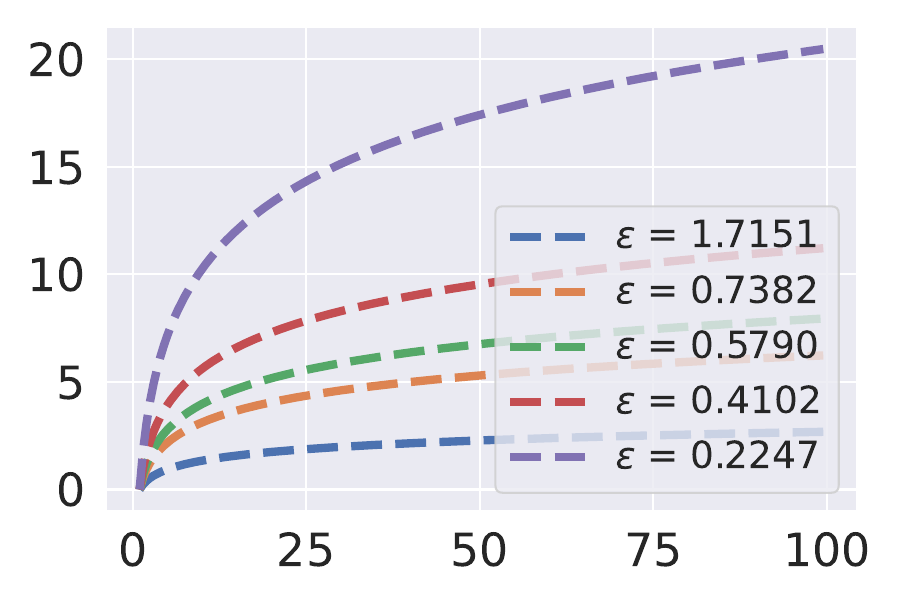}
\\[-2.5ex]
& \makecell{{\scriptsize $\epsilon$}}& &\makecell{{\scriptsize $\tau$}} \\[-0ex] 
\end{tabular}
\end{subtable}
}
\vspace{-3mm}
\caption{\small Certified \attackcost of  \userdpfedavg{}  with different $\epsilon$  (a-c), and the lower bound of $k$ given different  $\epsilon$ under different attack effectiveness $\tau$ (d-f).}
\label{fig:user_eps_tau} 
\vspace{-2mm}
\end{figure}

}

\vspace{-1mm}
\subsubsection{Comparison to Empirical FL Defenses.}
\label{exp:defense}
Another line of research is to develop empirical defenses such as robust aggregation mechanisms
~\cite{nguyen2022flame,fung2020limitations,blanchard2017machinekrum,el2018hidden}
to detect and remove malicious users. 
Compared to empirical FL defenses, our work provides robustness \textit{certifications}, while existing studies only offer \textit{empirical} robustness. One key advantage of our analysis is that our robustness certifications provide lower bounds for model accuracy or attack inefficiency against constrained attacks, and such certification is agnostic to actual attack strategies, which means  there are  no future attacks that can break the certification as long as the $k$  is within the certified range. Conversely, empirical countermeasures are typically designed against specific types of attacks, leaving them potentially vulnerable to stronger or adaptive attacks in unknown environments~\cite{wang2020attackthetails,fang2020local}.
Moreover, our certifications are general and uncover the inherent relations between DPFL and certified robustness, and DPFL algorithms with better utility or tighter privacy accountants can further enhance the certification results. 

As existing FL defenses do not provide robustness guarantees and hence cannot be directly compared under our certified criteria, we compare the \textit{empirical} robust accuracy of our certification method with \textit{six} FL robust aggregations, including Krum~\cite{blanchard2017machinekrum}, Multi-krum~\cite{blanchard2017machinekrum}, Trimmed-mean~\cite{yin2018byzantine}, Median~\cite{yin2018byzantine}, Bulyan~\cite{el2018hidden}, RFA~\cite{pillutla2019robustrfa}.  \cref{tb:empircal_cifar_sota_poisoning} in \cref{app:sota_fl_defenses_empirical}  shows that our certification method achieves similar and even higher robust accuracy than empirical defenses under the state-of-the-art poisoning attacks on \cifar{}, while our approach can further provide robustness guarantees under different criteria.  We defer detailed results and discussion to \cref{app:sota_fl_defenses_empirical}.

Moreover, it is worth noting that our certifications still hold when DPFL is combined with other empirical defense strategies. Theoretically, in the presence of such defensive mechanisms, the  $(\epsilon,\delta)$ privacy guarantee holds due to the post-processing property of DP, and therefore certified robustness guarantee given 
$(\epsilon,\delta)$-DP still holds.  
Combining DPFL with other robust aggregations would further enhance the empirical robustness, which remains an interesting future direction.

\end{rev}

{
\begin{figure*}
\newlength{\accfigheightc}
\settoheight{\accfigheightc}{\includegraphics[width=.21\linewidth]{figures/ccsfinal_plots/cer_acc_conf/cer_acc_mnist.pdf}}

\newcommand{\rowname}[1]%
{\rotatebox{90}{\makebox[\accfigheightc][c]{\footnotesize #1}}}

\centering

{
\renewcommand{\tabcolsep}{10pt}
\begin{subtable}[]{\linewidth}
\centering
\begin{tabular}{@{}p{5mm}@{}c@{}c@{}c@{}c@{}c@{}c@{}c@{}}
        & \makecell{{\footnotesize (a) \mnist}}
        & \makecell{{\footnotesize (b) \cifar}}
        & \makecell{{\footnotesize (c) \mnist }}
        & \makecell{{\footnotesize (d) \cifar}}
        \vspace{-1.7pt}\\
\rowname{\makecell{\ceracc}}&
\includegraphics[height=\accfigheightc]{figures/ccsfinal_plots/cer_acc_conf/cer_acc_mnist.pdf}&
\includegraphics[height=\accfigheightc]{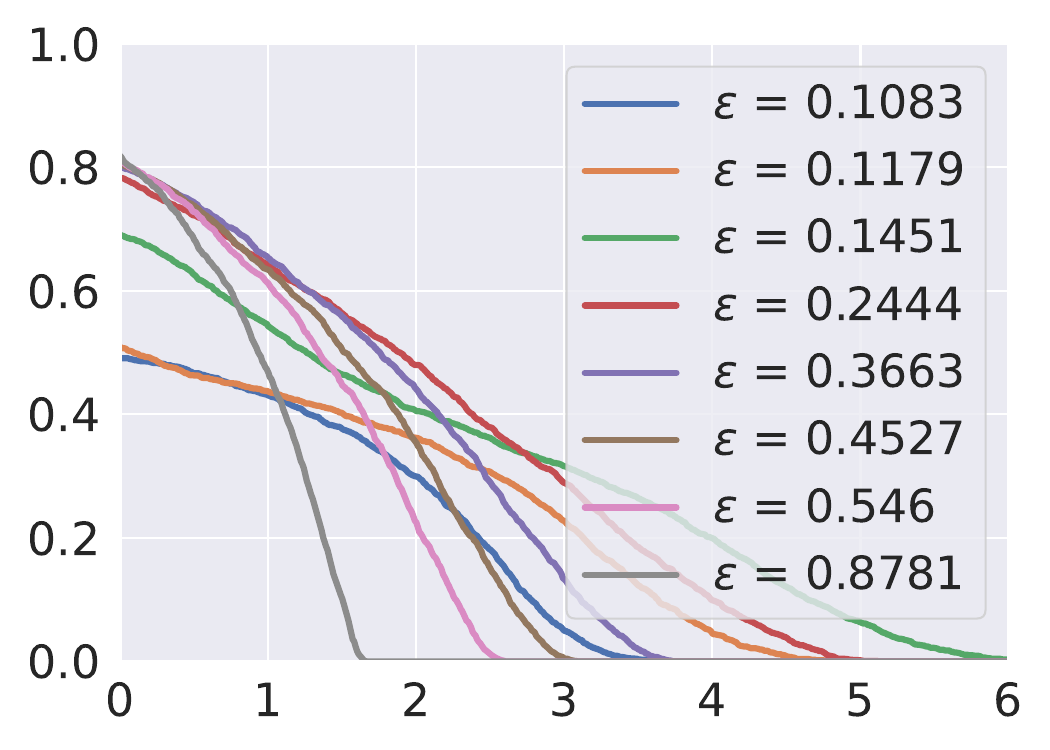}&
\includegraphics[height=\accfigheightc]{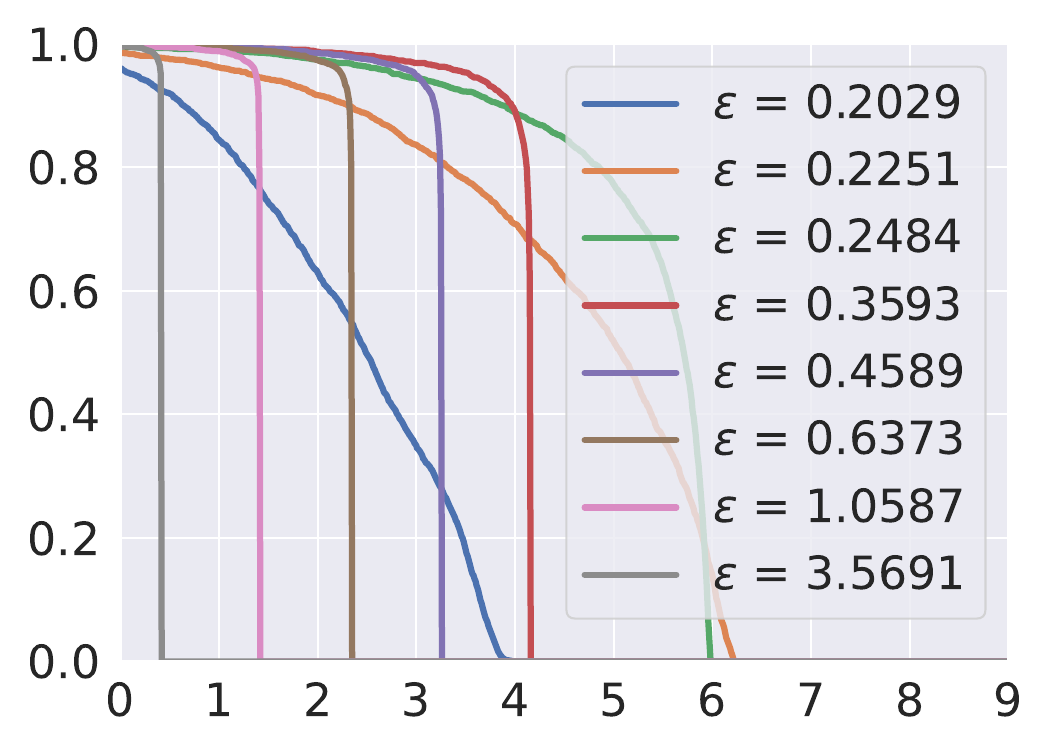}&
\includegraphics[height=\accfigheightc]{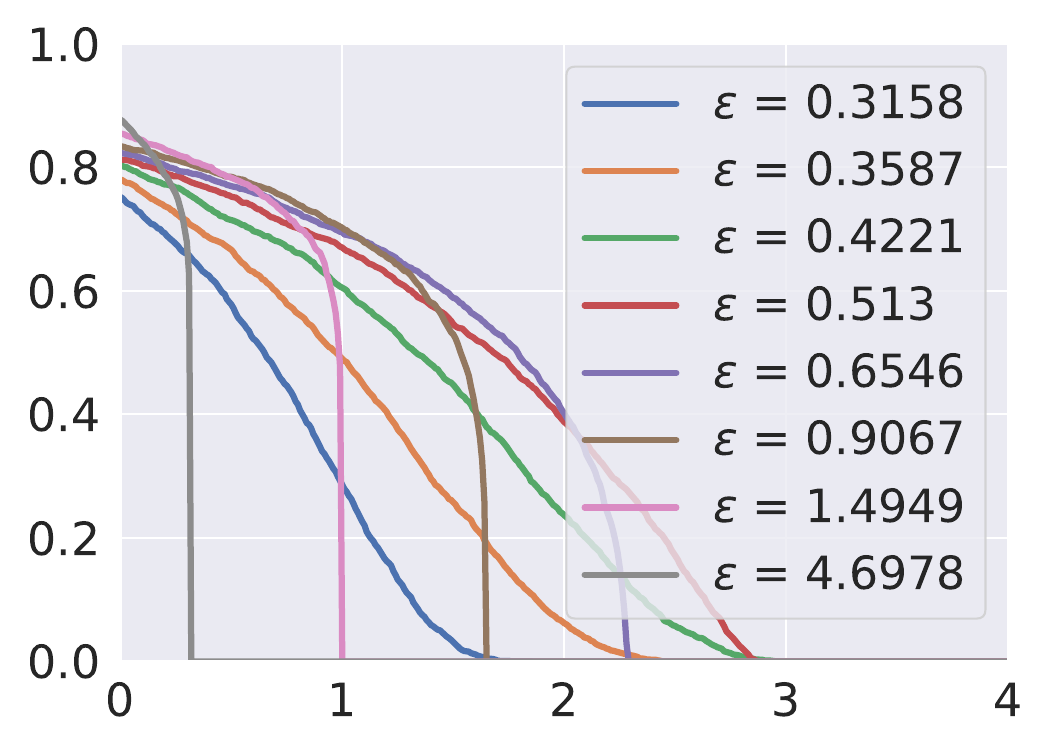}\\[-2ex]
& \makecell{{\footnotesize $k$}}& \makecell{{\footnotesize $k$}}& \makecell{\footnotesize $k$}& \makecell{\footnotesize $k$} \\
\end{tabular}
\end{subtable}
}
\vspace{-5mm}
\caption{\small Certified accuracy under $99\%$ confidence of FL satisfying user-level DP (a,b), and instance-level DP (c,d).}
\label{fig:cer_acc_conf}
\vspace{-4mm}
\end{figure*}

}

{
\begin{figure}[t]
\newlength{\fedsgdfigheightc}
\settoheight{\fedsgdfigheightc}{\includegraphics[width=0.45\linewidth]{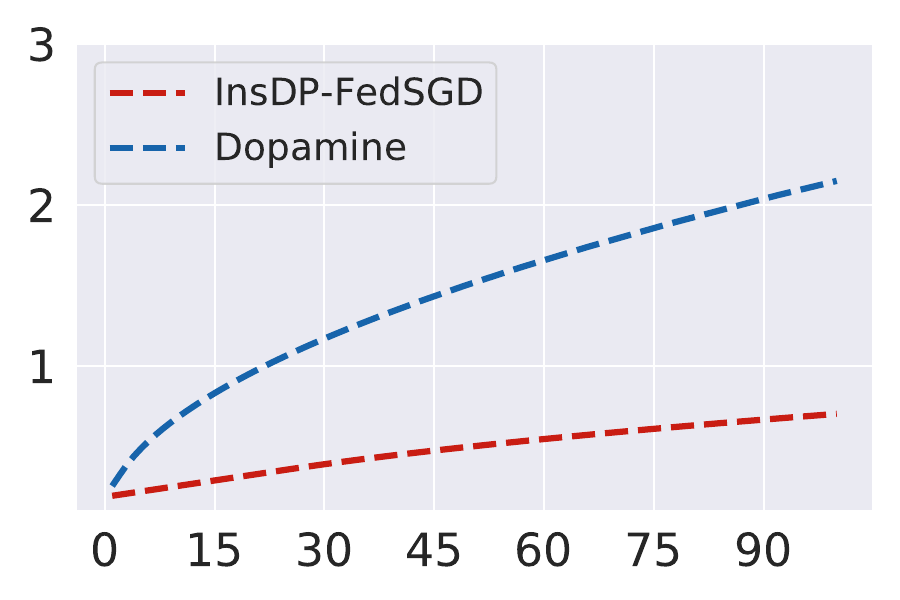}}

\newcommand{\rowname}[1]%
{\rotatebox{90}{\makebox[\fedsgdfigheightc][c]{\footnotesize #1}}}

\centering

{
\renewcommand{\tabcolsep}{10pt}
\begin{subtable}[]{\linewidth}
\centering
\begin{tabular}{@{}p{5mm}@{}c@{}c@{}c@{}c@{}c@{}c@{}c@{}}
        & \makecell{\shortstack{\footnotesize (a)~$ q=10/30$ }}
             & \makecell{\shortstack{\footnotesize (b)~$ q=20/30$}}
        \vspace{-1.7pt}\\
\rowname{\makecell{$\epsilon$}}&
\includegraphics[height=\fedsgdfigheightc]{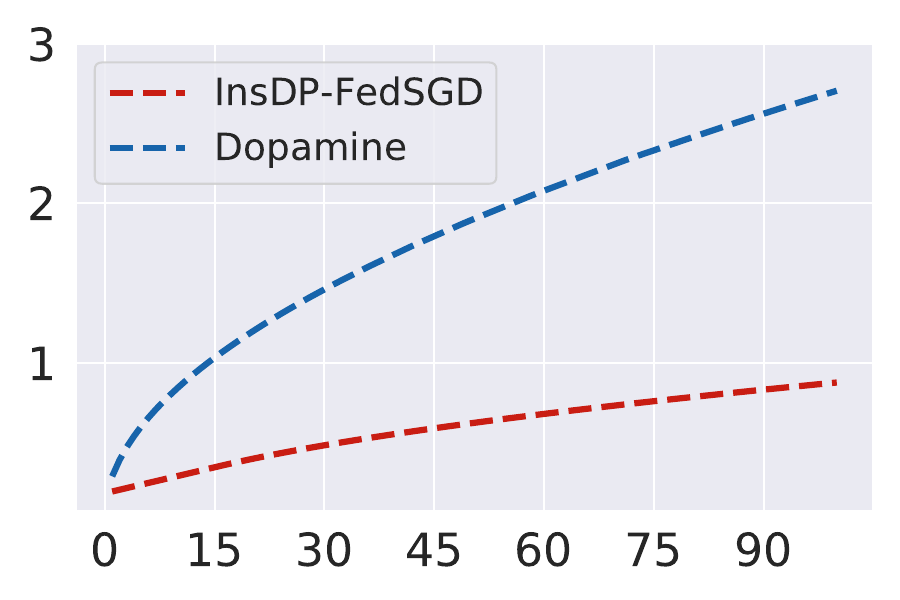}&
\includegraphics[height=\fedsgdfigheightc]{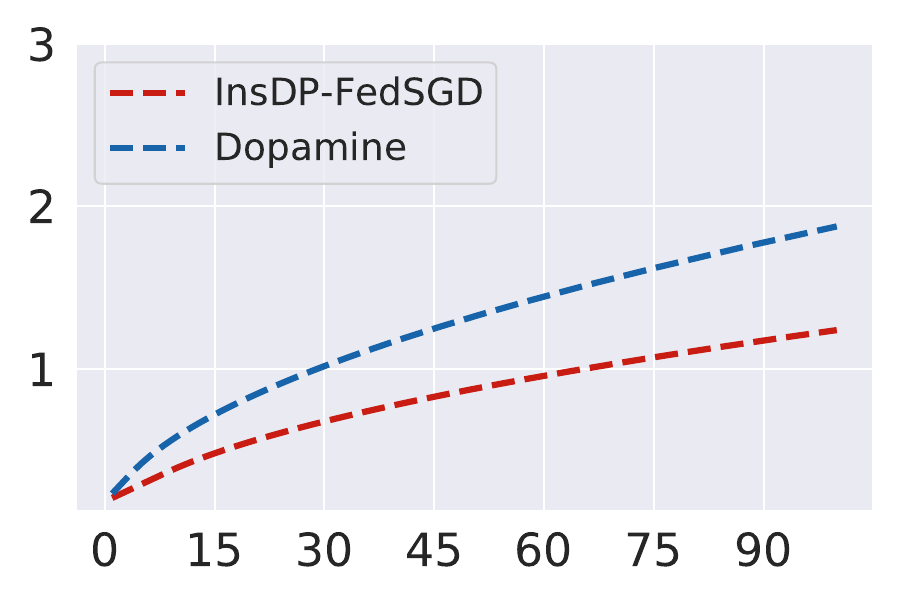}&\\
\\[-5ex]
& \makecell{{\footnotesize Round}}& \makecell{{\footnotesize Round}} 
\end{tabular}
\end{subtable}
}
\vspace{-4mm}
\caption{\small Privacy budget $\epsilon$ under differnt user subsampling  probability $q$.  \insdpfedsgd{} achieves a tighter bound than Dopamine.}
\vspace{-4mm}
\label{fig:privacybound_compare} 
\end{figure}

}

{
\begin{figure}[t]
\newlength{\tauinsfigheightc}
\settoheight{\tauinsfigheightc}{\includegraphics[width=0.43\linewidth]{figures/ccsfinal_plots/cer_acc_conf/cer_acc_mnist.pdf}}

\newcommand{\rowname}[1]%
{\rotatebox{90}{\makebox[\tauinsfigheightc][c]{\footnotesize #1}}}

\centering

{
\renewcommand{\tabcolsep}{10pt}
\begin{subtable}[]{\linewidth}
\centering
\begin{tabular}{@{}p{5mm}@{}c@{}c@{}c@{}c@{}c@{}c@{}c@{}}
        & \makecell{\shortstack{\footnotesize (a) \mnist{}}}
             & \makecell{\shortstack{\footnotesize (b) \cifar{}}}
        \vspace{-1.7pt}\\
\rowname{\makecell{$k$}}&
\includegraphics[height=\tauinsfigheightc]{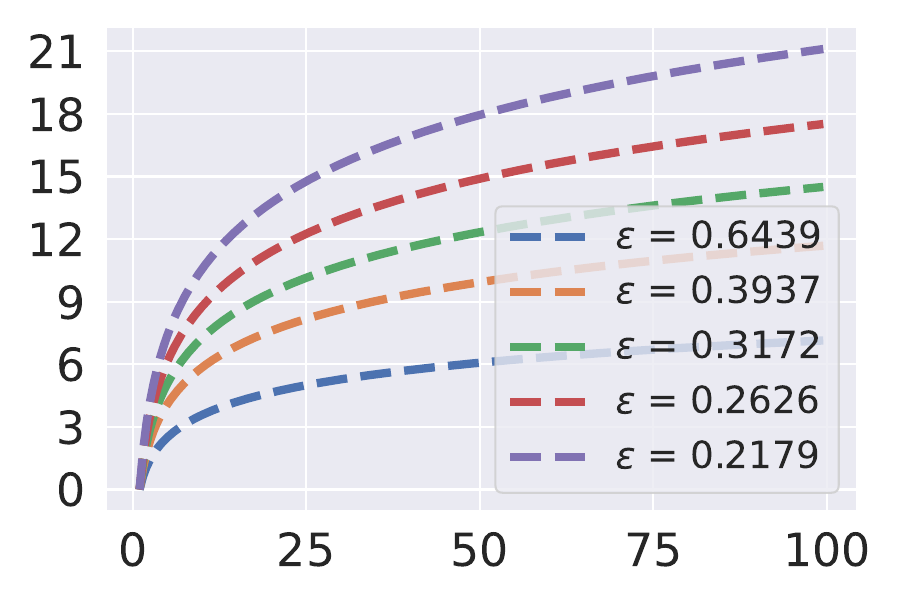}&
\includegraphics[height=\tauinsfigheightc]{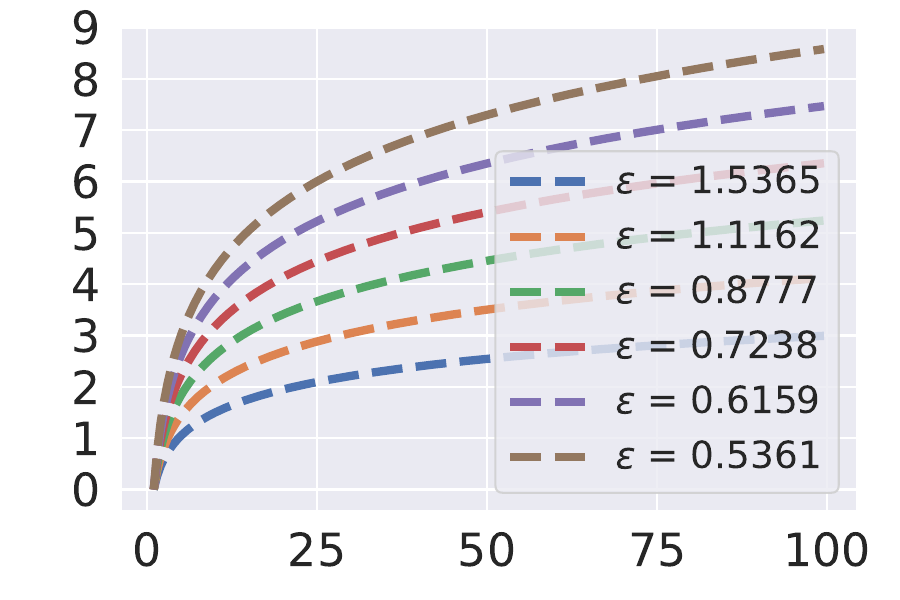}&\\
\\[-5ex]
& \makecell{{\footnotesize $\tau$}}& \makecell{{\footnotesize $\tau$}} 
\end{tabular}
\end{subtable}
}
\vspace{-5mm}
\caption{\small Lower bound of $k$ under instance-level $\epsilon$  given attack effectiveness $\tau$.}
\vspace{-5mm}
\label{fig:ins_tau} 
\end{figure}

}

\begin{rev}

\vspace{-1mm}
\subsubsection{Computational Overhead and Overall Privacy Costs of Robustness Certifications}
\label{exp:overhead_privacycost}
Our robustness certifications are based on DPFL, and we do not impose additional operations for DPFL, so the certifications are applicable for practical FL scenarios where the DPFL algorithm is implemented~\cite{googledpfl}. 
The major overhead of our certifications comes from re-training the DPFL algorithm $O$ times for Monte-Carlo approximation (see \cref{subsubsec:certification_monte_carlo_approx}). Notably,  retraining is a common requirement when providing certifications against poisoning attacks~\cite{weber2020rab,rosenfeld2020certifiedlableflip}. In addition, the multiple runs of re-training are parallelizable and can be speeded up with multiple GPUs. We report the running time for certifications  on \sent{} in \cref{app:running_time_overhead}. 
The re-training for Monte-Carlo approximation also increases the overall privacy costs, as discussed in ~\cref{subsubsec:certification_monte_carlo_approx}. In practice, one can adjust $O$ to prioritize robustness (i.e., a larger 
$O$ for higher certification confidence), or privacy (i.e., a smaller $O$ for fewer times of re-training).  As a result, certified robustness can be achieved by balancing the privacy budget and robustness confidence. For example, as shown in \cref{app:add-practical-consideration}, the maximal certified number of adversaries on \cifar{} is $k=4$  with the overall privacy cost  $10.15$ (calculated by $\epsilon O$) under a confidence level of 80\% (details about confidence level are deferred to \cref{appendix_confidence_level}).

\end{rev}

\looseness=-1
\vspace{-0.5em}
\subsubsection{Certified Attack Inefficacy under Different $k$ and Different Poisoning Attacks.} 
\label{subsubsec:cer_cost_different_k_attack_userdp}
To evaluate \cref{thm_costfunc_k_client} and characterize the tightness of our theoretical lower bound $\underline{J(D')}$, we compare it with the empirical \attackcost $J(D')$ under different local poison fraction $\alpha$, attack methods and scale factor $\gamma$ in \cref{fig:user_k_gamma_alpha}. Note that when $k=0$, the model is benign, so the empirical attack inefficacy equals the certified one.  
\begin{enumerate}[leftmargin=*]
\item When $k$ increases, the attack ability grows, and both the empirical  \attackcost and theoretical lower bound decrease.
\item In \cref{fig:user_k_gamma_alpha} row 1, given the same $k$, higher $\alpha$, i.e., poisoning more local instances for each attacker, achieves a stronger attack, under which the empirical $J(D)$ can be achieved and is closer to the certified lower bound. This indicates that the lower bound appears tighter when the poisoning attack is stronger. 
\item In \cref{fig:user_k_gamma_alpha} row 2, we fix $\alpha=100\%$ and evaluate \userdpfedavg{} under different $\gamma$ and attack methods. It turns out that DP serves as a strong defense empirically for FL, given that $J(D)$ did not vary much under different  $\gamma$~(1,50,100) and  different attack methods (BKD, DBA, LF). This is because the clipping operation restricts the magnitude of malicious updates, rendering the model replacement ineffective; the Gaussian noise perturbs the malicious updates and makes the DPFL model stable, and thus the FL model is less affected by poisoning instances.  
\item In both rows, the lower bounds are tight when $k$ is small. When $k$ is large, there remains a gap between our lower bounds  and  empirical  \attackcost under different attacks, suggesting that there is room for improvement in either devising more effective poisoning attacks or developing tighter robustness certification techniques.
\end{enumerate}

\vspace{-1mm}
\subsubsection{Certified Attack Inefficacy under Different $\epsilon$.} 
We further explore the impacts of different factors on the certified \attackcost.
\cref{fig:user_eps_tau} (a-c) present the empirical \attackcost and the certified \attackcost lower bound given different $\epsilon$ of user-level DP.
As the privacy guarantee becomes stronger (smaller $\epsilon$), the model is more robust, achieving higher $J(D')$ and $\underline{J(D')}$. The results under the BKD attack are omitted to \cref{app:add-user-dpfl_exp_results}.

In \cref{fig:user_eps_tau} (d-f), we train user-level ($\epsilon$, $\delta$) DPFL models, calculate corresponding $J(D)$, and plot the lower bound of $k$ given different attack effectiveness hyperparameter $\tau$  according to \cref{thm_k_clients_bound}.
It shows that
\textbf{(1)} when the required attack effectiveness is higher (larger $\tau$), more attackers are required. 
\textbf{(2)} To achieve the same effectiveness of the attack, a fewer number of attackers is needed for larger $\epsilon$, which means a DPFL model with weaker privacy is more vulnerable to poisoning attacks.

\vspace{-2mm}
\subsection{Evaluation Results of Instance-level DPFL}
\label{app:add-ins-dpfl_exp_results}
Here, we start by comparing the privacy protection  between our \insdpfedsgd{} and Dopamine, and then present certified robustness for \insdpfedavg{} based on \textbf{certified accuracy} under (1) different $\epsilon$, (2)  given confidence level; and \textbf{certified \attackcost} under (1) different $k$ and attacks, and (2) different $\epsilon$.

\sloppy
\vspace{-2mm}
\subsubsection{Privacy Bound Comparison.} \label{subsubsec:privacy_bound_compare_dopamine}
We compare \insdpfedsgd{} with Dopamine\revise{, both under RDP accountant~\cite{mironov2017renyi} for convenience of comparison, to validate the privacy amplification of \insdpfedsgd{} provided by user subsampling. }
With the same noise level ($\sigma=3.0$), clipping threshold ($S=1.5$), and batch sampling probability ($p=0.4$), we calculate the privacy budget under different user sampling probability $q=m/N$. 
\cref{fig:privacybound_compare} shows that \insdpfedsgd{} achieves tighter privacy bound over training rounds.
For instance, at round 200, with $q=10/30$, our method ($\epsilon=0.87$) achieves a much tighter privacy guarantee than Dopamine ($\epsilon=2.70$)\revise{, which comes from user subsampling $q<1$, while Dopamine neglects it. }

\vspace{-1mm}
\subsubsection{Certified Accuracy under Different $\epsilon$.}\label{subsubsec:cer_acc_under_insdp_eps}
We report the certified accuracy of \insdpfedavg{} under different $\epsilon$ on \mnist{} and \cifar{} in \cref{fig:ins_ceracc_k_eps_mnist} (a) and \cref{fig:ins_ceracc_k_eps_cifar} (a). 
We note that the optimal $\epsilon$ that is able to certify the largest number of poisoned instances $k$ is around $0.3593$ for \mnist{} and  $0.6546$ for \cifar{}.
Despite the different FL setups (e.g., the total number of users) under user/instance DP, we can approximately compare the \textit{certified robustness} in terms of the number of tolerable poisoned \textit{instances} for the two DP levels under the same $\epsilon$. When $\epsilon\approx0.4$ on \mnist, \userdpfedavg{} can certify a maximum of $k\approx5$ attackers\revise{, translating to a total of roughly 1250 poisoned instances}, while \insdpfedavg{} can certify up to $k\approx12$ poisoned instances.
Therefore, \userdpfedavg{} can certify many more poisoned instances under the same $\epsilon$ than \insdpfedavg{}, though with a different privacy scope.
We report the (uncertified) accuracy of \insdpfedavg{} in \cref{sec:app_exp_details}. 

\vspace{-1mm}
\subsubsection{{Certified Accuracy with a Confidence Level.}}
\label{appendix_confidence_level}
Here, we present the certified accuracy with the confidence level for both user and instance-level DPFL.
We use Hoeffding's inequality~\cite{hoeffding1994probability} to calibrate the empirical estimation with one-sided error tolerance $\psi$, i.e., one-sided confidence level $1-\psi$.
We denote the empirical estimation of the class confidence for class $c$ as $ \widetilde{F}_c(\Mcal(D),x) = \frac{1}{O} \sum_{o=1}^O f^s_c $.
For a test input $x$, suppose $\Aclass, \Bclass \in [C]$ satisfy $ \Aclass=\arg \max_{c \in [C]} \widetilde{F}_c(\mathcal{M}(D),x)$ and $\Bclass=\arg \max_{c \in [C]:c\neq \Aclass} \widetilde{F}_c(\mathcal{M}(D),x)$.
For a given error tolerance $\psi$, we use Hoeffding's inequality to compute a lower bound $\underline{F_{\Aclass}(\mathcal{M}(D),x)}=\widetilde{F}_{\Aclass}(\mathcal{M}(D),x) - \sqrt{\frac{\log(1/\psi)}{2O}}$ for $\Aclass$,
as well as  a upper bound $\overline{F_{\Bclass}(\mathcal{M}(D),x)}=\widetilde{F}_{\Bclass}(\mathcal{M}(D),x) + \sqrt{\frac{\log(1/\psi)}{2O}}$ for $\Bclass$.
We use $\psi=0.01$ (i.e., 99\% confidence).

From the results in \cref{fig:cer_acc_conf}, we observe similar trends between $\epsilon$ and certified accuracy as in \cref{fig:userdp_ceracc}, \cref{fig:ins_ceracc_k_eps_mnist} (a) and \cref{fig:ins_ceracc_k_eps_cifar} (a). In general, the largest number of certified adversarial users in \cref{fig:cer_acc_conf} is smaller than the previous results because we calibrate the empirical estimation, leading to the narrowed class confidence gap between classes $\Aclass$ and $\Bclass$.

\vspace{-2mm}
\subsubsection{Certified Attack Inefficacy under Different $k$.} 
We report the certified \attackcost  of \insdpfedavg{} on \mnist{} and \cifar{} in   \cref{fig:ins_ceracc_k_eps_mnist} and \cref{fig:ins_ceracc_k_eps_cifar}.  We see that from \cref{fig:ins_ceracc_k_eps_mnist} (b)(c) and \cref{fig:ins_ceracc_k_eps_cifar} (b)(c),  poisoning more instances (i.e.,  a larger $k$) induces lower theoretical and empirical \attackcost lower bounds.

\vspace{-2mm}
\subsubsection{Certified Attack Inefficacy under Different $\epsilon$.} From \cref{fig:ins_ceracc_k_eps_mnist} (d)(e) and \cref{fig:ins_ceracc_k_eps_cifar} (d)(e), it is clear that instance-level DPFL  with a stronger privacy guarantee ensures higher \attackcost both empirically and theoretically, meaning that it is more robust against poisoning attacks. 
In \cref{fig:ins_tau}, we train instance-level $(\epsilon,\delta)$ DPFL models, calculate corresponding $J(D)$, and plot the lower bound of $k$ given different attack effectiveness hyperparameter $\tau$ according to \cref{thm_k_clients_bound}.
We can observe that fewer poisoned instances are required to reduce the $J(D')$ to a similar level for a less private DPFL model, indicating that the model is easier to be attacked.

\vspace{-2mm}
\section{{Discussion \& Conclusion}}
\label{sec:dis_con}
In this work, we take the first step to characterize the connections between \textit{certified} robustness against poisoning attacks and DP in FL. 
We introduce two certification criteria, based on which we prove that an FL model satisfying user-level (instance-level) DP is certifiably robust against a bounded number of adversarial users (instances). We also provide formal privacy analysis to achieve improved instance-level privacy. 
{Through comprehensive evaluations, we validate our theories and establish a general measurement framework to assess the certified robustness yielded by DPFL.}

\begin{rev}
    
\textbf{Limitations \& Future Work.}
One limitation of our work is that we focus on the ``central'' DP with a trusted server for user-level DPFL, where the FL server clips and adds noise, as opposed to a ``local'' DP setting, where each client clips and adds their noise locally~\cite{naserilocal}. While we  follow~\cite{geyer2017differentially,mcmahan2018learning} to consider a trusted server in the central DP regime, it offers weaker privacy protection than local DP, since the privacy guarantee does not hold against the server who can see raw client updates. 
It would be interesting to further extend the analysis to FL with local DP guarantees. 
Another limitation is that our certifications could add computational overhead. Certifying training-time robustness necessitates training multiple models, demonstrated in prior certification studies~\cite{weber2020rab,rosenfeld2020certifiedlableflip}, though this can be accelerated using parallelization and multiple GPUs.

\end{rev}
{The future directions include 
\revise{(1) extending our analysis to more complicated DP settings, such as scenarios where only non-attackers apply local DP in FL while attackers do not~\cite{naserilocal};
(2) combining DPFL with robust FL aggregations to further boost robustness;}
(3) investigating the certified robustness of advanced FL algorithms~\cite{shi2023make,cheng2022differentially,noble2022differentially}
that would maintain higher utility under DP in non-IID data settings; 
(4) developing tighter privacy accountant techniques over FL training to improve the certified robustness from the DP theory perspective;
(5) investigating advanced model architectures 
and pretraining techniques 
to further improve the certified robustness of DPFL.
We hope our work will help provide more insights into the relationships between privacy and certified robustness in the context of FL, paving the way for more secure and privacy-preserving FL applications in practice.}

\vspace{-1mm}
\section{Acknowledgements}

The authors thank Linyi Li and the anonymous reviewers for their valuable feedback and suggestions. This work is partially supported by the NSF grant No.1910100,
No.2046726, No. 2046795, No. 2205329, NSF ACTION Institute, DARPA GARD (HR00112020007), C3AI, Alfred P. Sloan Foundation, Amazon Research Award, DARPA contract \#N66001-15-C-4066,
and the National Research Foundation Singapore and DSO National Laboratories under the AI Singapore Programme (AISG Award No: AISG2-RP-2020-018).

\bibliographystyle{ACM-Reference-Format}
\bibliography{reference}
%\newpage
\appendix

The Appendix is organized as follows:
\begin{itemize}
    \item \cref{sec:dpfl_algo} provides the proofs for the privacy guarantees of our DPFL algorithms.
    \item  \cref{sec:app_exp_details} provides more details on experimental setups and the additional experimental results on  robustness certifications.
    \item   Appendix~\ref{sec:app_proofs} provides the proofs for the certified robustness-related analysis, including  \cref{lemma_group_dp}, \cref{thm_pred_consist_one_client}, \cref{thm_pred_consist_k_client},~\cref{thm_costfunc_k_client}, \cref{thm:thmsapplytoinsdp} and \cref{thm_k_clients_bound}.
     \item   \revise{\cref{app:renyi_certifications} provides the theoretical results and corresponding proofs for certified robustness against FL poisoning attacks derived from R\'enyi DP and Randomized Smoothing via R\'enyi Divergence. }
\end{itemize}

\section{Differentially Private Federated Learning}
\label{sec:dpfl_algo}

\begin{table*}[htbp]
    \caption{{Table of notations.}}
    \label{tb:notations}   
    \centering
\scalebox{1}{
\begin{tabular}{l| l}
\toprule
Notation & Description  \\
\midrule 
$N$ & number of FL users\\
$D_1,\ldots, D_N$ & local datasets of $N$ users \\
$D$ &  $\{D_1, \ldots, D_N \}$ clean FL dataset\\
$T$ & total number of communication rounds \\ 
$\eta$ & learning rate \\
$E$ & local epochs \\
$q$ &  user sampling probability  \\ 
$m$ &  number of selected users at each round \\ 
$U_t$ &  selected user set at round $t$\\
$w_{t}$ &  global model at round $t$\\ 
$\Delta w_{t}^i$ &  local update of client $i$ at round $t$\\ 
$D'$ &  poisoned FL dataset\\
$k$ &  number of adversarial users or adversarial instances \\
$S$ & clipping threshold \\
$\sigma$& noise level\\
$\delta$ & DP privacy parameter\\
$\epsilon$ & DP privacy budget\\
$\mathcal{M}$ &  DPFL training protocol  \\
$\mathcal{M}(D)$ & clean DPFL model at round $T$ \\ 
$\mathcal{M}(D')$ & poisoned DPFL model at round $T$ \\ 
$f_{c}(\mathcal{M}(D), x)$ &  confidence for class $c$ on test sample $x$ \\
$F_{c}(\mathcal{M}(D), x)$ &  expected confidence for class $c$ on test sample $x$\\
$H(\mathcal{M}(D),x)$ & prediction, i.e.,  top-1 class based on the expected confidence\\
$C(\mathcal{M}({D'}))$ & attack cost on the poisoned model $\mathcal{M}({D'})$\\
$J({D'})$ & expected attack cost on the poisoned model $\mathcal{M}({D'})$\\
$ \bar C$ & bound on attack cost $C(\cdot)$\\
$\bar g(x_j)$ & clipped gradient for sample $x_j$ in \insdpfedsgd{}\\
$\widetilde g$ & noise-perturbed and clipped gradient for sample $x_j$ in \insdpfedavg{}\\
$\gamma$ & scale factor in model replacement attack \\
$O$ & number of Monte Carlo samples \\
$\psi$& one-sided error tolerance in Monte-Carlo sampling \\
$\cerk$ & theoretical upper bound for the number of adversarial users/instances that can satisfy the certified prediction\\
$\underline{J(D')}$ & theoretical lower bound of the attack cost for poisoned DPFL model based on the certified cost \\
$\epsilon (\alpha)$ & RDP parameter \\
$\alpha$ & RDP order \\
\bottomrule 
\end{tabular}}
\end{table*}

We first present all the notations  used in our paper in \cref{tb:notations}.

\subsection{ \userdpfedavg{}}
\begin{small}
\begin{algorithm2e}
\caption{\small \userdpfedavg. }
\label{algo:Fedavg_client_privacy}
\scalebox{1}{
\begin{minipage}{0.8\linewidth}

\SetKwFunction{funone}{UserUpdate}
\SetKwFunction{funtwo}{Clip}
\KwIn{Initial model $w_0$, user sampling probability $q$, privacy parameter $\delta$, clipping threshold $S$, noise level $\sigma$, local datasets $D_1,...,D_N$, local epochs $E$, learning rate $\eta$. }
 
\KwOut{ FL model $w_T$ and privacy cost $\epsilon$}

\textbf{Server executes:}\\
\For{each round $t=1$ \KwTo $T$}
    {   
         $m \gets \max(q\cdot N, 1)$\;
          $U_t \gets$ (random subset of $m$ users)\;
        \For{each user $i \in U_t$ in parallel}
        {   
         $\Delta w_{t}^i  \gets $ \funone{$i, w_{t-1}$} \;
        }
   $w_{t} \gets w_{t-1 } + \frac{1}{m}\left(\sum_{i\in U_t}  \funtwo{$\Delta w^i_{t} ,S$}+  \mathcal{N}\left(0, \sigma^{2} S^{2}\right) \right); $ 

   $\mathcal{M}.$\textit{accum\_priv\_spending}$(\sigma,q,\delta)$ \;
    }
    $ \epsilon = \mathcal{M}.$\textit{get\_privacy\_spent}() \;
   \KwRet{$ w_T, \epsilon $ }
   
\SetKwProg{foo}{Procedure}{}{}
\foo{\funone{$i, w_{t-1}$}}{
  $w \gets w_{t-1}$ \; 
  \For{local epoch  $e =1$ \KwTo $E$}
    {   
        \For{batch $b \in $  local dataset $D_i$}
        {   
            $w \gets w- \eta \nabla l(w;b)$
        }
    }
  $\Delta w_{t}^i  \gets w -w_{t-1}$ \;
    \KwRet{$\Delta w_{t}^i$  }
}
\SetKwProg{fooo}{Procedure}{}{}
\fooo{\funtwo{$\Delta  , S$}}{
   \KwRet{$ \Delta  / \max \left(1, \frac{\left\|\Delta\right\|_{2}}{S}\right)$   }
}
\end{minipage}}
\end{algorithm2e}
\end{small}

To calculate the privacy costs for \cref{algo:Fedavg_client_privacy}, existing works utilize  moments accountant~\cite{abadi2016deep} for privacy analysis~\cite{geyer2017differentially,mcmahan2018learning}.
We note that R\'enyi Differential Privacy (RDP)~\cite{mironov2017renyi} supports a tighter composition of privacy budget than the moments accounting technique for DP~\cite{mironov2017renyi}. Therefore,  we utilize RDP~\cite{mironov2017renyi} to perform the privacy analysis in Algorithm~\ref{algo:Fedavg_client_privacy}. Specifically,  $\mathcal{M}.$\textit{accum\_priv\_spending}() is the call on RDP accountant~\cite{wang2019subsampled}, and $\mathcal{M}.$\textit{get\_privacy\_spent}() transfers RDP guarantee to DP guarantee based on the RDP to DP conversion theorem of \cite{balle2020hypothesis}.

\subsection{ \insdpfedsgd{}}

Here, we present the algorithm \insdpfedsgd{}.
\begin{small}
\begin{algorithm2e}
\caption{\small \insdpfedsgd. }
\label{algo:Fedsgd_sample_privacy}
\scalebox{0.9}{
    \begin{minipage}{1.0\linewidth}
    \SetKwFunction{funone}{UserUpdate}
    \SetKwFunction{funtwo}{Clip}
    \KwIn{
    Initial model $w_0$,  user sampling probability $q$, privacy parameter $\delta$, local clipping threshold $S$, local noise level $\sigma$, local datasets $D_1,...,D_N$, learning rate $\eta$, batch sampling probability $p$. }

    \KwOut{ FL model $w_T$ and privacy cost $\epsilon$}
    \textbf{Server executes:}\\
    \For{each round $t=1$ \KwTo $T$}
        {   
             $m \gets \max(q\cdot N, 1)$\;
              $U_t \gets$ (random subset of $m$ clients)\;
            \For{each user $i \in U_t$ in parallel}
            {   
             $\Delta w_{t}^i  \gets $ \funone{$i, w_{t-1}$} \;
            }
       $w_{t} \gets w_{t-1} + \frac{1}{m} \sum_{i\in U_t}  \Delta w^i_{t}   $ \;
        $\mathcal{M}.$\textit{accum\_priv\_spending}$(\sqrt{m} \sigma,pq,\delta)$
        }
        $ \epsilon = \mathcal{M}.$\textit{get\_privacy\_spent}() \;
       \KwRet{$ w_T, \epsilon $ }
       
    \SetKwProg{foo}{Procedure}{}{}
    \foo{\funone{$i, w_{t-1}$}}{
      $w \gets w_{t-1}$ \; 
     
        $b_t^i \gets $(uniformly sample a batch from $D_i$ with probability $p =L/|D_i|$)\;
        \For{each $x_j  \in b_t^i$}
        {
                $g(x_j) \gets \nabla l(w;x_j)$\; 
            $ \bar g(x_j) \gets $  \funtwo{$g(x_j)  , S$} \;
        }
            
  		$ \widetilde g \gets \frac{1}{L}\left(\sum_j \bar g(x_j) + \mathcal{N}\left(0, \sigma^{2} S^{2}\right) \right )$\;
         $w \gets w- \eta  \widetilde g  $ \;
      $\Delta w_{t}^i  \gets w - w_{t-1} $ \;
        \KwRet{$\Delta w_{t}^i$ }
    }
    \SetKwProg{fooo}{Procedure}{}{}
    \fooo{\funtwo{$\Delta  , S$}}{
       \KwRet{$ \Delta  / \max \left(1, \frac{\left\|\Delta\right\|_{2}}{S}\right)$   }
    }

        \end{minipage}}
    \end{algorithm2e}
\end{small}

Next, we recall \cref{proposition:insdpfedsgd_privacy_guarantee} and present its proof.
\propinsdpfedsgd*

\begin{proof}
\textbf{(1)} In instance-level DP, we consider the sampling probability of each instance under the combination of user-level sampling and batch-level sampling. Since the user-level sampling probability is $q$ and the batch-level sampling probability is $p$, each instance is sampled with probability $pq$. 
\textbf{(2)} Additionally, since the sensitivity of instance-wise gradient w.r.t one instance is $S$, after local gradient descent and server FL aggregation, the equivalent sensitivity of global model w.r.t one instance is $S'=\frac{\eta S}{Lm}$ according to Eq~(\ref{eq:onestep_globalupdate_fedsdg}). 
\textbf{(3)} Moreover, since the local noise is $n_i \sim \mathcal{N}(0,\sigma^2 S^2)$ ,  the  ``virtual'' global noise is $n= \frac{\eta}{mL}\sum_{i \in U_t} n_i $ according to Eq~(\ref{eq:onestep_globalupdate_fedsdg}), so  $n \sim \mathcal{N}(0,\frac{ \eta^2 \sigma^2 S^2}{mL^2})$. Let $\frac{ \eta^2 \sigma^2 S^2}{mL^2}=   {\sigma'}^2 {S'}^2$ such that $n \sim \mathcal{N}(0,  {\sigma'}^2 {S'}^2 )$.  
Since $S'=\frac{\eta S}{Lm}$,  the equivalent global noise level is ${\sigma'}^2 = \sigma^2m$, i.e.,  $\sigma'= \sigma \sqrt{m}$.  Then, we use $pq$ to represent \emph{instance-level} sampling probability, $T$ to represent FL training rounds, $\sigma \sqrt{m}$ to represent the equivalent global noise level. 
The rest of the proof follows \textbf{(1)} RDP subsampling amplification~\cite{wang2019subsampled}, \textbf{(2)} RDP composition for privacy budget accumulation over $T$ rounds based on the RDP composition~\cite{mironov2017renyi}
and \textbf{(3)} transferring RDP guarantee to DP guarantee based on the conversion theorem~\cite{balle2020hypothesis}.
\end{proof}

\subsection{ \insdpfedavg{}}

\begin{small}
\begin{algorithm2e}
\caption{\small \insdpfedavg. }
\label{algo:Fedavg_sample_privacy}
\scalebox{0.9}{
    \begin{minipage}{1.0\linewidth}
    \SetKwFunction{funone}{UserUpdate}
    \SetKwFunction{funtwo}{Clip}
    \KwIn{Initial model $w_0$,  user sampling probability $q$, privacy parameter $\delta$, local clipping threshold $S$, local noise level $\sigma$, local datasets $D_1,...,D_N$, local steps $V$, learning rate $\eta$, batch sampling probability $p$. }
    \KwOut{ FL model $w_T$ and privacy cost $\epsilon$}
    \textbf{Server executes:}\\
    \For{each round $t=1$ \KwTo $T$}
        {   
             $m \gets \max(q\cdot N, 1)$\;
              $U_t \gets$ (random subset of $m$ users)\;
            \For{each user $i \in U_t$ in parallel}
            {   
             $\Delta w_{t}^i , \epsilon^i_{t}  \gets $ \funone{$i, w_{t-1}$} \;
            }
            \For{each user $i \notin  U_t$}
            {   
              $\epsilon^i_{t}  \gets \epsilon^i_{t-1}  $ \;
            }
       $w_{t} \gets w_{t-1} + \frac{1}{m} \sum_{i \in U_t}  \Delta w^i_{t}   $ \;
       $ \epsilon_t =  \mathcal{M}.$\textit{parallel\_composition}$(\{\epsilon^i_{t}\}_{i\in [N]})$
        }
        $ \epsilon = \epsilon_T $ \;
       \KwRet{$ w_T, \epsilon $ }
       
    \SetKwProg{foo}{Procedure}{}{}
    \foo{\funone{$i, w_{t-1}$}}{
      $w \gets w_{t-1}$ \; 
      \For{each local step  $v =1$ \KwTo $V$}
        {   
            $b \gets $(uniformly sample a batch from $D_i$ with probability $p =L/|D_i|$)\;
          \For{each $x_j  \in b$}
        {
                $g(x_j) \gets \nabla l(w;x_j)$\; 
            $ \bar g(x_j) \gets $  \funtwo{$g(x_j)  , S$} \;
        }
            
  		$ \widetilde g \gets \frac{1}{L}(\sum_j \bar g(x_j) + \mathcal{N}\left(0, \sigma^{2} S^{2}\right)  )$\;
         $w \gets w- \eta \widetilde g   $ \;
       $ \mathcal{M}^i.$\textit{accum\_priv\_spending}$(\sigma,p,\delta)$ \;
        }
      $ \epsilon^i_{t} = \mathcal{M}^i.$\textit{get\_privacy\_spent}() \;
      $\Delta w_{t}^i  \gets w - w_{t-1} $ \;
          
        \KwRet{$\Delta w_{t}^i, \epsilon^i_{t}$ }
    }
    \SetKwProg{fooo}{Procedure}{}{}
    \fooo{\funtwo{$\Delta  , S$}}{
       \KwRet{$ \Delta  / \max \left(1, \frac{\left\|\Delta\right\|_{2}}{S}\right)$   }
    }
    \end{minipage}}
\end{algorithm2e}
\end{small}

Next, we will first consider the special case of one FL training round (i.e., $T=1$) to showcase the privacy cost \textit{aggregation}. Then, we will combine \textit{local privacy cost accumulation} in each user and the privacy cost aggregation in the server for the general case with any  $t$ FL rounds. When $T=1$, the relationship between  the privacy cost of the local model $\epsilon^i, {i \in [N]}$  and the privacy cost of global model $\epsilon$ for one FL training round is characterized in \cref{lemma:insdp_one_round}. 
For the general case of any  $t$ FL rounds, we  provide the privacy guarantee by combing the RDP accountant for the local model and the parallel composition for the global model in \cref{thm:insdp_t_round}.

\begin{lemma}[\insdpfedavg{} Privacy Guarantee when $T=1$]
\label{lemma:insdp_one_round}
In Algorithm~\ref{algo:Fedavg_sample_privacy}, when $T=1$, suppose local mechanism $\Mcal^i$ satisfies $(\epsilon^i, \delta^i)$-DP,  then global mechanism $\Mcal$ satisfies $(\max_{i \in [N]} \epsilon^i,\delta^i )$-DP.
\end{lemma}

\begin{proof}
\sloppy
We can regard FL as partitioning a dataset $D$ into $N$ disjoint subsets $\{D_1, D_2,\ldots,D_N\}$. $N$ local randomized mechanisms $\{\Mcal^1, \ldots, \Mcal^N \}$ are operated on these $N$ parts separately and each $\Mcal^i$ satisfies its own $\epsilon^i$-DP for $i \in [1,N]$. 
Without loss of generality, we assume the modified data sample $x'$ ($x \rightarrow x'$ causes $D \rightarrow D'$) is in the local dataset of $k$-th client $D_k$. Then $D,D'$ are two neighboring datasets, and $D_k,D_k'$ are also two neighboring datasets.
Consider  a sequence of outcomes (i.e., local model updates) from local mechanisms  $\{ z_1= \Mcal^1(D_1), z_2= \Mcal^2(D_2;z_1),z_3 = \Mcal^3(D_3;z_1,z_2) , \ldots \}$. The global mechanism consists of a series of linear operators on the sequence $z=\Mcal(D)=w_0+ \frac{1}{m}\sum_{i=1}^N {z_i}$.  Note that if $i$-th user is not selected,  $z_i=0$.
According to the parallel composition~\cite{lecnote2013}, we have

{\small
\begin{align*}
	&\Pr[\Mcal(D)=z]\\
	&= \Pr[\Mcal^1(D_1)=z_1] \cdot \Pr[\Mcal^2(D_2;  z_1)=z_2]\cdots \\
	&\quad \quad \quad \cdot \Pr[\Mcal^N(D_N;  z_1,\ldots,z_{N-1})=z_N]\\
	&\leq \left( \exp(\epsilon^k) \Pr[\Mcal^k(D_k';  z_1,\ldots,z_{k-1})=z_k] + \delta^k \right)\\
	&\quad \quad \quad \cdot
	\prod_{i \neq k}\Pr[\Mcal^i(D_i;  z_1,\ldots,z_{i-1})=z_i]\\
&=  \exp(\epsilon^k) \Pr[\Mcal^k(D_k';  z_1,\ldots,z_{k-1})=z_k] \prod_{i \neq k}\Pr[\Mcal^i(D_i;  z_1,\ldots,z_{i-1})=z_i]  \\
	&\quad \quad \quad + 
	\prod_{i \neq k}\Pr[\Mcal^i(D_i;  z_1,\ldots,z_{i-1})=z_i]   \delta^k  \\
	&=\exp(\epsilon^k)  \Pr[\Mcal(D')=z]  + \prod_{i \neq k}\Pr[\Mcal^i(D_i;  z_1,\ldots,z_{i-1})=z_i]   \delta^k \\
 & \leq \exp(\epsilon^k)  \Pr[\Mcal(D')=z]   + \delta^k
\end{align*}
}	 
So  $\Mcal$ satisfies $\epsilon^k$-DP when the modified data sample lies in the subset $D_k$. Considering the worst case where the modified data samples are sampled, we derive that $\Mcal$ satisfies $(\max_{i \in [N]} \epsilon^i)$-DP.
\end{proof}

Next, we recall \cref{thm:insdp_t_round} and present its proof.
\thminsdpfedavgprivacytround* 

\begin{proof}
Again, without loss of generality, we assume the modified data sample $x'$ ($x \rightarrow x'$ causes $D \rightarrow D'$) is in the local dataset of $k$-th user $D_k$. We first consider the case when all users are selected. At each round $t$, $N$ mechanisms are operated on $N$ disjoint parts, and each $\Mcal^i_t$ satisfies its own $\epsilon^i$-DP where $\epsilon^i$ is the privacy cost for accessing the local dataset $D_i$ \emph{for one round} (not accumulating over previous rounds). 
 Let $D,D'$ be two neighboring datasets ($D_k,D_k'$ are also two neighboring datasets).
Suppose $z_0=\Mcal_{t-1}(D)$ is the aggregated randomized global model at round $t-1$, and $\{z_1,\ldots, z_N\}$ are the randomized local updates at round $t$, we have a sequence of computations $\{ z_1= \Mcal^1_t(D_1;z_0), z_2= \Mcal^2_t(D_2;z_0,z_1),z_3 = \Mcal^3_t(D_3;z_0,z_1,z_2) \ldots \}$ and  $z= \Mcal_t(D)=z_0+ \frac{1}{m}\sum_{i}^N {z_i}$.  
We first consider the sequential composition~\cite{dwork2014algorithmic} to accumulate the privacy cost over FL rounds to gain intuition. According to parallel composition, we have

\begin{align*}	
     &\Pr[\Mcal_t(D)=z]\\
    &= \Pr[\Mcal_{t-1}(D)=z_0] \cdot \prod_{i=1}^{N}\Pr[\Mcal^i_{t}(D_i; z_0, z_1,\ldots,z_{i-1})=z_i]\\
    &=\Pr[\Mcal_{t-1}(D)=z_0] \cdot  \Pr[\Mcal^k_{t}(D_k; z_0, z_1,\ldots,z_{k-1})=z_k] \\
        &   \quad \quad \cdot \prod_{i \neq k}\Pr[\Mcal^i_{t}(D_i; z_0, z_1,\ldots,z_{i-1})=z_i]\\
    &\leq \exp(\epsilon_{t-1})\Pr[\Mcal_{t-1}(D')=z_0]  \\
    &	\quad \quad \cdot  \exp(\epsilon^k) \cdot \Pr[\Mcal^k_{t}(D_k'; z_0,z_1,\ldots,z_{k-1})=z_k]   \\
     &	\quad \quad \cdot  \prod_{i\neq k}\Pr[\Mcal^i_{t}(D_i; z_0, z_1,\ldots,z_{i-1})=z_i]\\
    &= \exp(\epsilon_{t-1}+\epsilon^k) \cdot \Pr[\Mcal_t(D')=z]	
\end{align*}

Therefore, $\mathcal{M}_t$ satisfies  
$\epsilon_t$-DP, where $\epsilon_t= \epsilon_{t-1}+ \epsilon^k$.
Because the modified data sample always lies in $D_k$ over $t$ rounds and $\epsilon_0=0$, we can have $\epsilon_t = t\epsilon^k$, which means that the privacy guarantee of global mechanism $\Mcal_t$ is only determined by the local mechanism of $k$-th user over $t$ rounds.

Moreover, RDP accountant~\cite{wang2019subsampled} is known to reduce the privacy cost from $\mathcal{O}(t)$ to $\mathcal{O}(\sqrt{t})$.  We can use this advanced composition, instead of the sequential composition, to accumulate the privacy cost for local mechanism $\Mcal^k$ over $t$ FL rounds. In addition, we consider user selection. As described in Algorithm~\ref{algo:Fedavg_sample_privacy}, if the user $i$ is not selected at round $t$, then its local privacy cost is kept unchanged at this round.

Take the worst case of where $x'$ could lie in, at round $t$, $\mathcal{M}$ satisfies 
$\epsilon_t$-DP, where $\epsilon_t= \max_{i \in [N]} \epsilon_t^i$, 
local mechanism $\mathcal{M}^i$ satisfies $ \epsilon_t^i$-DP, and the local privacy cost  $ \epsilon_t^i$ is accumulated via local RDP accountant in $i$-th user over $t$ rounds.

\end{proof}

\section{Experimental Details and Additional Results}
\label{sec:app_exp_details}
\subsection{Experimental Details}
\label{subsec:exp_details}
\subsubsection{Additional Implementation Details}

\begin{table*}[!htbp]
    \centering
    \caption{  {Clean accuracy of \userdpfedavg{} on \mnist{}}}
\scalebox{1}{
    \begin{tabular}{R || R R R R R R R R R}
    \toprule
      $\sigma$ & 0 & 0.5 & 0.6 & 0.8 & 1.0 & 1.5 & 1.8 & 2.3 & 3.0  \\ 
      \midrule
     $\epsilon$ &    $\infty$ & 6.9269 & 4.8913 & 2.8305 & 1.8504 & 0.8694 & 0.6298 & 0.4187 & 0.2808  \\ 
     \midrule
     Clean Acc. &    99.66\%  & 99.72\%  & 99.69\%  & 99.71\%  & 99.59\%  & 98.86\%  & 97.42\%  & 89.15\% & 72.79\% \\ 
    \bottomrule
    \end{tabular}}
    \label{tb:userdp_mnist_cleanacc}
\end{table*}

\begin{table*}[!htbp]
    \centering
     \caption{ Clean accuracy of \userdpfedavg{} on \cifar{}} 
\scalebox{1}{
    \begin{tabular}{R || R R R R R R R R R}
    \toprule
    $\sigma$ &  0 & 1.7 & 2.3 & 2.6 & 3.0 & 4.0 & 6.0 \\ \midrule
       $\epsilon$ &   $\infty$ & 0.8781 & 0.546 & 0.4527 & 0.3663 & 0.2444 & 0.1451  
     \\      \midrule
     Clean Acc. &   81.90\%  & 81.82\%  & 80.09\% & 79.27\% & 77.89\% & 73.07\%  & 64.36\%  
    \\ 
    \bottomrule
    \end{tabular}}

    \label{tb:userdp_cifar_cleanacc} 
\end{table*}

\begin{table*}[!htbp]
    \centering
     \caption{ {Clean accuracy of \userdpfedavg{} on \sent{}}}
\scalebox{1}{
    \begin{tabular}{R || R R R R R R R R R}
    \toprule
    $\sigma$ &  0 & 1 & 1.5& 1.7 & 2.0 & 3.0    \\ \midrule
     $\epsilon$ &   $\infty$ &  1.7151 & 0.7382 & 0.579  &  0.4102 &0.2247   \\      \midrule
    Clean Acc. &  64.33\%  & 62.64 \%  & 60.76 \%  & 59.57\% & 58.00\% & 55.28\% \\
    \bottomrule
    \end{tabular}}
    \label{tb:userdp_sent_cleanacc}
\end{table*}

\begin{table*}[!ht]
    \centering
    \caption{  {Clean accuracy of \insdpfedavg{} on \mnist{}}}
    \label{tb:insdp_mnist_cleanacc}
\scalebox{1}{
    \begin{tabular}{R || R R R R R R R R R}
    \toprule
  $\sigma$ & 0 & 1 & 2 & 3 & 4 & 5 & 8 & 10 & 15  \\ \midrule
   $\epsilon$ &    $\infty$  & 3.5691 & 1.0587 & 0.6373 & 0.4589 & 0.3593 & 0.2484 & 0.2251 & 0.2029   \\ \midrule
   Clean Acc. &   99.85\%  & 99.73\%  & 99.73\%  & 99.70\% & 99.65\%  & 99.57\%  & 97.99\%  & 93.30\%  & 77.12\%   \\ 
    \bottomrule
    \end{tabular}}
     
\end{table*}

\begin{table*}[!ht]
    \centering
    \caption{  {Clean accuracy of \insdpfedavg{} on \cifar{}}}
\scalebox{1}{
    \begin{tabular}{R || R R R R R R R R R}
    \toprule
     $\sigma$ & 0 & 1 & 2 & 3 & 4 & 5 & 6 & 7 & 8  \\ \midrule
     $\epsilon$ &   $\infty$  & 4.6978 & 1.4949 & 0.9067 & 0.6546 & 0.513 & 0.4221 & 0.3587 & 0.3158  \\ \midrule
        Clean Acc. &  91.15\%  & 87.91\%  & 86.02\%  & 83.85\%  & 81.43\%  & 77.59\%  & 72.69\%  & 66.47\%  & 62.26\%  \\ 
    \bottomrule
    \end{tabular}}
    \label{tb:insdp_cifar_cleanacc}
\end{table*}

We simulate the federated learning setup (1 server and N users) on a Linux machine with Intel\textregistered{} Xeon\textregistered{} Gold 6132 CPUs and 8 NVidia\textregistered{} 1080Ti GPUs. All code is implemented in Pytorch~\cite{NEURIPS2019_pytorch}.

\subsubsection{Training Details}
\label{subsubsec:training_details}
\sloppy
Next, we summarize the privacy guarantees and clean accuracy offered when we study the certified prediction and certified \attackcost, which are also the training parameters setups when $k=0$ in Figure~\ref{fig:userdp_ceracc},~\ref{fig:user_k_gamma_alpha},~\ref{fig:user_eps_tau},~\ref{fig:ins_ceracc_k_eps_cifar},~\ref{fig:user_eps_backdoor},~\ref{fig:ins_tau},~\ref{fig:ins_ceracc_k_eps_mnist}.
\paragraph{User-level DPFL}
In order to study the user-level certified prediction under different privacy guarantees, for \mnist{}, we set $\epsilon$ to be $0.2808, 0.4187, 0.6298, 0.8694, 1.8504, 2.8305, 4.8913, 6.9269$, which are obtained by training \userdpfedavg{} FL model for $3$ rounds with noise level $\sigma=3.0, 2.3, 1.8, 1.5, 1.0, 0.8, 0.6, 0.5$, respectively (Figure~\ref{fig:userdp_ceracc}(a)). 
For \cifar{}, we set  $\epsilon$ to be $0.1083, 0.1179, 0.1451, 0.2444,  0.3663, 0.4527, 0.5460, 0.8781 $, which are obtained by  training \userdpfedavg{} FL model for one round with noise level $\sigma=10.0, 8.0, 6.0, 4.0, 3.0,2.6,2.3, 1.7$, respectively (Figure~\ref{fig:userdp_ceracc}(b)).
{For \sent{}, we set $\epsilon$ to be $0.2234, 0.2238, 0.2247, 0.4102, 0.579, 0.7382, 1.7151$, which are obtained by  training \userdpfedavg{} FL model for three rounds with noise level $\sigma=5, 4, 3, 2, 1.7,1.5, 1$, respectively (Figure~\ref{fig:userdp_ceracc}(c)).}

 {The clean accuracy (average over 1000 runs) of \userdpfedavg{} under non-DP training ($\epsilon=\infty$) and DP training (varying  $\epsilon$) on \mnist{}, \cifar{}, and \sent{} are reported in Table.~\ref{tb:userdp_mnist_cleanacc}, Table.~\ref{tb:userdp_cifar_cleanacc} and Table.~\ref{tb:userdp_sent_cleanacc} respectively. 
We note that smaller $\epsilon$ results in lower accuracy, but we evaluate small $\epsilon$ only to study the relationship between  privacy and certified accuracy in Figure~\ref{fig:userdp_ceracc}, so as to show the tradeoff. Such extreme cases are not required for certification. 
For other evaluations on our \revise{paper} (such as \cref{fig:user_k_gamma_alpha}, \cref{fig:user_eps_tau}), we use normal $\epsilon$ with reasonable clean accuracy.
}

To certify the \attackcost under the different number of adversarial users $k$ (Figure~\ref{fig:user_k_gamma_alpha}), for \mnist{},  we set the noise level $\sigma$ to be 2.5. When $k=0$, after training \userdpfedavg{} for $T=3,4,5$ rounds, we obtain   FL models with privacy guarantee  $\epsilon = 0.3672, 0.4025, 0.4344$  and clean accuracy (average over $O$ runs) $86.69\%, 88.76\%, 88.99\%$. 
For \cifar{},  we set the noise level  $\sigma$ to be $3.0$. After training \userdpfedavg{} for $T=3,4$ rounds under $k=0$, we obtain   FL models with privacy guarantee $\epsilon = 0.5346, 0.5978 $  and clean accuracy $78.63\%, 78.46\%$.
For \sent{},  we set the noise level  $\sigma$ to be $2.0$. After training \userdpfedavg{} for $T=3$ rounds under $k=0$, we obtain   FL models with privacy guarantee $\epsilon =0.4102 $  and clean accuracy $58.00\%$.

With the interest of certifying \attackcost under different user-level DP guarantees (Figure~\ref{fig:user_eps_tau}, Figure~\ref{fig:user_eps_backdoor}), we explore the empirical \attackcost, and the certified \attackcost lower bound given different $\epsilon$. For \mnist{}, we set  the privacy guarantee $\epsilon $ to be $1.2716,0.8794, 0.6608, 0.5249,0.4344 $, which are obtained by training \userdpfedavg{} FL models for five rounds under noise level $\sigma= 1.3, 1.6, 1.9, 2.2, 2.5$, respectively, and the clean accuracy for the corresponding models are $99.50\%, 99.06\%, 96.52\%, 93.39\%,88.99\%$.
For \cifar{},  we set  the privacy guarantee $\epsilon $ to be $1.600, 1.2127, 1.0395. 0.8530, 0.7616,0.6543,0.5978 $, which are obtained by training \userdpfedavg{} FL models for four rounds under noise level $\sigma= 1.5,1.8,2.0,2.3,2.5,2.8,3.0$, respectively, and the clean accuracy for the corresponding models are $85.59\% ,84.52\%,  83.23\% , 81.90\%, 81.27\%, 79.23\% ,78.46\%$.
For \sent{},  we use the same set of  $\epsilon $ as in certified prediction.

\paragraph{Instance-level DPFL}
To certify the prediction for instance-level DPFL under different privacy guarantees, for \mnist{}, we set privacy cost $\epsilon$ to be $0.2029, 0.2251, 0.2484, 0.3593, 0.4589, 0.6373, 1.0587, 3.5691$, which are obtained by training \insdpfedavg{} FL models for 3 rounds with noise level $\sigma= 15, 10, 8, 5, 4, 3, 2, 1$, respectively (Figure~\ref{fig:ins_ceracc_k_eps_mnist}(a)). For \cifar{},  we set privacy cost  $\epsilon$ to be $0.3158, 0.3587, 0.4221, 0.5130, 0.6546, 0.9067, 1.4949, 4.6978$, which are obtained by training \insdpfedavg{} FL models for one round with noise level $\sigma= 8, 7, 6, 5, 4, 3, 2, 1$, respectively (\cref{fig:ins_ceracc_k_eps_cifar}(a)). 
 {
The clean accuracy (average over 1000 runs) of \insdpfedavg{} under non-DP training ($\epsilon=\infty$) and DP training (varying  $\epsilon$) on \mnist{} and \cifar{} are reported in \cref{tb:insdp_mnist_cleanacc} and \cref{tb:insdp_cifar_cleanacc} respectively. 
}

With the aim to study certified \attackcost under the different number of adversarial instances $k$, for \mnist{},  we set the noise level $\sigma$ to be 10. When $k=0$, after training \insdpfedavg{} for $T=4$ rounds, we obtain   FL models with privacy guarantee  $\epsilon = 0.2383 $  and clean accuracy (average over $O$ runs) $96.40\%$ (\cref{fig:ins_ceracc_k_eps_mnist}(b)(c)). For \cifar{},  we set the noise level  $\sigma$ to be $8.0$. After training \insdpfedavg{} for one round under $k=0$, we obtain   FL models with privacy guarantee $\epsilon =  0.3158$  and clean accuracy $61.78\%$ (\cref{fig:ins_ceracc_k_eps_cifar}(b)(c)).

In order to study the empirical \attackcost and certified \attackcost lower bound under different instance-level DP guarantees, we set the privacy guarantee $\epsilon$ to be $0.5016, 0.311, 0.2646 , 0.2318, 0.2202,0.2096,0.205 $ for \mnist{}, which are obtained by training \insdpfedavg{} FL models for six rounds under noise level $\sigma= 5,8,10,13,15,18,20$, respectively, and the clean accuracy for the corresponding models are $99.60\% , 98.81\%,97.34\%, 92.29\%,88.01\% , 80.94\%,  79.60\% $ (Figure~\ref{fig:ins_ceracc_k_eps_mnist} (d)(e)).
For \cifar{},  we set  the privacy guarantee $\epsilon $ to be $ 1.261,0.9146 , 0.7187, 0.5923, 0.5038 ,0.4385   $, which are obtained by training \insdpfedavg{} FL models for two rounds under noise level $\sigma= 3,4,5,6,7,8 $, respectively, and the clean accuracy for the corresponding models are $84.47\% , 80.99\%,76.01\%, 68.65\%,63.07\% , 60.65\%$ (Figure~\ref{fig:ins_ceracc_k_eps_cifar} (d)(e)).

With the intention of exploring the upper bound for $k$ given $\tau$ under different instance-level DP guarantee, for \mnist{}, we set noise level $\sigma$ to be $5,8,10,13,20 $, respectively, to obtain instance-DP FL models after ten rounds with  privacy guarantee $\epsilon= 0.6439,0.3937,0.3172, 0.2626 ,0.2179 $ and  clean accuracy $ 99.58\%, 98.83\%, 97.58\%,  95.23 \%,  85.72\%$ (Figure~\ref{fig:ins_tau}(a)). For \cifar{}, 
 we set noise level $\sigma$ to be $3,4,5,6,7,8$ and train \insdpfedavg{} for $T=3$ rounds to obtain FL models with privacy guarantee $\epsilon=1.5365, 1.1162,0.8777, 0.7238 , 0.6159 ,0.5361 $ and clean accuracy $ 84.34\%, 80.27\%, 74.62\%,66.94\% ,62.14\% ,59.75\%$ (Figure~\ref{fig:ins_tau}(b)).

\begin{rev}

\subsubsection{Detailed Setup for Different User-level DPFL Algorihtms}
\label{app:exp_details_user_diff_dpfl_algos}
For \mnist{} (\cifar{}, \sent{}), we set $\epsilon$ to be 0.6319 (0.5346, 0.4089), which is obtained by training all DPFL algorithms with the same noise level $\sigma= 2.3$ ($\sigma=  3.0$, $\sigma=  2.0$) for same number of rounds. For flat clipping and per-layer clipping, we set $S=0.7$ ($S=1$, $S=0.5$) on \mnist{} (\cifar{}, \sent{}). Except for local epoch $E= 1$, other FL parameter setups are the same as in \cref{tb:dataset_parameters}. 
We set $E= 1$ because we find that the FL model in our experiments can  be trained with median norm clipping approaches~\cite{geyer2017differentially} only when the number of the local epoch is small.
Recall that in the server aggregation step, the noise is sampled from $\mathcal{N}(0, \sigma^2S^2)$, so $S$ cannot be too large in order to keep the amount of noise reasonable and preserve a good model utility. As more local epoch leads to a larger norm of model updates, we set the local epoch as 1 to keep the median norm small.
    
\end{rev}

\subsection{Additional Experimental Results}

\begin{rev}

\subsubsection{Running Time Analysis for the Certifications.}\label{app:running_time_overhead}
Compared to non-DP FL, the mechanisms introduced by DPFL, i.e., clipping and noise addition, are low-cost and easy to implement. In our experiments, the averaged running time for each communication round on \sent{} dataset is 6.06s for FedAvg and 6.11s for  \userdpfedavg{} (averaged over 1000 times), based on a Linux machine with Intel 8 Core i7-7820X CPU and 4 NVidia 2080Ti GPUs.  
The major overhead of our certifications comes from re-training the DPFL algorithm $O$ times for Monte-Carlo approximation (see \cref{subsubsec:certification_monte_carlo_approx}). Notably,  re-training is a common requirement when providing certifications against poisoning attacks~\cite{weber2020rab,rosenfeld2020certifiedlableflip}. Also, multiple runs of training are parallelizable and can be speeded up with multiple GPUs. 
Given all trained models and the inference results from each model, running the certifications (e.g., averaging class confidence, and making predictions) has negligible costs, which is 0.04s on the \sent{} dataset. 

\subsubsection{Certifications with Moderate Overall Privacy Budget.}
\label{app:add-practical-consideration}
Certified robustness can be achieved under a moderate overall privacy budget and robustness confidence. As shown in the \cref{fig:practical_eps_cifar},  on \cifar{}, when $\epsilon=0.1451$ and $O=70$, the overall privacy cost is about $\epsilon O =10.15$. Under the confidence level of 80\%, the maximal number of adversaries that can be certified is about $k=4$. 

\begin{table*}[!htbp]
\caption{\revise{Comparison of empirical robust accuracy between our certification approach and  empirical FL defenses against state-of-the-art poisoning attacks on \cifar{}.  
``\userdpfedavg{}-cert'' denotes our certification approach based on \userdpfedavg{}. 
\userdpfedavg{}-cert provides similar or even higher empirical robust accuracy than empirical defenses. The certified accuracy of \userdpfedavg{}-cert serves as the lower bound for its empirical robust accuracy.}}
\label{tb:empircal_cifar_sota_poisoning}
\scalebox{0.7}{
\begin{tabular}{ccccccccccccccc}
  \toprule
 &  & \multicolumn{5}{c}{{$k$=2}} & \multicolumn{5}{c}{{$k$=3}} \\
  \cmidrule(lr){3-7}  \cmidrule(lr){8-12} 
 &  & \multicolumn{4}{c}{Empirical Robust Acc.} & \multicolumn{1}{l}{\multirow{2}{*}{\makecell{Certified\\Robust Acc.}}} & \multicolumn{4}{c}{Empirical Robust Acc.} & \multicolumn{1}{l}{\multirow{2}{*}{\makecell{Certified\\Robust Acc.}}} \\
  \cmidrule(lr){3-6} \cmidrule(lr){8-11} 
 & No Attack & \makecell{STAT-OPT~\cite{shejwalkar2021manipulating}\\(Min-Max)}& \makecell{STAT-OPT~\cite{shejwalkar2021manipulating}\\(Min-Sum)} & \makecell{BKD~\cite{bagdasaryan2020backdoor}\\($\gamma=100$)} & \makecell{LF~\cite{bhagoji2018analyzing}\\($\gamma=100$)} &  & \makecell{STAT-OPT~\cite{shejwalkar2021manipulating}\\(Min-Max)}& \makecell{STAT-OPT~\cite{shejwalkar2021manipulating}\\(Min-Sum)} & \makecell{BKD~\cite{bagdasaryan2020backdoor}\\($\gamma=100$)} & \makecell{LF~\cite{bhagoji2018analyzing}\\($\gamma=100$)} &  \\
  \midrule
FedAvg~\cite{mcmahan2016communication} & \textbf{88.08}\% & 87.29\% & 87.35\% & 65.73\% & 65.47\% & / & 86.36\% & 86.55\% & 58.39\% & 58.07\% & / \\
Median~\cite{yin2018byzantine} & 87.76\% & 87.09\% & 87.16\% & 87.73\% & 87.74\% & / & 86.22\% & 86.42\% & 87.74\% & 87.75\% & / \\
Trimmed-mean~\cite{yin2018byzantine} & \textbf{88.08}\% & 87.28\% & 87.35\% & 87.98\% & 87.98\% & / & 86.36\% & 86.55\% & 87.94\% & 87.94\% & / \\
Krum~\cite{blanchard2017machinekrum} & 85.97\% & 85.84\% & 85.96\% & 85.87\% & 85.87\% & / & 85.12\% & 85.4\% & 85.85\% & 85.85\% & / \\
Multi-Krum~\cite{blanchard2017machinekrum}\% & 88.02\% & 87.23\% & 87.29\% & 87.99\% & \textbf{87.99}\% & / & 86.31\% & 86.51\% & \textbf{87.98}\% & \textbf{87.98}\% & / \\
Bulyan~\cite{el2018hidden} & 88.02\% & 87.24\% & 87.3\% & 87.93\% & 87.94\% & / & 86.31\% & 86.52\% & 87.89\% & 87.89\% & / \\
RFA~\cite{pillutla2019robustrfa} & 87.97\% & 87.21\% & 87.28\% & 87.94\% & 87.94\% & / & 86.29\% & 86.49\% & 87.96\% & 87.95\% & / \\
\rowcolor{lightgray}\userdpfedavg{}-cert ($\epsilon=0.7693$) & 88.05\% & 87.65\% & \textbf{88}\% & \textbf{88.05}\% & 87.8\% & 17.65\% & \textbf{87.15}\% & 87.5\% & 87.8\% & 87.85\% & 1.4\% \\
\rowcolor{lightgray}\userdpfedavg{}-cert ($\epsilon=0.648$) & 87.35\% & \textbf{87.8}\% & 87.6\% & 87.9\% & 87.5\% & 28.15\% & 86.45\% & \textbf{87.6}\% & 87.2\% & 87.6\% & 4.3\% \\
\rowcolor{lightgray}\userdpfedavg{}-cert ($\epsilon=0.5346$) & 86.45\% & 86.5\% & 87\% & 87.15\% & 86.8\% & 37.75\% & 87.05\% & 86.65\% & 87.15\% & 87.15\% & 11.45\% \\
\rowcolor{lightgray}\userdpfedavg{}-cert ($\epsilon=0.3205$) & 85.2\% & 85.15\% & 86.05\% & 85.1\% & 85.7\% & \textbf{48.5}\% & 83.9\% & 85.85\% & 85.8\% & 84.95\% & \textbf{21.85}\%\\
\midrule
\midrule
\end{tabular}
}
\bigskip
\scalebox{0.7}{
\begin{tabular}{ccccccccccccccc}
 &  & \multicolumn{5}{c}{{$k$=5}} & \multicolumn{5}{c}{{$k$=10}} \\
 \cmidrule(lr){3-7}  \cmidrule(lr){8-12} 
 &   & \multicolumn{4}{c}{Empirical Robust Acc.} & \multicolumn{1}{l}{\multirow{2}{*}{\makecell{Certified\\Robust Acc.}}} & \multicolumn{4}{c}{Empirical Robust Acc.} & \multicolumn{1}{l}{\multirow{2}{*}{\makecell{Certified\\Robust Acc.}}} \\ 
 \cmidrule(lr){3-6} \cmidrule(lr){8-11} 
 & No Attack  & \makecell{STAT-OPT~\cite{shejwalkar2021manipulating}\\(Min-Max)}& \makecell{STAT-OPT~\cite{shejwalkar2021manipulating}\\(Min-Sum)} & \makecell{BKD~\cite{bagdasaryan2020backdoor}\\($\gamma=100$)} & \makecell{LF~\cite{bhagoji2018analyzing}\\($\gamma=100$)} &  & \makecell{STAT-OPT~\cite{shejwalkar2021manipulating}\\(Min-Max)}& \makecell{STAT-OPT~\cite{shejwalkar2021manipulating}\\(Min-Sum)} & \makecell{BKD~\cite{bagdasaryan2020backdoor}\\($\gamma=100$)} & \makecell{LF~\cite{bhagoji2018analyzing}\\($\gamma=100$)} &  \\
 \midrule
FedAvg~\cite{mcmahan2016communication} & \textbf{88.08}\% & 84.58\% & 85.75\% & 54.69\% & 54.35\% & / & 80.89\% & 84.52\% & 51.17\% & 51.21\% & / \\
Median~\cite{yin2018byzantine} & 87.76\% & 84.5\% & 85.67\% & 87.69\% & 87.69\% & / & 80.86\% & 84.5\% & 87.56\% & 87.56\% & / \\
Trimmed-mean~\cite{yin2018byzantine} & \textbf{88.08}\% & 84.58\% & 85.75\% & 87.8\% & 87.8\% & / & 80.89\% & 84.52\% & 87.44\% & 87.43\% & / \\
Krum~\cite{blanchard2017machinekrum} & 85.97\% & 83.78\% & 85\% & 85.85\% & 85.85\% & / & 80.62\% & 84.29\% & 85.89\% & 85.88\% & / \\
Multi-Krum~\cite{blanchard2017machinekrum} & 88.02\% & 84.54\% & 85.72\% & \textbf{87.94}\% & \textbf{87.95}\% & / & 80.88\% & 84.52\% & \textbf{87.92}\% & \textbf{87.92}\% & / \\
Bulyan~\cite{el2018hidden} & 88.02\% & 84.54\% & 85.72\% & 87.79\% & 87.79\% & / & 80.88\% & 84.52\% & 87.66\% & 87.66\% & / \\
RFA~\cite{pillutla2019robustrfa} & 87.97\% & 84.54\% & 85.71\% & 87.93\% & 87.93\% & / & 80.87\% & 84.51\% & 87.82\% & 87.82\% & / \\
\rowcolor{lightgray}\userdpfedavg{}-cert ($\epsilon=0.7693$) & 88.05\% & \textbf{86.2}\% & \textbf{86.35}\% & 87.4\% & 87.3\% & 0\% & \textbf{85.25}\% & \textbf{86.5}\% & 86.95\% & 86.75\% & 0\% \\
\rowcolor{lightgray}\userdpfedavg{}-cert ($\epsilon=0.648$) & 87.35\% & \textbf{86.2}\% & 86.3\% & 87.15\% & 87.4\% & 0\% & 85.1\% & 85.75\% & 86.75\% & 85.85\% & 0\% \\
\rowcolor{lightgray}\userdpfedavg{}-cert ($\epsilon=0.5346$) & 86.45\% & 85.6\% & 86.1\% & 87.05\% & 87.1\% & 0\% & 84.65\% & 85.2\% & 86.65\% & 85.1\% & 0\% \\
\rowcolor{lightgray}\userdpfedavg{}-cert ($\epsilon=0.3205$) & 85.2\% & 83.4\% & 85.25\% & 84.5\% & 85.35\% & 0.35\% & 82.35\% & 84.95\% & 84.2\% & 85.6\% & 0\% \\
\bottomrule
\end{tabular}
}
\end{table*}

\setlength{\columnsep}{0pt}%
\begin{figure}{}
  \begin{center}
    \resizebox{0.6\linewidth}{!}{\includegraphics{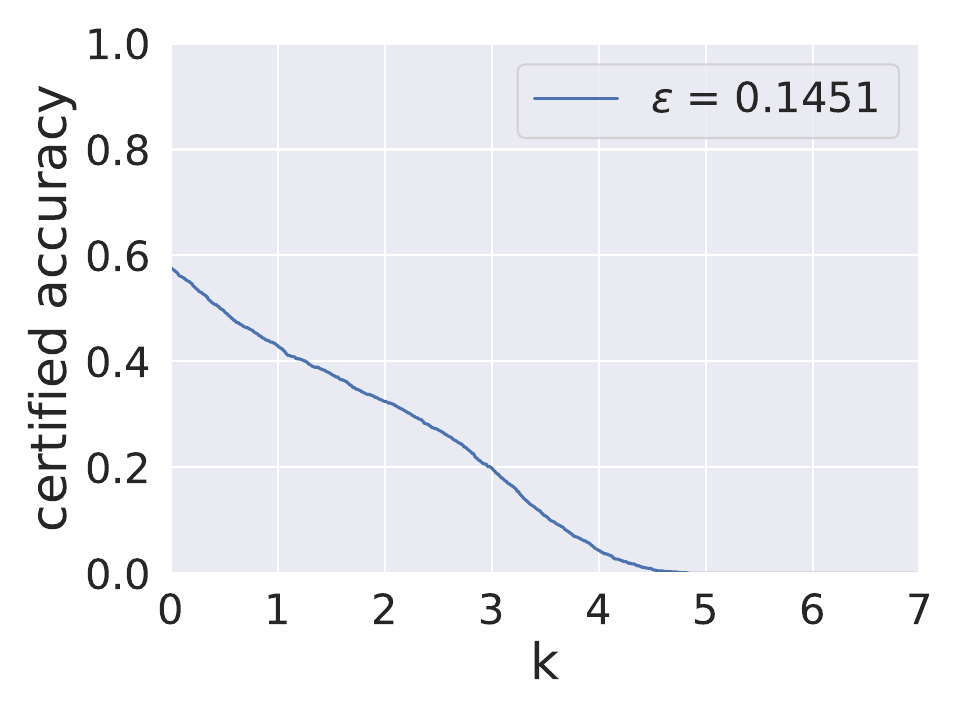}}
  \end{center}
  \caption{\revise{Certified accuracy of \userdpfedavg{} on \cifar{}  under $80\%$ confidence with $\epsilon O=10.15$.}}
  \label{fig:practical_eps_cifar}
\end{figure}

\subsubsection{Empirical Robust Accuracy against State-of-the-art Poisoning Attacks.}
\label{app:sota_poison_empirical}
In this section, we evaluate our certification method against state-of-the-art poisoning attacks and report the empirical accuracy and certified accuracy. Specifically, we consider the following  attacks. 
Static Optimization ({STAT-OPT}) attack~\cite{shejwalkar2021manipulating} solves adversarial optimization problems to find optimal poisoned local model updates. We consider the ``agnostic'' setting of STAT-OPT attack,
where the gradients of benign devices and the server's aggregation algorithm are unknown to the attacker, based on the attacker's knowledge of our settings. We evaluate two variants of STAT-OPT attack:
\textit{STAT-OPT (Min-Max)} and \textit{STAT-OPT (Min-Sum)}; for details please refer to \cite{shejwalkar2021manipulating}. 
We also consider backdoor attack (\textit{BKD}) and label flipping attack (\textit{LF})   under {model replacement} strategy with a scale factor $\gamma$ to boost malicious local update~\cite{bagdasaryan2020backdoor,bhagoji2018analyzing}.
For our \userdpfedavg{} certification approach, denoted as \textit{\userdpfedavg{}-cert}, the prediction for each test sample is calculated based on  Equation~\ref{eq:prediction_certification}, and we train \userdpfedavg{} algorithms $O=100$ times for Marto-Carlo approximation of the expected class confidence in Equation~\ref{eq:prediction_certification}. 

From \cref{tb:empircal_cifar_sota_poisoning}, we see that the empirical robust accuracy of our certification method on \cifar{} is high and remains stable in the presence of $k=2,3,5,10$ attackers under various attacks (i.e., less than 1\%$\sim$2\% accuracy drop compared with the no-attacker setting). It shows that our DPFL certification is empirically robust against poisoning attacks.

\cref{tb:empircal_cifar_sota_poisoning} also shows that the certified accuracy of \userdpfedavg{}-cert serves as the lower bound for its empirical robust accuracy. 
We notice that under relatively strong attack settings such as $k=5, 10$,  our DPFL certification cannot provide non-trivial certified accuracy. Nevertheless, \textit{our DPFL certification approach still  exhibits strong empirical effective robustness},  even without theoretical guarantees.   
The gap between certified robust accuracy and empirical robust accuracy indicates potential advancements either in crafting stronger poisoning attacks to further reduce empirical robust accuracy,  or in developing tighter robustness certification techniques to improve theoretical lower bound.

\subsubsection{Comparison to Empirical FL Defenses.}
\label{app:sota_fl_defenses_empirical}
Here, we compare the \textit{empirical} robust accuracy of our certification method with \textit{six} FL robust aggregations, including \textit{Krum}~\cite{blanchard2017machinekrum}, \textit{Multi-krum}~\cite{blanchard2017machinekrum}, \textit{Trimmed-mean}~\cite{yin2018byzantine}, \textit{Median}~\cite{yin2018byzantine}, \textit{Bulyan}~\cite{el2018hidden}, \textit{RFA}~\cite{pillutla2019robustrfa}.

show that our certification method {\userdpfedavg{}-cert} achieves similar and even higher accuracy than empirical defenses under state-of-the-art poisoning attacks, while providing privacy and robustness guarantees.  
Specifically, under the optimization-based attacks  STAT-OPT (Min-Max) and  STAT-OPT (Min-Sum), {\userdpfedavg{}-cert} consistently achieves higher empirical robust accuracy than other FL robust aggregation methods when $k=2,3,5,10$; 
under BKD and LF attacks, {\userdpfedavg{}-cert} exhibits similar robustness as FL robust aggregation methods.
 Note that Multi-Krum,  Trimmed-mean, and Bulyan require specifying the number of attackers in their aggregation rules to detect the outliers, while our approach does not require such knowledge about attackers during DPFL training.

We also notice that when $\epsilon$ is too small (e.g., $\epsilon=0.3205$), {\userdpfedavg{}-cert}  has lower empirical robust accuracy than robust aggregation defenses. This is mainly because of the noise level being large during \userdpfedavg{} training to achieve a strong privacy guarantee, which hurts the utility of the DPFL model, as we can see in the no-attack setting. Therefore, we recommend adopting a reasonable $\epsilon$ with good utility to achieve robustness, as elaborated in \cref{sec:eval_user_level}.

\end{rev}

\subsubsection{Additional Robustness Evaluation of User-level DPFL}
\label{app:add-user-dpfl_exp_results}

{
\begin{figure}[t]
\newlength{\userbkdepsheightc}
\settoheight{\userbkdepsheightc}{\includegraphics[width=0.45\linewidth]{figures/ccsfinal_plots/cer_acc_conf/cer_acc_mnist.pdf}}

\newcommand{\rowname}[1]%
{\rotatebox{90}{\makebox[\userbkdepsheightc][c]{\footnotesize #1}}}

\centering

{
\renewcommand{\tabcolsep}{10pt}
\begin{subtable}[]{\linewidth}
\centering
\begin{tabular}{@{}p{5mm}@{}c@{}c@{}c@{}c@{}c@{}c@{}c@{}}
        & \makecell{\shortstack{\footnotesize (a) \mnist{} (k=4)}}
             & \makecell{\shortstack{\footnotesize (b) \cifar{} (k=4)}}
        \vspace{-1.7pt}\\
\rowname{\makecell{$J(D')$}}&
\includegraphics[height=\userbkdepsheightc]{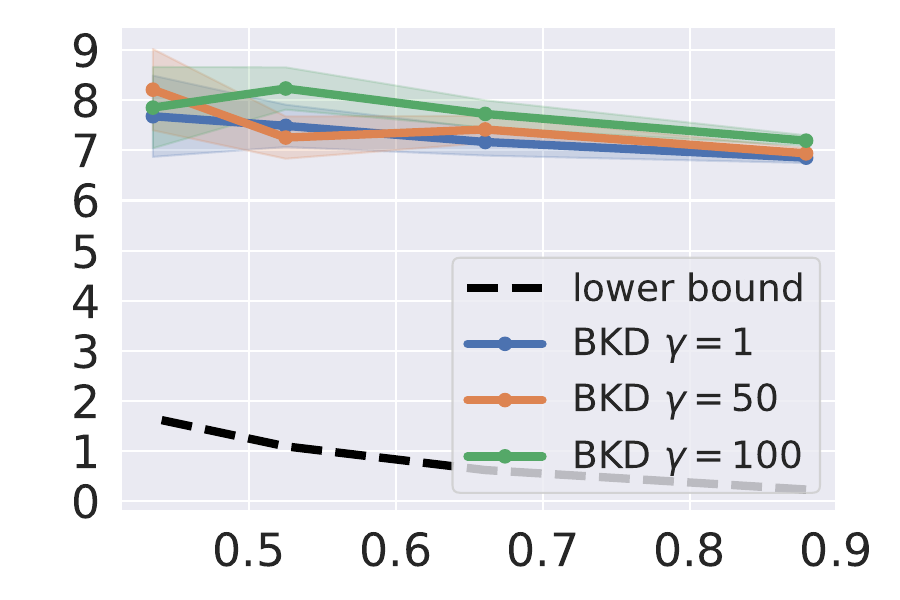}&
\includegraphics[height=\userbkdepsheightc]{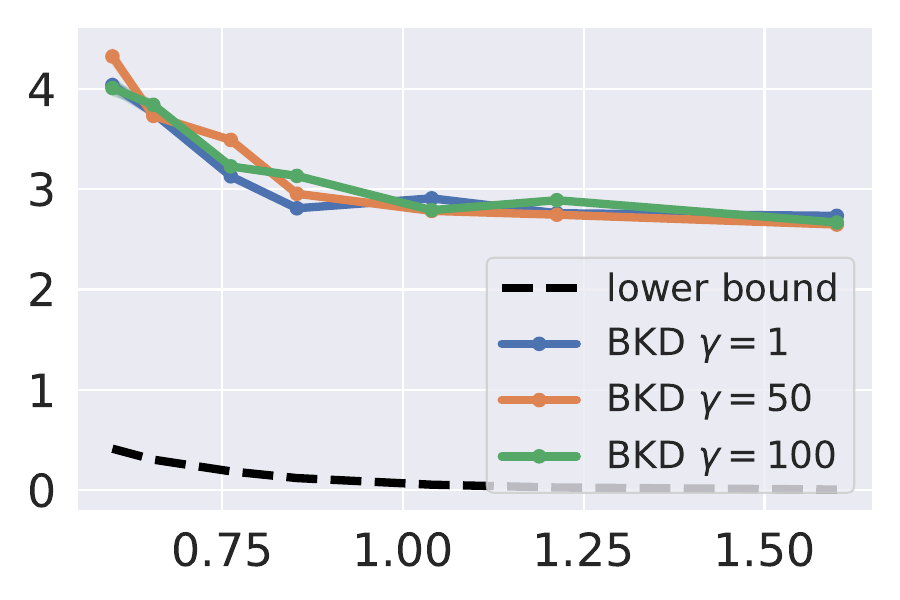}&\\
\\[-5ex]
& \makecell{{\footnotesize $\epsilon$}}& \makecell{{\footnotesize $\epsilon$}} 
\end{tabular}
\end{subtable}
}
\vspace{-3mm}
\caption{\small Certified \attackcost of \userdpfedavg{} with different $\epsilon$ under backdoor attack.}
\label{fig:user_eps_backdoor} 
\end{figure}

}

Here we further explore the impacts of $\epsilon$ on the certified \attackcost.
Similar to the results of label flipping attacks in Figure~\ref{fig:user_eps_tau} (a-c), the results of backdoor attacks in Figure~\ref{fig:user_eps_backdoor} show that  as the privacy guarantee becomes stronger, i.e. smaller $\epsilon$, the model is more robust, achieving higher $J(D')$ and $\underline{J(D')}$.

\subsubsection{Robustness Evaluation on 10-class Classification}
\label{sec:ten_class_results}
Here we report the robustness evaluation of user-level DPFL under backdoor attacks on 10-class classification problems.
Figure~\ref{fig:cer_acc_10class} presents the certified accuracy under different $\epsilon$. We observe the interplay between $\epsilon$ and certified accuracy on \mnist{}. On \cifar{}, larger $k$ can be certified with smaller $\epsilon$. The certified $\cerk$ is relatively small because we set large $\epsilon$ to preserve a reasonable accuracy for 10-class classification. 
Our results suggest that advanced DP mechanisms would be preferred to provide tighter privacy guarantees (i.e., smaller $\epsilon$) while achieving a similar level of accuracy.
 In terms of certified \attackcost, as shown in \cref{fig:tau_10class}  and \cref{fig:user_k_eps_10class}, the trends are similar to the 2-class results in   \cref{fig:user_eps_tau} and \cref{fig:user_k_gamma_alpha},

{
\begin{figure}[t]
\newlength{\tenceraccheight}
\settoheight{\tenceraccheight}{\includegraphics[width=0.5\linewidth]{figures/ccsfinal_plots/cer_acc_conf/cer_acc_mnist.pdf}}

\newcommand{\rowname}[1]%
{\rotatebox{90}{\makebox[\tenceraccheight][c]{\footnotesize #1}}}

\centering

{
\renewcommand{\tabcolsep}{10pt}
\begin{subtable}[]{\linewidth}
\centering
\begin{tabular}{@{}p{5mm}@{}c@{}c@{}c@{}c@{}c@{}c@{}c@{}}
        & \makecell{\shortstack{\footnotesize (a) \mnist{} (k=4)}}
             & \makecell{\shortstack{\footnotesize (b) \cifar{} (k=4)}}
        \vspace{-1.7pt}\\
\rowname{\makecell{Certified accuracy}}&
\includegraphics[height=\tenceraccheight]{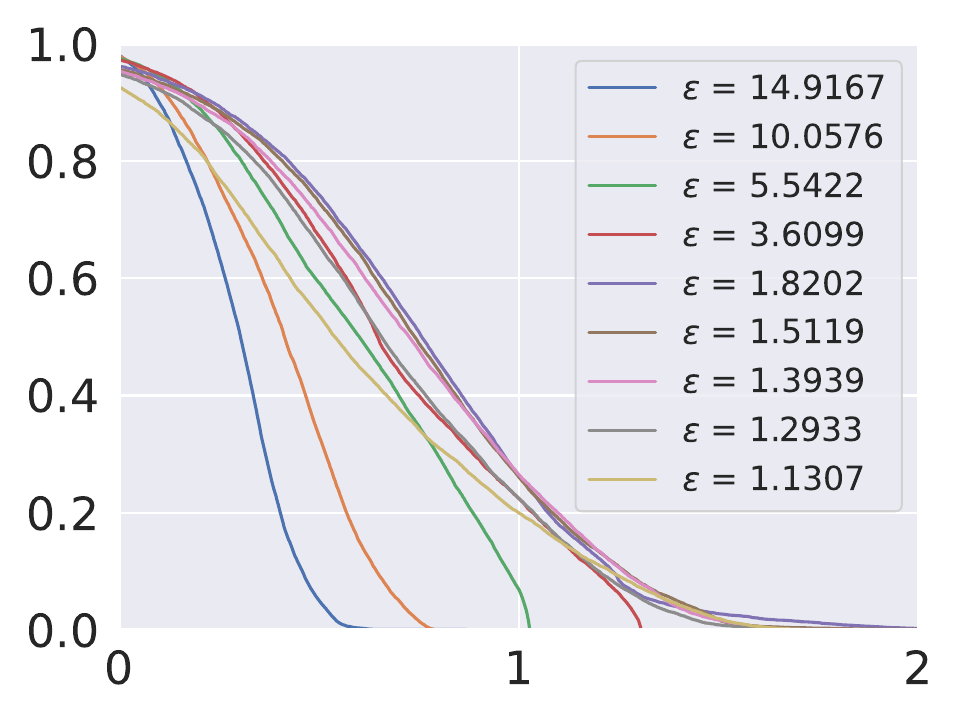}&
\includegraphics[height=\tenceraccheight]{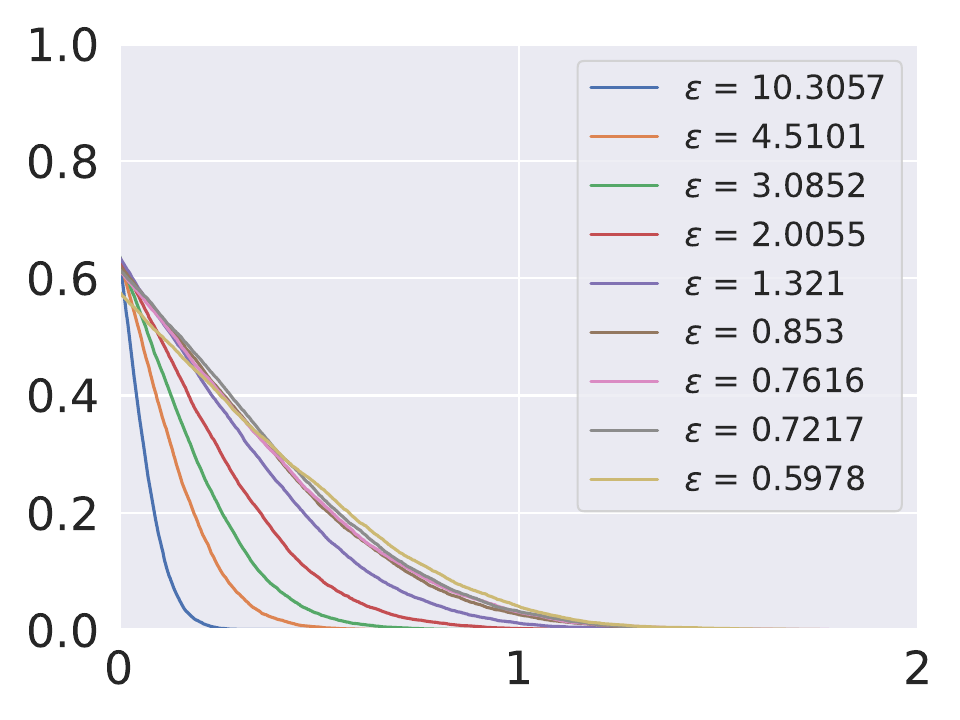}&\\
\\[-5ex]
& \makecell{{\footnotesize $k$}}& \makecell{{\footnotesize $k$}} 
\end{tabular}
\end{subtable}
}
\vspace{-3mm}
\caption{\small Certified accuracy of FL \userdpfedavg{} on 10-class classification.}
\label{fig:cer_acc_10class}
\end{figure}

}

{
\begin{figure}[t]
\newlength{\tentauheight}
\settoheight{\tentauheight}{\includegraphics[width=0.42\linewidth]{figures/ccsfinal_plots/cer_acc_conf/cer_acc_mnist.pdf}}

\newcommand{\rowname}[1]%
{\rotatebox{90}{\makebox[\tentauheight][c]{\footnotesize #1}}}

\centering

{
\renewcommand{\tabcolsep}{10pt}
\begin{subtable}[]{\linewidth}
\centering
\begin{tabular}{@{}p{5mm}@{}c@{}c@{}c@{}c@{}c@{}c@{}c@{}}
        & \makecell{\shortstack{\footnotesize (a) \mnist{}}}
             & \makecell{\shortstack{\footnotesize (b) \cifar{}}}
        \vspace{-1.7pt}\\
\rowname{\makecell{$k$}}&
\includegraphics[height=\tentauheight]{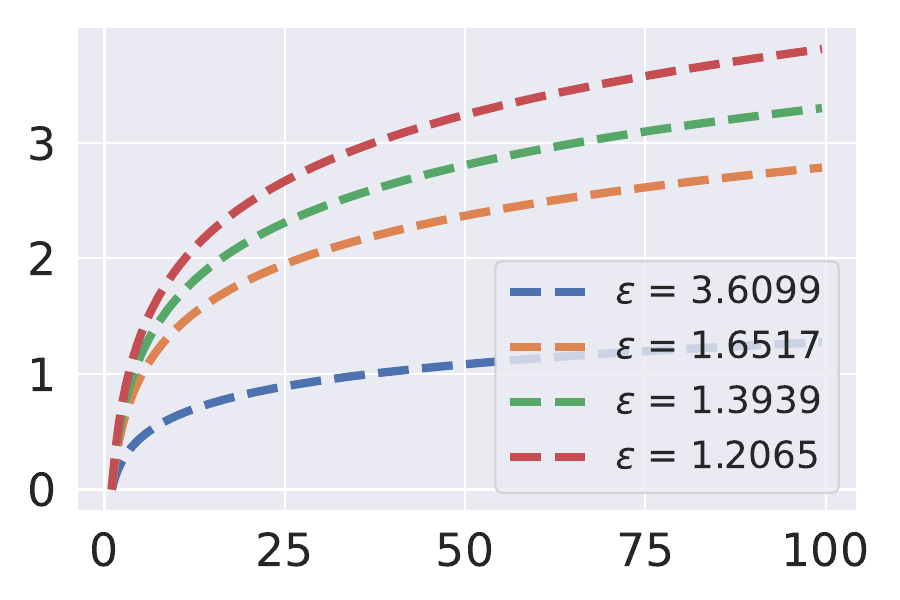}&
\includegraphics[height=\tentauheight]{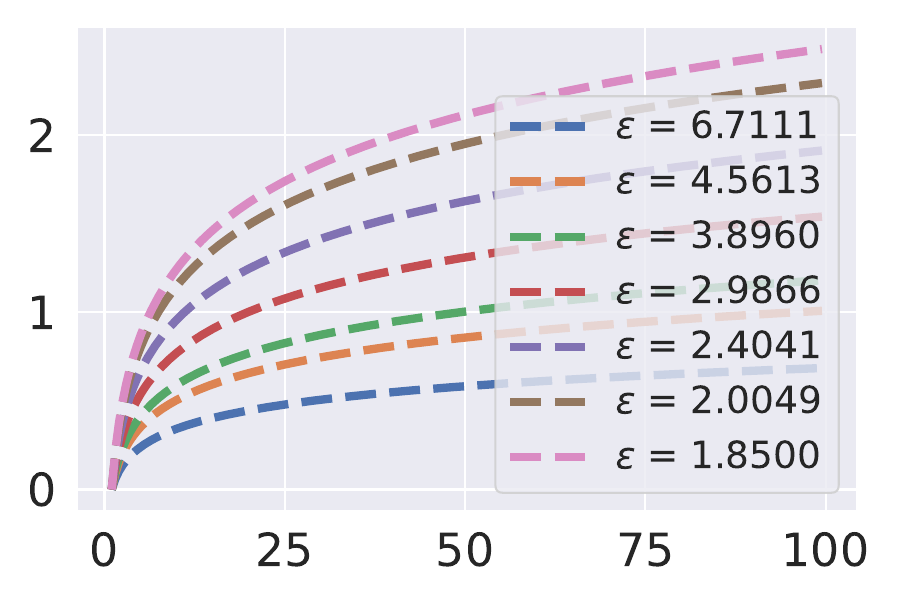}&\\
\\[-5ex]
& \makecell{{\footnotesize $\tau$}}& \makecell{{\footnotesize $\tau$}} 
\end{tabular}
\end{subtable}
}
\vspace{-3mm}
\caption{\small Lower bound of $k$ on 10-class classification under user-level $\epsilon$ given attack effectiveness $\tau$.}
\label{fig:tau_10class} 
\end{figure}

}

{
\begin{figure*}
\newlength{\tencostbkdheightc}
\settoheight{\tencostbkdheightc}{\includegraphics[width=.23\linewidth]{figures/ccsfinal_plots/cer_acc_conf/cer_acc_mnist.pdf}}

\newcommand{\rowname}[1]%
{\rotatebox{90}{\makebox[\tencostbkdheightc][c]{\footnotesize #1}}}

\centering

{
\renewcommand{\tabcolsep}{10pt}
\begin{subtable}[]{\linewidth}
\centering
\begin{tabular}{@{}p{5mm}@{}c@{}c@{}c@{}c@{}c@{}c@{}c@{}}
        & \makecell{{\footnotesize (a) \mnist{} BKD ($\epsilon=0.67$)}}
        & \makecell{{\footnotesize (b) \cifar{} BKD ($\epsilon=0.12$)}}
        & \makecell{{\footnotesize (c) \mnist{} BKD $k=3$}}
        & \makecell{{\footnotesize (d) \cifar{} BKD $k=1$}}
        \vspace{-1.7pt}\\
\rowname{\makecell{$J(D')$}}&
\includegraphics[height=\tencostbkdheightc]{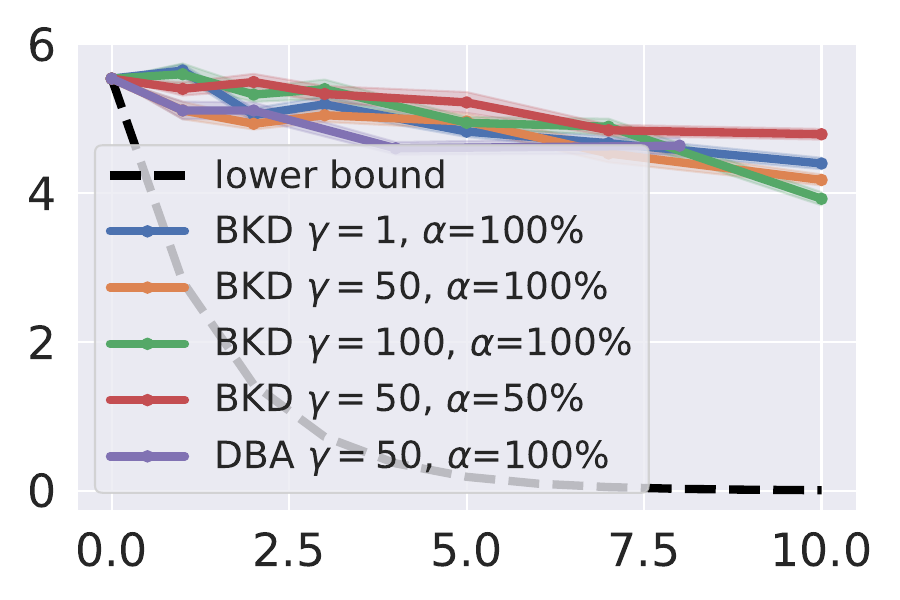}&
\includegraphics[height=\tencostbkdheightc]{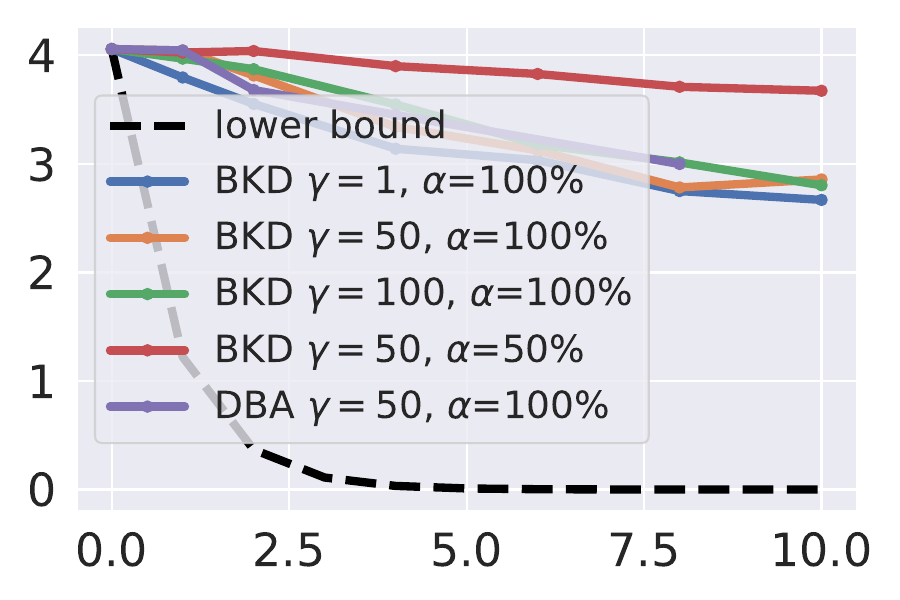}&
\includegraphics[height=\tencostbkdheightc]{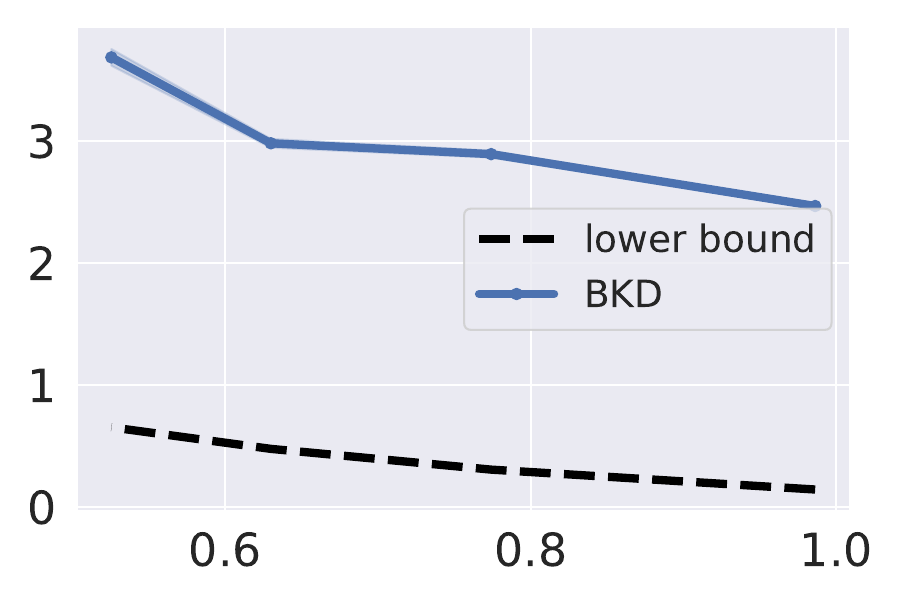}&
\includegraphics[height=\tencostbkdheightc]{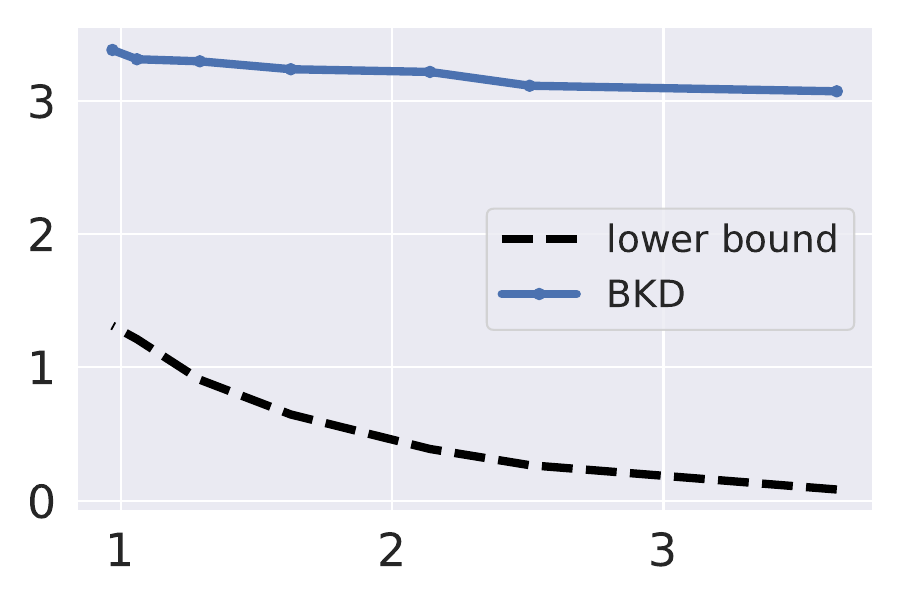}\\[-1.2ex]
& \makecell{{\footnotesize $k$}}& \makecell{{\footnotesize $k$}}& \makecell{\footnotesize $\epsilon$}& \makecell{\footnotesize $\epsilon$} \\
\end{tabular}
\end{subtable}
}
\vspace{-2mm}
\caption{\small Certified \attackcost of \userdpfedavg{} on 10-class classification given the different number of malicious instances $k$ (a)(b) and different $\epsilon$ (c)(d).}
\label{fig:user_k_eps_10class} 
\end{figure*}

}

\section{Proofs of Certified Robustness Analysis}
\label{sec:app_proofs}

We restate our \cref{lemma_group_dp} here.

\lemmagroupdp*
\begin{proof}
We denote $d$ as $d_0$, $d'$ as $d_k$.  $d_i$ differ $i$ individuals with $d_0$. 
For any $i \in [1,k]$, $d_i$ and $d_{i-1}$ differ by one individual, thus
\begin{equation}
\label{eq:iter_group_dp}
	\Pr[M(d_{i-1})] \leq e^{\epsilon} \Pr[M(d_{i})]+\delta.
\end{equation}
By iteratively applying Eq.~(\ref{eq:iter_group_dp}) $k$ times, we have
\begin{align*}
\Pr[M(d_0)] &\leq e^{k\epsilon} \Pr[M(d_k)] + (1+e^{\epsilon}+e^{2\epsilon}+\ldots+ e^{(k-1)\epsilon}) \delta  \\
&= e^{k\epsilon} \Pr[M(d_k)] + \frac{1-e^{k \epsilon }}{1-e^\epsilon}\delta 
\end{align*}
\end{proof}

Before we prove  Theorem~\ref{thm_pred_consist_one_client}, we introduce the following lemma:
\begin{lemma}
\label{lemma:class_conf}
Suppose a randomized mechanism $\Mcal$ satisfies user-level $(\epsilon, \delta)$-DP. For two user sets $B$ and $B^\prime$ that differ by one user, $D$ and $D'$ are the corresponding training datasets. 
For a test input $x$, for any $c \in [C]$ , $f_c(\mathcal{M}(D),x)\in [0,1]$ is the class confidence, then the expected class confidence $F_{c}(\mathcal{M}(D), x):= \mathbb{E}[f_{c}(\mathcal{M}(D),x)]$ meets the following property: 
\begin{align} \label{eq:expected_class_conf_bound}
	F_{c}(\mathcal{M}(D), x)  \leq  e^{\epsilon} F_{c}(\mathcal{M}(D'), x)+  \delta
\end{align}
\end{lemma}

\begin{proof}
Define $\Theta (a) :=\{\theta: f_{c}(\theta,x)>a\}$. Then
\begin{align*}
	F_{c}(\mathcal{M}(D), x)  &= \mathbb{E}[f_{c}(\mathcal{M}(D),x)] \\
	&=\int_{0}^{1}  \mathbb{P} \left[f_{c}(\mathcal{M}(D),x) > a \right] da \\
	&=\int_{0}^{1} \mathbb{P} \left[\mathcal{M}(D) \in \Theta(a) \right] da \\
	&\leq \int_{0}^{1} \left(e^{\epsilon}\mathbb{P} \left[ \mathcal{M}(D') \in \Theta(a) \right]+ \delta \right) da\\
	&=\int_{0}^{1} e^{\epsilon} \mathbb{P}\left[f_{c}(\mathcal{M}(D'),x) > a \right] da  + \int_{0}^{1} \delta da\\
	&= e^{\epsilon} F_{c}(\mathcal{M}(D'), x)  + \delta 
\end{align*}
\end{proof}

We recall Theorem~\ref{thm_pred_consist_one_client}.
\thmonecertpred*

\begin{proof}
According to Lemma~\ref{lemma:class_conf}, 
	\begin{equation} \label{eq:class_conf_A}
F_{\Aclass}(\mathcal{M}(D), x) \leq e^{\epsilon} F_{\Aclass}(\mathcal{M}(D'), x) + \delta
\end{equation}
\begin{equation} \label{eq:class_conf_B}
	F_{\Bclass}(\mathcal{M}(D'), x) \leq e^{\epsilon} 	F_{\Bclass}(\mathcal{M}(D), x) + \delta.
\end{equation}
Then 
\begin{align*}
	 F_{\Aclass}(\mathcal{M}(D'), x)   &\geq  \frac{ F_{\Aclass}(\mathcal{M}(D), x) -\delta }{e^{\epsilon}}  \tag*{(Because of Eq.~\ref{eq:class_conf_A})}\\
	&\geq  \frac{ e^{2\epsilon} F_{\Bclass}(\mathcal{M}(D), x)  +  (1+ e^{\epsilon})\delta -\delta }{e^{\epsilon}} \tag*{(Because of the given condition Eq.~\ref{pred_cons_condition})} \\ 
	& =   e^{\epsilon} F_{\Bclass}(\mathcal{M}(D), x) +   \delta  \\
	&\geq  e^{\epsilon} \left( \frac{ F_{\Bclass}(\mathcal{M}(D'), x)  - \delta}{e^{\epsilon}} \right) +   \delta  \tag*{(Because of Eq. \ref{eq:class_conf_B})}\\
	& =   F_{\Bclass}(\mathcal{M}(D'), x), 
\end{align*}
which indicates that the prediction of $\mathcal{M}(D')$ at $x$ is $\Aclass$ by definition.
\end{proof}

Before we prove Theorem~\ref{thm_pred_consist_k_client}, we introduce the following lemma:

\begin{lemma}
\label{lemma:class_conf_groupdp}
Suppose a randomized mechanism $\Mcal$ satisfies user-level $(\epsilon, \delta)$-DP. For two user sets $B$ and $B^\prime$ that differ $k$ users, $D$ and $D'$ are the corresponding training datasets.  
For a test input $x$, for any $c \in [C]$ , $f_c(\mathcal{M}(D),x)\in [0,1]$ is the class confidence, then the expected class confidence $F_{c}(\mathcal{M}(D), x):= \mathbb{E}[f_{c}(\mathcal{M}(D),x)]$ meets the following property: 
\begin{align} \label{eq:expected_class_conf_bound_groupdp}
	F_{c}(\mathcal{M}(D), x)  \leq  e^{k\epsilon} F_{c}(\mathcal{M}(D'), x)+ \frac{1-e^{k \epsilon }}{1-e^\epsilon}\delta.
\end{align}
\begin{rev}

and
\begin{align} 
	F_{c}(\mathcal{M}(D'), x)  \leq  e^{k\epsilon} F_{c}(\mathcal{M}(D), x)+ \frac{1-e^{k \epsilon }}{1-e^\epsilon}\delta. \nonumber
\end{align}
\end{rev}
\end{lemma}

\begin{proof}
Define $\Theta (a) :=\{\theta: f_{c}(\theta,x)>a\}$. Then
{\small
\begin{align*}
	F_{c}(\mathcal{M}(D), x)    
	&=\int_{0}^{1}  \mathbb{P} \left[f_{c}(\mathcal{M}(D),x) > a \right] da \\
	&=\int_{0}^{1} \mathbb{P} \left[\mathcal{M}(D) \in \Theta(a) \right] da \\
	&\leq \int_{0}^{1} \left(e^{k\epsilon}\mathbb{P} \left[ \mathcal{M}(D') \in \Theta(a) \right]+  \frac{1-e^{k \epsilon }}{1-e^\epsilon}\delta  \right) da \tag*{(Because of Group DP property in ~\cref{lemma_group_dp})}\\
	&=\int_{0}^{1} e^{k\epsilon} \mathbb{P}\left[f_{c}(\mathcal{M}(D'),x) > a \right] da  + \int_{0}^{1} \frac{1-e^{k \epsilon }}{1-e^\epsilon}\delta da\\
	&= e^{k\epsilon} F_{c}(\mathcal{M}(D'), x)  + \frac{1-e^{k \epsilon }}{1-e^\epsilon} \delta 
\end{align*}

\begin{rev}

Similarly, due to the symmetric property of adjacent datasets in the DP definition (\cref{def:dp}) and Group DP definition (\cref{lemma_group_dp}), $D$ and $D'$ are interchangeable, and therefore we have 
\begin{align*}
	F_{c}(\mathcal{M}(D'), x)    
	&=\int_{0}^{1}  \mathbb{P} \left[f_{c}(\mathcal{M}(D'),x) > a \right] da \\
	&=\int_{0}^{1} \mathbb{P} \left[\mathcal{M}(D') \in \Theta(a) \right] da \\
	&\leq \int_{0}^{1} \left(e^{k\epsilon}\mathbb{P} \left[ \mathcal{M}(D) \in \Theta(a) \right]+  \frac{1-e^{k \epsilon }}{1-e^\epsilon}\delta  \right) da \tag*{(Because of Group DP property in ~\cref{lemma_group_dp})}\\
	&=\int_{0}^{1} e^{k\epsilon} \mathbb{P}\left[f_{c}(\mathcal{M}(D),x) > a \right] da  + \int_{0}^{1} \frac{1-e^{k \epsilon }}{1-e^\epsilon}\delta da\\
	&= e^{k\epsilon} F_{c}(\mathcal{M}(D), x)  + \frac{1-e^{k \epsilon }}{1-e^\epsilon} \delta 
\end{align*}

\end{rev}

}

\end{proof}

We recall Theorem~\ref{thm_pred_consist_k_client}.
\thmkcertpred*

\begin{proof}
According to Lemma~\ref{lemma:class_conf_groupdp}, we have
	\begin{equation} \label{A_group}
F_{\Aclass}(\mathcal{M}(D), x) \leq  e^{k\epsilon} F_{\Aclass}(\mathcal{M}(D'), x) +  \frac{1-e^{k \epsilon }}{1-e^\epsilon}\delta 
\end{equation} 
\begin{equation} \label{B_group}
F_{\Bclass}(\mathcal{M}(D'), x) \leq  e^{k\epsilon} F_{\Bclass}(\mathcal{M}(D), x)  + \frac{1-e^{k \epsilon }}{1-e^\epsilon}\delta.
\end{equation}
We can re-write the given condition $k < \cerk $ according to Eq.~(\ref{eq:cerk}) as 
\begin{align} \label{eq:group_pred_cons_condition}
	e^{2k\epsilon} F_{\Bclass}(\mathcal{M}(D), x)  + (1+e^{k\epsilon}) \frac{1-e^{k \epsilon }}{1-e^\epsilon}\delta   <  F_{\Aclass}(\mathcal{M}(D), x) .  
\end{align}
Then 
{\small
\begin{align*}
	 F_{\Aclass}(\mathcal{M}(D'), x)  &\geq  \frac{  F_{\Aclass}(\mathcal{M}(D), x)   - \frac{1-e^{k \epsilon }}{1-e^\epsilon}\delta   }{e^{k\epsilon}}  \tag*{(Because of Eq. \ref{A_group})}\\
	&>  \frac{ e^{2k\epsilon}  F_{\Bclass}(\mathcal{M}(D), x)  + (1+e^{k\epsilon}) \frac{1-e^{k \epsilon }}{1-e^\epsilon}\delta   -  \frac{1-e^{k \epsilon }}{1-e^\epsilon}\delta  }{e^{k\epsilon}} \tag*{(Because of the given condition Eq.\ref{eq:group_pred_cons_condition})} \\ 
	& =   e^{k\epsilon} F_{\Bclass}(\mathcal{M}(D), x)  +  \frac{1-e^{k \epsilon }}{1-e^\epsilon}\delta  \\
	&\geq  e^{k\epsilon} \left( \frac{ F_{\Bclass}(\mathcal{M}(D'), x) - \frac{1-e^{k \epsilon }}{1-e^\epsilon}\delta 
 }{e^{k\epsilon} } \right) +  \frac{1-e^{k \epsilon }}{1-e^\epsilon}\delta  \tag*{(Because of Eq. \ref{B_group})}\\
	& =   F_{\Bclass}(\mathcal{M}(D'), x),
\end{align*}
which indicates that the prediction of $\mathcal{M}(D')$ at $x$ is $\Aclass$ by definition.
}
\end{proof}

We recall Theorem~\ref{thm_costfunc_k_client}.
\thmcostfunck*

\begin{proof}
We first consider $C (\cdot )\geq 0$. Define $\Theta (a) =\{\theta: C(\theta)>a\}$.
\begin{align*}
	J(D) &= \int_{0}^{\bar C}  \mathbb{P} \left[C(\mathcal{M}(D)) >a \right] da \\
	&=\int_{0}^{\bar C} \mathbb{P} \left[ \mathcal{M}(D)) \in \Theta(a) \right] da \\
	&\leq \int_{0}^{\bar C} \left(e^{k\epsilon}\mathbb{P} \left[ \mathcal{M}(D')) \in \Theta(a) \right]+ \frac{1-e^{k \epsilon }}{1-e^\epsilon}\delta \right) da \tag*{(Because of Group DP property in \cref{lemma_group_dp})}\\
	&= \int_{0}^{\bar C} e^{k\epsilon}\mathbb{P} \left[ \mathcal{M}(D')) \in \Theta(a) \right]  da +  \frac{1-e^{k \epsilon }}{1-e^\epsilon}\delta \bar C  \\
	&=\int_{0}^{\bar C} e^{k\epsilon} \mathbb{P} \left[C(\mathcal{M}(D')) >a \right] da + \frac{1-e^{k \epsilon }}{1-e^\epsilon}\delta \bar C \\
	&= e^{k\epsilon} J(D') +  \frac{1-e^{k \epsilon }}{1-e^\epsilon}\delta \bar C 
\end{align*}
i.e.,
$$
 J(D') \geq e ^{-k\epsilon} J(D)-  \frac{1-e^{-k \epsilon }}{e^\epsilon-1}\delta \bar C. 
$$
Switch the role of $D$ and $D'$, we have 
$$ J(D') \leq e^{k\epsilon} J(D) +  \frac{1-e^{k \epsilon }}{1-e^\epsilon}\delta \bar C.
  $$
Also note that $0\leq J(D') \leq \bar C$ trivially holds due to $ 0 \leq C(\cdot) \leq \bar C$, thus 
\begin{align*}
& \min\{	 e^{k\epsilon} J(D) +  \frac{e^{k \epsilon }-1}{e^\epsilon-1}\delta \bar C , \bar C \}
 \geq  J(D') \\
 & \quad \quad \geq \max \{e ^{-k\epsilon} J(D) - \frac{1-e^{-k \epsilon }}{e^\epsilon-1}\delta \bar C ,0 \}.     
\end{align*}

Next, we consider $C(\cdot)  \leq 0$. Define $\Theta (a) =\{\theta: C(\theta)< a\}$.
\begin{align*}
	J(D) &=- \int_{-\bar C}^{0}  \mathbb{P} \left[C(\mathcal{M}(D)) < a \right] da \\
	&=- \int_{-\bar C}^{0} \mathbb{P} \left[ \mathcal{M}(D)) \in \Theta(a) \right] da \\
	&\geq  - \int_{-\bar C}^{0}  \left(e^{k \epsilon}\mathbb{P} \left[ \mathcal{M}(D')) \in \Theta(a) \right]+ \frac{1-e^{k \epsilon }}{1-e^\epsilon}\delta \right) da \tag*{(Because of Group DP property in \cref{lemma_group_dp})} \\
	&= - \int_{-\bar C}^{0} e^{k \epsilon}\mathbb{P} \left[ \mathcal{M}(D')) \in \Theta(a) \right]  da -  \frac{1-e^{k \epsilon }}{1-e^\epsilon}\delta  \bar C  \\
	&=- \int_{-\bar C}^{0} e^{k \epsilon} \mathbb{P} \left[C(\mathcal{M}(D')) <a \right] da - \frac{1-e^{k \epsilon }}{1-e^\epsilon}\delta \bar C \\
	&= e^{k \epsilon} J(D') - \frac{1-e^{k \epsilon }}{1-e^\epsilon}\delta \bar C 
\end{align*}
i.e.,
$$
 J(D') \leq e^{-k \epsilon} J(D)+ \frac{1-e^{-k \epsilon }}{e^\epsilon-1}\delta \bar C.
$$
Switch the role of $D$ and $D'$, we have 
$$J(D') \geq e^{k \epsilon} J(D) - \frac{1-e^{k \epsilon }}{1-e^\epsilon}\delta \bar C. 
 $$
Also note that $-\bar C \leq J(D') \leq 0 $ trivially holds due to $ - \bar C \leq C(\cdot) \leq 0$, thus 
\begin{align*}
&\min\{ e^{-k \epsilon} J(D)+ \frac{1-e^{-k \epsilon }}{e^\epsilon-1}\delta \bar C  , 0 \}  \geq   J(D') \\
&\quad \quad \geq \max \{ e ^{k\epsilon} J(D) - \frac{ e^{k \epsilon } - 1}{e^\epsilon-1}\delta \bar C  , -\bar C  \} 
\end{align*}
\end{proof}

We recall Corollary~\ref{thm_k_clients_bound}.
\corcostfunctau*

\begin{proof}

We first consider $C (\cdot )\geq 0$. According to the lower bound in Theorem~\ref{thm_costfunc_k_client}, when $B'$ and $B$ differ $k$ users,  
$ J(D') \geq e ^{-k\epsilon} J(D) - \frac{1-e^{-k \epsilon }}{e^\epsilon-1}\delta \bar C 
 $. 
 Since we require $J(D^\prime) \leq  \frac{1}{\tau}  J(D) $,
then $e ^{-k\epsilon} J(D) - \frac{1-e^{-k \epsilon }}{e^\epsilon-1}\delta \bar C 
  \leq \frac{1}{\tau}  J(D)$.
Rearranging gives the result.

Next, we consider $C (\cdot ) \leq 0$. According to the lower bound in  Theorem~\ref{thm_costfunc_k_client}, when $B'$ and $B$ differ $k$ users,  
$ J(D') \geq e ^{k\epsilon} J(D) - \frac{ e^{k \epsilon } - 1}{e^\epsilon-1}\delta \bar C 
 $. 
 Since we require $J(D^\prime) \leq  \tau  J(D) $,
then $ e ^{k\epsilon} J(D) - \frac{ e^{k \epsilon } - 1}{e^\epsilon-1}\delta \bar C  
  \leq \tau  J(D)$.
Rearranging gives the result.

\end{proof}

We note that all the above robustness certification-related proofs are built upon the user-level $(\epsilon,\delta)$-DP property and the Group DP property.
According to Definition~\ref{def:userdp} and Definition~\ref{def:insdp}, the definition of user-level DP and instance-level DP are both induced from DP (Definition~\ref{def:dp}) despite the different definitions of adjacent datasets. 
By applying the definition of instance-level $(\epsilon,\delta)$-DP and following the proof steps of Theorem~\ref{thm_pred_consist_one_client},~\ref{thm_pred_consist_k_client},~\ref{thm_costfunc_k_client} and Corollary~\ref{thm_k_clients_bound}, we can derive similar theoretical conclusions for instance-level DP,  leading to Theorem~\ref{thm:thmsapplytoinsdp} to achieve the certifiably robust FL given the DP property.

\begin{rev}

\section{Certified Robustness Analysis via R\'enyi DP and Randomized Smoothing}
\label{app:renyi_certifications}

\subsection{Preliminary}
We start by providing preliminaries on R\'enyi Differential Privacy~\cite{mironov2017renyi} and the $ f$-divergence-based randomized smoothing~\cite{Dvijotham2020A}, which is a relaxation of $\ell_p$-norm-based randomized smoothing~\cite{cohen2019certifiedrandsmoothing}.
\begin{definition}(R\'enyi Divergence)\label{def:renyi_div}
For two probability distributions $\rho$ and $\nu$, the Rényi divergence of order $\alpha>1$ is
\begin{equation}
D_\alpha(\rho \| \nu) \triangleq \frac{1}{\alpha-1} \log \mathrm{E}_{x \sim \nu}\left(\frac{\rho(x)}{\nu(x)}\right)^\alpha
\end{equation}
\end{definition}

\begin{definition}($(\alpha,\epsilon_{R,\alpha})$-RDP \cite{mironov2017renyi})\label{def:rdp}
A randomized mechanism $\mathcal{M}:\mathcal{D} \rightarrow \Theta$ with domain $\mathcal{D}$ and output set $\Theta$
satisfies $(\alpha,\epsilon_{R,\alpha})$ R\'enyi Differential Privacy (RDP) if for any pair of two adjacent datasets $ d,d' \in \mathcal{D}$, it holds that 
\begin{align}
	D_\alpha(\mathcal{M}(d)  \| \mathcal{M}(d')) \leq \epsilon_{R, \alpha}
\end{align}
\end{definition}

\begin{definition}(Group R\'enyi DP \cite{mironov2017renyi})\label{def:group_rdp}
For mechanism $\mathcal{M}$ that satisfies $(\alpha,\epsilon_{R,\alpha})$-RDP, it satisfies $(\alpha/2^k, 3^k\epsilon_{R,\alpha})$-DP for groups of size $k$. That is, for any  $ d,d' \in \mathcal{D}$ that differ by $k$ individuals, it holds that
\begin{align}
D_{\alpha/2^k}(\mathcal{M}(d)  \| \mathcal{M}(d')) \leq   3^k\epsilon_{R,\alpha}
\end{align}
\end{definition}

\begin{lemma}(R\'enyi DP and DP conversion \cite{mironov2017renyi})\label{lemma:renyidp_dp}
The mechanism $\mathcal{M}$ that satisfies  $(\alpha,  \epsilon_{R, \alpha})$-RDP $\alpha > 1$, also satisfies 
  $(\epsilon_{R, \alpha} + \frac{\log 1/\delta}{\alpha-1}, \delta)$-DP for any $0 < \delta < 1$.
\end{lemma}

\begin{lemma}(Certificates for R\'enyi-divergence  [Table 4 of \cite{Dvijotham2020A}])
\label{lemma:renyi_f_div_rand_smoothing}
Given two distributions $\rho$ and $\nu$ with bounded Rényi divergences $(\alpha \geq 0)$
$
D_\alpha(\rho \| \nu) \leq \epsilon_{R, \alpha},
$
and two probabilities $p_a$, $p_b$ that satify $p_a, p_b\geq0$, $p_a + p_b \leq 1$, and define the class of specification $S$ as 
\begin{small}
\begin{equation}
S=\left\{\phi: \mathcal{X} \rightarrow\{-1,0,+1\} \text { s.t. } \underset{x \sim \rho}{\mathbb{P}}[\phi(x)=+1] \geq p_a, \underset{x \sim \rho}{\mathbb{P}}[\phi(x)=-1] \leq p_b\right\}.    
\end{equation}
\end{small}
It is certified that $\mathbb{E}_{x \sim \nu}[\phi(x)] \geq 0$ for all $\nu$ and $\phi \in S$ if 
\begin{align*}
& \epsilon_{R, \alpha} \leq-\log \left(1-p_a-p_b+2 \eta\right) , \text{with\quad} \eta=\left(\frac{p_a^{(1-\alpha)}+p_b^{(1-\alpha)}}{2}\right)^{\left(\frac{1}{1-\alpha}\right)}.
\end{align*}
\end{lemma}

\subsection{Main Results on RDP-based Certified Prediction}
We present our main results for certified robustness against FL poisoning attacks based on R\'enyi DP (RDP)~\cite{mironov2017renyi} and Randomized Smoothing via R\'enyi Divergence~\cite{Dvijotham2020A}. 
\cref{thm_renyidp_pred_consist_one_client} states the certification under one adversarial user and \cref{thm_renyidp_pred_consist_k_client} further extends the certification to $k$ adversarial suers.

\begin{theorem}[RDP-based Certified Prediction under One Adversarial User]
\label{thm_renyidp_pred_consist_one_client}
\sloppy
Suppose a randomized mechanism $\Mcal$ satisfies user-level $(\alpha,  \epsilon_{R, \alpha})$-RDP, which also satisfies user-level $(\epsilon_{R, \alpha} + \frac{\log 1/\delta}{\alpha-1}, \delta)$-DP,  where $\alpha>1$ and $0<\delta<1$. For two user sets $B$ and $B^\prime$ that differ by one user, let $D$ and $D'$ be the corresponding training datasets. 
Define the classifier as $h:(\theta , \mathbb{R}^d)  \rightarrow [C]$ with the finite set of labels $[C]$, and the randomly smoothed classifier $h_s$ as 
 $h_s(\mathcal{M}(D),x) :=\arg \max_{c \in [C]} {\mathbb{P}} [h(\mathcal{M}(D),x)=c]$.
For a test input $x$, suppose that
\begin{align*}
&  {\mathbb{P}} [h(\mathcal{M}(D),x)=\Aclass] \geq p_a \geq p_b  \geq  \arg \max_{c \in [C]:c\neq \Aclass}  {\mathbb{P}} [h(\mathcal{M}(D),x)=c].
\end{align*}
Then, it is guaranteed that 
 $  h_s(\mathcal{M}(D'),x) = h_s(\mathcal{M}(D),x) = \Aclass$ if:

\begin{align}
& \epsilon_{R, \alpha} \leq-\log \left(1-p_a-p_b+2 \left(\frac{p_a^{(1-\alpha)}+p_b^{(1-\alpha)}}{2}\right)^{\left(\frac{1}{1-\alpha}\right)} \right).  \nonumber
\end{align}
\end{theorem}

\begin{theorem}[RDP-based Certified Prediction under $k$ Adversarial User]
\label{thm_renyidp_pred_consist_k_client}
Using the same setting as in Theorem~\ref{thm_renyidp_pred_consist_one_client} but let two user sets $B$ and $B^\prime$ differ by $k$ users, and $D$ and $D'$ be the corresponding training datasets. 
Then, it is guaranteed that   $h_s(\mathcal{M}(D'),x) = h_s(\mathcal{M}(D),x) = \Aclass$ if:
\begin{align}
&  \epsilon_{R, \alpha}  \leq- \frac{1}{ 3^k} \log \left(1-p_a-p_b+2 \left(\frac{p_a^{(1-\alpha/2^k)}+p_b^{(1-\alpha/2^k)}}{2}\right)^{\left(\frac{1}{1-\alpha/2^k}\right)} \right).  \nonumber
\end{align}
\end{theorem}

\textit{Remark.} 
From Theorem~\ref{thm_renyidp_pred_consist_one_client} and Theorem~\ref{thm_renyidp_pred_consist_k_client}, we observe that 
\textbf{(1)}  RDP-based certifications are more complex than DP-based certifications due to the additional tunable parameter, the RDP order $\alpha$, and its foundational R\'enyi Divergence-based privacy definition.  
\textbf{(2)} Theorem~\ref{thm_renyidp_pred_consist_k_client} presents a more intricate RHS, making it challenging to derive a simple closed-form upper bound $\cerk$ for the certified number of attackers where $k<\cerk$, as seen in Theorem~\ref{thm_pred_consist_k_client}. Nevertheless, Theorem~\ref{thm_renyidp_pred_consist_k_client} can be utilized to perform a binary check for certified robustness by verifying if the current RDP privacy budget satisfies the inequality.
\textbf{(3)} Different from DP-based certifications in Theorem~\ref{thm_pred_consist_one_client}  and Theorem~\ref{thm_pred_consist_k_client} that are built upon the \textit{expected class confidence} $F_A$ and $F_B$, RDP-based certifications are built upon the \textit{probability of model prediction}, e.g.,  the probability of the model predicting a certain class ${\mathbb{P}} [h(\mathcal{M}(D),x)=\Aclass]$, where $h(\mathcal{M}(D),x)$ is the predicted class.
To compute RDP-based certifications in practice, one can also use Marto Carlo sampling to approximate  ${\mathbb{P}} [h(\mathcal{M}(D),x)=\Aclass]$.

\subsection{Proofs}

We now provide the proofs for Theorem~\ref{thm_renyidp_pred_consist_one_client} and  Theorem~\ref{thm_renyidp_pred_consist_k_client} below.

\begin{proof}  [Proof for Theorem~\ref{thm_renyidp_pred_consist_one_client}]
Recall that we define the classifer $h:(\theta , \mathbb{R}^d)  \rightarrow [C]$ with the finite set of labels $[C]$, and the randomly smoothed classifer $h_s$ as 
\begin{align} h_s(\mathcal{M}(D),x) :=\arg \max_{c \in [C]} {\mathbb{P}} [h(\mathcal{M}(D),x)=c],
\end{align}
where $x$ is a test sample, $\mathcal{M}(D)$ is the stochastic model trained from the randomized DP mechanism $\mathcal{M}$ on a training dataset $D$.

For a test input $x$, suppose that

\begin{align*}
&  {\mathbb{P}} [h(\mathcal{M}(D),x)=\Aclass] \geq p_a \geq p_b  \geq  \arg \max_{c \in [C]:c\neq \Aclass}  {\mathbb{P}} [h(\mathcal{M}(D),x)=c].
\end{align*}
Therefore, $\Aclass =h_s(\mathcal{M}(D),x)$.

Let $\Bclass =\arg \max_{c \in [C]:c\neq \Aclass}  {\mathbb{P}} [h(\mathcal{M}(D),x)=c]$. 
We define the specification $\phi_{\Aclass, \Bclass}$ as follows:
\begin{equation}
    \phi_{\Aclass, \Bclass}(\mathcal{M}(D))= \begin{cases}+1 & \text { if } h(\mathcal{M}(D), x)=\Aclass \\ -1 & \text { if } h(\mathcal{M}(D), x)=\Bclass \\ 0 & \text { otherwise }\end{cases}
\end{equation}

Based on the certificates for R\'enyi-divergence in Lemma~\ref{lemma:renyi_f_div_rand_smoothing} and  Definition~\ref{def:rdp} for Rényi DP, if 
\begin{align*}
& \epsilon_{R, \alpha} \leq-\log \left(1-p_a-p_b+2 \eta\right) , \text{with\quad} \eta=\left(\frac{p_a^{(1-\alpha)}+p_b^{(1-\alpha)}}{2}\right)^{\left(\frac{1}{1-\alpha}\right)}, 
\end{align*}
and if the mechanism $\mathcal{M}$ satisfies $(\alpha,\epsilon_{R,\alpha})$-RDP $(\alpha > 1)$
\begin{equation}
D_\alpha(\mathcal{M}(D)  \| \mathcal{M}(D')) \leq \epsilon_{R, \alpha},
\end{equation}
it is certified that 
 $\mathbb{E}[ \phi_{\Aclass, \Bclass}(\mathcal{M}(D'))] =   {\mathbb{P}} [h(\mathcal{M}(D'),x)=\Aclass] - {\mathbb{P}} [h(\mathcal{M}(D'),x)=\Bclass]  \geq 0$, that is, 
\begin{align*}
& {\mathbb{P}} [h(\mathcal{M}(D'),x)=\Aclass]    \geq   {\mathbb{P}} [h(\mathcal{M}(D'),x)=\Bclass].
\end{align*}
It further implies that 
\begin{align*}
h_s(\mathcal{M}(D'),x) =h_s(\mathcal{M}(D),x)=\Aclass.
\end{align*}
Finally, we can convert Rényi DP to DP by Lemma~\ref{lemma:renyidp_dp}
\end{proof}
\begin{proof}[Proof for Theorem~\ref{thm_renyidp_pred_consist_k_client}]
According to the group Rényi DP in Definition~\ref{def:group_rdp}, the mechanism $\mathcal{M}$ that satifies user-level $(\alpha,\epsilon_{R,\alpha})$-RDP also satifies user-level $(\alpha/2^k,  3^k\epsilon_{R, \alpha})$-RDP for two user sets $B$ and $B^\prime$  that differ by $k$ users. That is, 
 \begin{equation}
D_{\alpha/2^k}(\mathcal{M}(D)  \| \mathcal{M}(D')) \leq  3^k \epsilon_{R, \alpha}.
\end{equation}

For a test input $x$, suppose that
\begin{align*}
&  {\mathbb{P}} [h(\mathcal{M}(D),x)=\Aclass] \geq p_a \geq p_b  \geq  \arg \max_{c \in [C]:c\neq \Aclass}  {\mathbb{P}} [h(\mathcal{M}(D),x)=c].
\end{align*}
Then, according to Lemma~\ref{lemma:renyi_f_div_rand_smoothing} and following similar steps in the proofs of the Theorem~\ref{thm_renyidp_pred_consist_one_client},  if 
\begin{align*}
&  3^k \epsilon_{R, \alpha}  \leq-\log \left(1- p_a - p_b  +2 \eta\right), \\
&\text{with\quad} \eta=\left(\frac{p_a^{(1-\alpha/2^k)}+p_b^{(1-\alpha/2^k)}}{2}\right)^{\left(\frac{1}{1-\alpha/2^k}\right)},
\end{align*}
it is certified that 
\begin{align*}
h_s(\mathcal{M}(D'),x) =h_s(\mathcal{M}(D),x)=\Aclass.
\end{align*}
\end{proof}

\end{rev}

\newpage
\onecolumn

\end{document}